\documentclass[12pt,onecolumn,journal,letterpaper,romanappendices]{IEEEtran}
\usepackage{epsfig,rotating,setspace,latexsym,amsmath,epsf,amssymb,bm}
\usepackage{color}

\usepackage{amsmath}
\usepackage{rotating,latexsym,amsmath,amssymb,bm}
\usepackage{cite}
\usepackage{algpseudocode}
\usepackage{algorithmicx}
\usepackage{algorithm}
\usepackage{amsthm}
\usepackage{subfigure}
\usepackage{graphicx}

\usepackage{algpseudocode}
\usepackage{algorithm}
\usepackage{url}

\usepackage{wrapfig}
\usepackage{dsfont}
\usepackage{wasysym}
\usepackage{hyperref}
\usepackage{balance}

\newcommand{\myeq}[1]{\mathrel{\overset{\makebox[0pt]{\mbox{\normalfont\tiny\sffamily #1}}}{=}}}
\newcommand{\myleq}[1]{\mathrel{\overset{\makebox[0pt]{\mbox{\normalfont\tiny\sffamily #1}}}{\leq}}}

\newcommand*{\Scale}[2][4]{\scalebox{#1}{$#2$}}%
\newcommand{\n}{\noindent}

\theoremstyle{plain}

\newtheorem{lemma}{Lemma}
\newtheorem{Cor}{Corollary}
\newtheorem{prop}{Proposition}

\theoremstyle{definition}
\newtheorem{Def}{Definition}

\usepackage{thmtools}

\declaretheoremstyle[
  spaceabove=\topsep, spacebelow=\topsep,
  headfont=\normalfont\bfseries,
  notefont=\mdseries, notebraces={(}{)},
  bodyfont=\normalfont,
  postheadspace=1em,
  qed=\qedsymbol
]{mythmstyle}
\declaretheorem[style=mythmstyle]{example}

\declaretheoremstyle[
  spaceabove=\topsep, spacebelow=\topsep,
  headfont=\normalfont\bfseries,
  notefont=\mdseries, notebraces={(}{)},
  bodyfont=\normalfont,
  postheadspace=1em,
  qed=$\diamond$
]{mythmstyle}
\declaretheorem[style=mythmstyle]{remark}

\DeclareMathOperator*{\minn}{minimize}

\begin{document}

\allowdisplaybreaks

\sloppy

\title{Fog-Aided Wireless Networks for Content Delivery: Fundamental Latency Trade-Offs }

\author{Avik Sengupta, \IEEEmembership{Member,~IEEE}, Ravi Tandon, \IEEEmembership{Senior Member,~IEEE},\\ Osvaldo Simeone, \IEEEmembership{Fellow,~IEEE}\\
\thanks{A. Sengupta is with the Next Generation and Standards (NGS) Group at Intel Corporation, Santa Clara CA 95054, USA. Email: avik.sengupta@intel.com. R. Tandon is with the Department of Electrical and Computer Engineering, University of Arizona, Tucson, AZ $85721$ USA. Email: tandonr@email.arizona.edu. O. Simeone is with the Centre for Telecommunications Research at the Department of Informatics at King's College London, UK. Email: osvaldo.simeone@kcl.ac.uk. O. Simeone has received funding from the European Research Council (ERC) under the European Union’s Horizon $2020$ research and innovation programme (grant agreement No $725731$). The work of O. Simeone was also partially supported by the U.S. NSF through grant CCF-$1525629$.}
\thanks{This work was presented in part at the $50$th Annual Conference on Information Sciences and Systems (CISS), March $2016$, IEEE International Symposium on Information Theory (ISIT), July $2016$, IEEE International Workshop on Signal Processing Advances in Wireless Communications (SPAWC) $2016$ and IEEE Globecom $2016$}}

\maketitle
\thispagestyle{plain}
\pagestyle{plain}
\begin{abstract}
A fog-aided wireless network architecture is studied in which edge-nodes (ENs), such as base stations, are connected to a cloud processor via dedicated fronthaul links, while also being endowed with caches. Cloud processing enables the centralized implementation of cooperative transmission strategies at the ENs, albeit at the cost of an increased latency due to fronthaul transfer. In contrast, the proactive caching of popular content at the ENs allows for the low-latency delivery of the cached files, but with generally limited opportunities for cooperative transmission among the ENs. The interplay between cloud processing and edge caching is addressed from an information-theoretic viewpoint by investigating the fundamental limits of a high Signal-to-Noise-Ratio (SNR) metric, termed normalized delivery time (NDT), which captures the worst-case coding latency for delivering any requested content to the users. The NDT is defined under the assumptions of either serial or pipelined fronthaul-edge transmission, and is studied as a function of fronthaul and cache capacity constraints. Placement and delivery strategies across both fronthaul and wireless, or edge, segments are proposed with the aim of minimizing the NDT. Information-theoretic lower bounds on the NDT are also derived. Achievability arguments and lower bounds are leveraged to characterize the minimal NDT in a number of important special cases, including systems with no caching capabilities, as well as to prove that the proposed schemes achieve optimality within a constant multiplicative factor of $2$ for all values of the problem parameters.
\end{abstract}


\section{Introduction}\label{sec:intro}
Traffic over wireless networks is undergoing a significant transformation to become increasingly dominated by multimedia content, and particularly by video \cite{bastug}. As a result, caching of popular content during off-peak traffic periods at various levels of the wireless network architecture has emerged as a major technology trend for the next generation ($5$G) of wireless systems. Moving the location of the caches closer to the edge of the network has the advantage of reducing the latency required for accessing and delivering users' requests. In particular, caching at the edge nodes (ENs), such as at macro or small-cell base stations, allows the delivery of content to mobile users with limited need for backhaul usage to connect to a remote content server (see \cite{cache_5g} and references therein).

While potentially reducing delivery latency and backhaul load, edge caching generally limits the operation of ENs to non-cooperative transmission strategies. This is because, with edge caching, each EN can design its transmitted signal based only on its local cached content, which may only partially overlap with that of other ENs, hence preventing cooperative transmission schemes such as joint beamforming. The \textit{localized} processing afforded by edge caching is in contrast to the \textit{centralized} processing that is instead possible in network architectures in which the ENs are controlled by a \textit{cloud} processor. An important example of this class of networks is the Cloud Radio Access Network (C-RAN) architecture, in which the ENs are connected to a cloud processor by means of so called \emph{fronthaul} links. In a C-RAN, the signals transmitted by the ENs are produced at the cloud based on a direct connection to the content server and forwarded to the ENs on the fronthaul links. As such, cloud processing in C-RAN enables the implementation of cooperative transmission strategies across the ENs, but at the cost of a potentially large latency, owing to the time required for fronthaul transfer (see, e.g., \cite{fronthaul,SimeoneMPSY15}).
\begin{figure}[!t]
\centering
\includegraphics[width=3.8in,height=3.3in]{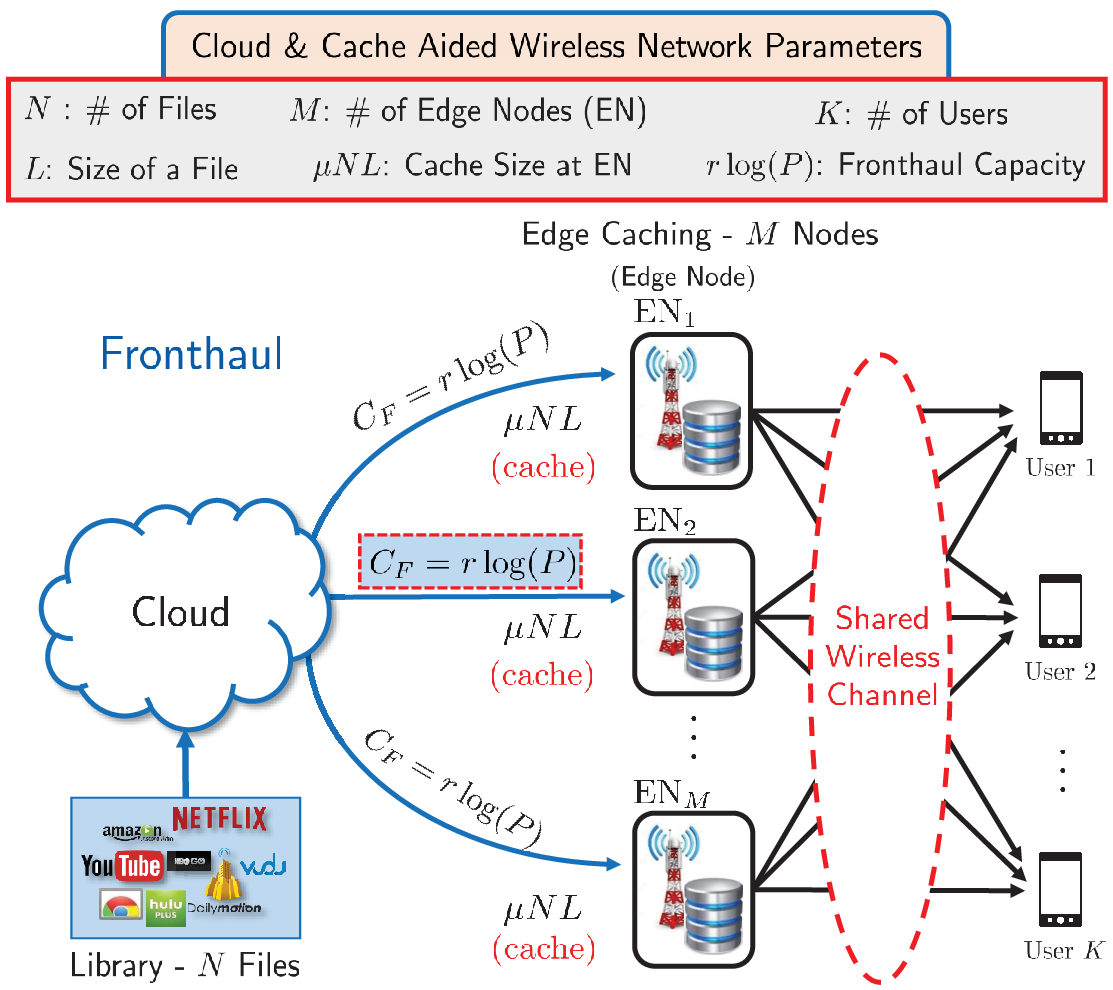}
\caption{Information-theoretic model for a cloud and cache-aided wireless system, referred to as Fog-Radio Access Network (F-RAN).}
\label{fig:sysmod}
\vspace{-12pt}
\end{figure}

Motivated by the complementary benefits highlighted above between cloud-based and edge-based architectures, in this work we consider a \textit{cloud and cache-aided} wireless network architecture, which we term \textit{Fog Radio Access Network} (F-RAN). The F-RAN architecture reflects major trends in the evolution towards 5G systems, which are envisaged to represent a profound paradigm shift with respect to previous generations at both the \textit{architectural} and \textit{functional} levels. At the architectural level, the trend is towards a \textit{fog architecture}, which encompasses cloud and edge segments \cite{fog_commag}. The cloud segment contains network servers, while the edge segment includes ENs. The two segments are connected by a transport, or \textit{fronthaul}, network. The key novelty of the fog architecture is the availability of computing and storage resources at both cloud servers and ENs, which can be programmed and reconfigured via software. At the functional level, the key technology of Network Function Virtualization (NFV) enables the flexible and adaptive allocation of the network functions that make up a network service, such as computing tasks and caching, on the software controlled network elements in the cloud and edge segments \cite{NFV_commag}. The reconfigurability afforded by the fog architecture via NFV contrasts with the rigid allocation of network functions in current cellular systems. 


The proposed system model accounts for a fog architecture in which the network functionalities that contribute to content delivery can be flexibly allocated across cloud and edge segments. In an F-RAN, as seen in Fig. \ref{fig:sysmod}, the ENs are connected to a cloud processor via dedicated fronthaul links, while also being endowed with caches that can be used to proactively store popular content \cite{Peng_fog}. The design of F-RAN systems involves two key design questions: (\emph{i}) \textit{What to cache at the ENs?}; and (\emph{ii}) \textit{How to deliver the requested content across the fronthaul and wireless, or edge, segments?} The two questions pertain to network functions, namely caching and delivery, that operate at different time scales: while caches are updated only at the time scale over which popular content is expected to change, e.g., every night, delivery is performed in each transmission interval in order to satisfy the current users' requests from the content library. Nevertheless, the two questions are strongly intertwined since delivery strategies need to operate by leveraging the existing cached content, as well as cloud processing.

In order to address the design of F-RAN, in this work, we adopt as a performance metric the \emph{worst-case coding based delivery latency} accrued when serving any set of users' requests in a given transmission interval. We aim to characterize optimal caching and delivery strategies that minimize the delivery latency. To enable analytical insights, we specifically propose a latency metric, termed \textit{Normalized Delivery Time} (NDT), which captures the high signal-to-noise-ratio (SNR) ratio of the latency achievable in an F-RAN, with given fronthaul and caching limitations, as compared to that of a \emph{reference} system with unlimited caching capability and interference-free links to the users. We first focus on delivery strategies in which fronthaul and wireless channels are operated in a \textit{serial} manner, so that the overall latency is the sum of the time spent for fronthaul communication between cloud and ENs and of the time required for wireless transmission from ENs to users. We then extend the analysis to characterize the NDT for systems using delivery strategies in which fronthaul and wireless channels are operated in a \textit{pipelined} (parallel) manner, so that fronthaul and wireless transmissions can take place at the same time (see, e.g., \cite{Leconte}).

\begin{figure}[t]
\centering
\subfigure[]{
\includegraphics[width=3.3in,height=2.425in]{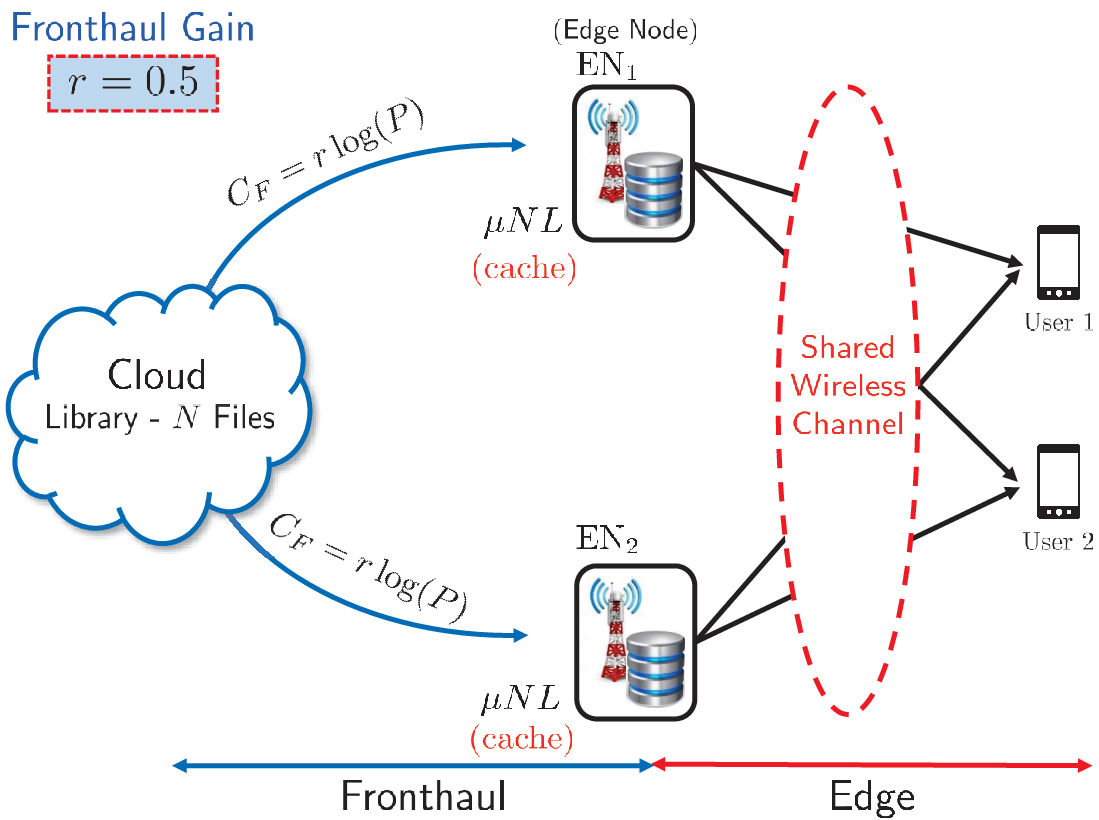}
\label{fig:ITmodel}
}
\subfigure[]{
\includegraphics[width=3.0in,height=2.55in]{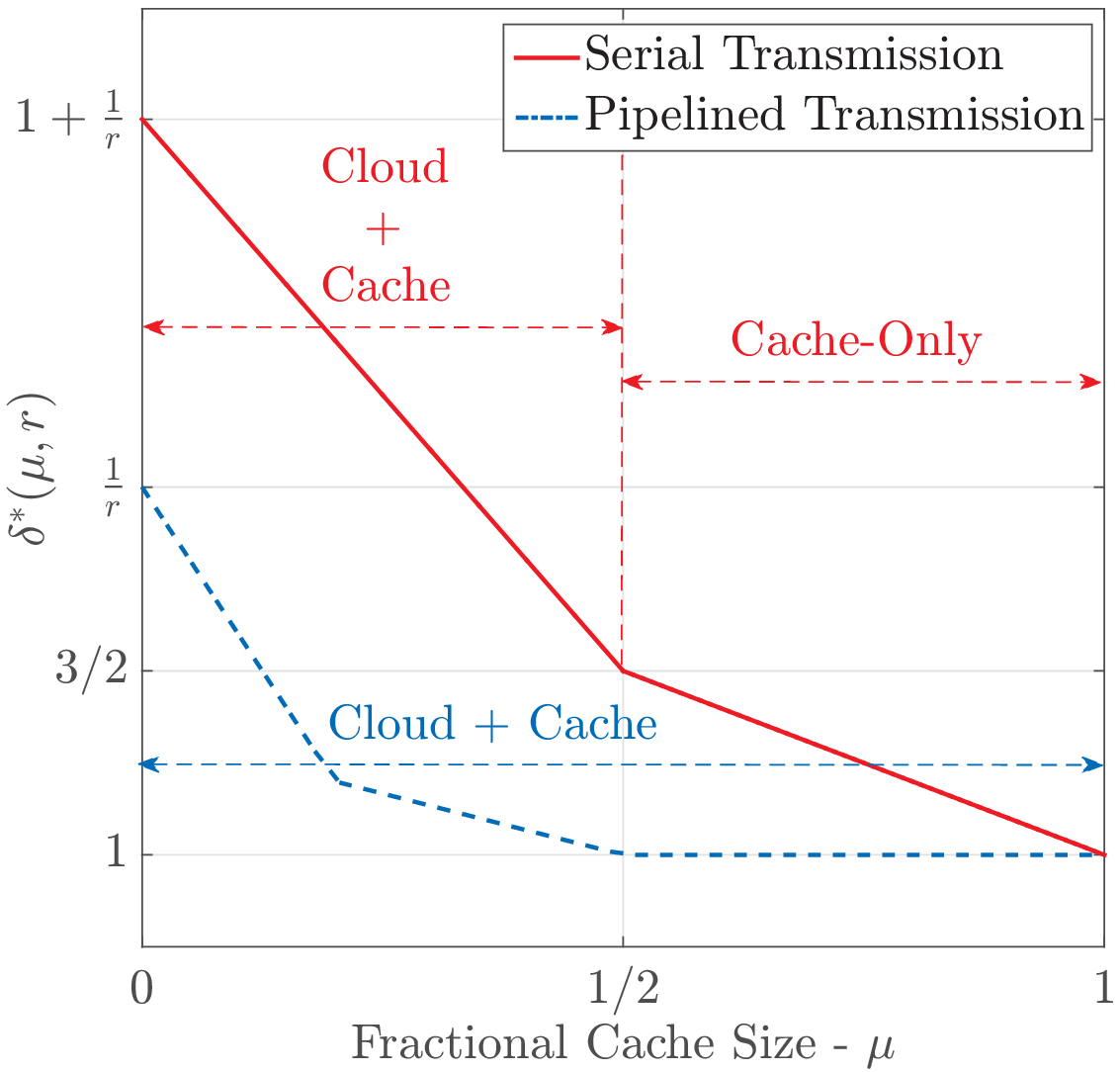}
\label{fig:MK22}
}
\caption{{(a)} Information-theoretic model for an F-RAN with $M=2$ ENs serving $K=2$ users and a fronthaul rate $r = 0.5$; {(b)} Trade-off between the normalized delivery time (NDT) and the fractional cache size $\mu$ in the presence of full CSI at ENs, users and the cloud.} \label{fig:MK22_ex}\vspace{-10pt}
\end{figure}
\begin{example}
To exemplify the analysis put forth in this paper, we briefly illustrate here the F-RAN set-up of Fig. \ref{fig:ITmodel}, in which two ENs (labeled as $\textrm{EN}_1$ and $\textrm{EN}_2$) are deployed to serve two users over a shared wireless channel. The ENs are connected to the cloud via fronthaul links whose capacity scales with the SNR $P$ of the wireless edge links as $r\log(P)$, with $r\geq 0$ being defined as the \textit{fronthaul rate}. We assume that there is a library of $N\geq 2$ popular files, each of a given size, and that each EN can cache at most a fraction $\mu \in [0,1]$ of the library content, where $\mu$ is defined as the \emph{fractional cache size}. Full Channel State Information (CSI) of the shared wireless channel is assumed at all nodes. For this example, the \textit{information-theoretically optimal} trade-off $\delta^{*}(\mu,r)$ between the NDT and the fractional cache size $\mu$ is shown in Fig. \ref{fig:MK22} for $r=0.5$ for serial as well as for pipelined fronthaul-edge transmission. The NDT captures the worst-case delivery latency required by the F-RAN to deliver all files requested by the users across both fronthaul and wireless segments. An NDT $\delta^{*}(\mu,r)$ indicates that an F-RAN with fractional cache size $\mu$ and fronthaul rate $r$ requires a total latency that is $\delta^{*}(\mu,r)$ times the time required by the mentioned reference system with unlimited caching and no interference. 

Among other conclusions, as illustrated in Fig. \ref{fig:MK22_ex}, the analysis presented in this paper reveals that, for serial fronthaul-edge transmission, in the regime of low fronthaul capacity ($r \leq 0.5$), the latency due to fronthaul transfer makes cloud processing not useful in reducing the overall delivery latency when the cache capacity is large enough $(\mu\geq 1/2)$. In contrast, for pipelined fronthaul-edge transmission, cloud processing is instrumental in obtaining the minimum delivery latency for all values of $\mu$, even when the fronthaul capacity is small. This is because, with pipelined transmission, the ENs need not wait for the fronthaul transmission to be completed before communicating to the users on the edge links. For the same reason, pipelined fronthaul-edge transmission generally improves the NDT compared to serial transmission. In particular, even with partial caching, that is, with $\mu<1$, the ideal NDT $\delta^*=1$ is achievable with pipelined fronthaul-edge transmission, while this is not the case with serial transmission. More details can be found in Sections \ref{ssec:ndtapprox} and \ref{ssec:nk22_pipe}. 
\end{example}

\n \textbf{Related Work:} The line of work pertaining to the information-theoretic analysis of cache-aided communication systems can be broadly classified into studies that consider caching at the end-users' devices or at the ENs. This research direction was initiated by \cite{Maddah-Ali,Maddah-Ali-decentralized} for a set-up that consists of a multicast link with cache-aided receivers. This work demonstrates that coded multicasting enables global caching gains to be reaped, as opposed to the conventional local caching gains of uncoded transmission. Follow-up papers on related models with receiver-end caching include \cite{Maddah-Ali-nonuniform, Maddah-Ali-online,ISWCS_Ji, shanmugam_finite,aviksg-isit15,ghasemi,tifr,sahraei_multi-library,Motahari_multi-server_arxiv,wan_moreusers,Gastpar_newconv,tuninetti_optimal,aviksg-tifs,fund_ji}. The present paper is instead inscribed in the parallel line of work that concerns caching at the ENs of a wireless network. A pioneering effort on this subject is \cite{Femto-journal}, in which ``femto-caching'', that is, caching at small-cell base stations, is introduced as a means to reduce backhaul usage and delivery latency. This and follow-up papers, including \cite{Gunduz_ML,Gunduz_ISIT,aviksg-iswcs,mds-cache,PES_journal,diggavi_hcc,diggavi_hetnet,diggavi_mlc}, assume that cache-aided ENs are not allowed to cooperate on the basis of the cached content to mitigate or cancel mutual interference. In contrast, references \cite{MA-CAIC,MA-CAIC-arxiv} investigate, from an information-theoretic viewpoint, an interference-limited wireless system with cache-aided ENs that can carry out coordinated transmission strategies, such as interference alignment, as well as cooperative transmission schemes, such as joint beamforming. Specifically, in \cite{MA-CAIC,MA-CAIC-arxiv} an upper bound on the worst-case delivery latency, which is formulated in terms of the inverse of the degrees-of-freedom metric, is derived for $M=3$ ENs and $K=3$ users by proposing a specific caching and delivery policy. Upper and lower bounds on this metric are obtained in \cite{avestimehr-caim} by accounting for caching at both ENs and users, under the assumption of delivery strategies based on linear precoding on the wireless channel. Related works that focus on the optimization of signal processing strategies at the ENs can be found in \cite{letaief,yu_cran,Azari15,osv_fog,Sezgin_cran}. This work was partially presented in \cite{aviksg-ciss,aviksg-ciss-arxiv,aviksg-spawc,RT_ISIT16,aviksg-gc16} and a brief informal summary was provided in \cite{osv_com_mag}. Reference \cite{meixia_fdt} extends the lower bounds in \cite{aviksg-ciss} to the case with caching also at the receivers and proposes upper bounds on delivery latency for systems with two and three receivers. \vspace{2.5pt}

\n \textbf{Main Contributions:} The main contributions of this paper are summarized as follows.

\begin{itemize}
\item An information-theoretic model of a fog-aided system, termed F-RAN, is presented, along with a novel latency metric, namely the \textit{normalized delivery time} (NDT). The NDT measures the worst-case latency required to deliver an arbitrary vector of requests to the users in the high-SNR regime, as compared to a reference system with full caching and no interference on the wireless channel. 

\item Under the assumptions of uncoded inter-file caching (but allowing for arbitrary intra-file coding), full CSI at all nodes and serial transmission, we develop general information-theoretic lower bounds on the minimum NDT for an F-RAN with any number of ENs and users as a function of the caching and fronthaul limitations as defined by the parameters $\mu$ and $r$, respectively. The lower bounds are derived by adopting cut-set arguments that are tailored to the set-up at hand that includes both fronthaul and wireless segments.
\item We present a general upper bound on the NDT of an arbitrary F-RAN with serial transmission by leveraging file-splitting between cloud-aided and cache-aided transmission strategies. For the cloud-aided scheme, we consider a novel \textit{soft-transfer} fronthauling approach, inspired by the standard operation of C-RAN \cite{SimeoneMPSY15}, which is based on the transmission of quantized encoded signals on the fronthaul links. For cache-aided strategies, we leverage both coordination via interference alignment and cooperation via joint beamforming at the ENs based on cached content. A number of alternative strategies are also considered for reference, including the conventional \textit{hard-transfer} of uncached content on the fronthaul links.
\item The proposed achievable schemes are shown to achieve the minimum NDT to within a factor of $2$ for all values of the system parameters for serial transmission.
\item The minimum NDT is characterized exactly in a number of important special cases for serial transmission. These include: cloud-only F-RANs, also known as C-RAN; cache-only F-RANs, that is the cache-aided wireless system studied in \cite{MA-CAIC,MA-CAIC-arxiv} for extremal values of fractional cache size $\mu$; and general F-RAN models with both cloud processing and caching for the case when the number of users exceeds the number of ENs in the low fronthaul regime. We present a specific case study for the $2\times 2$ F-RAN, where the minimum NDT is completely characterized by the proposed bounds (see Fig. \ref{fig:MK22}).
\item We define and investigate an F-RAN model in which the fronthaul and wireless edge segments can be operated in a pipelined, or parallel, manner. We show that, in comparison to serial transmission, pipelined fronthaul-edge transmission can improve the NDT by a multiplicative factor of at most $2$.
\item We present a general lower bound on the minimum NDT for the the pipelined fronthaul-edge transmission model as well as achievable schemes which leverage block-Markov encoding along with file-splitting between cloud and cache-aided transmission strategies. 
\item We characterize the minimum NDT for cloud-only F-RAN with pipelined fronthaul-edge transmission. Furthermore, for a general $M\times K$ F-RAN with pipelined fronthaul-edge transmission, the proposed schemes are shown to achieve the minimum NDT to within a factor of $2$ for all values of system parameters. We present the case study for the $2\times 2$ F-RAN for which the minimum NDT is completely characterized by the proposed bounds (see Fig. \ref{fig:MK22}).
\end{itemize}
The remainder of the paper is organized as follows. Section \ref{sec:sysmod} presents the information-theoretic model for a general $M\times K$ F-RAN and introduces the NDT metric for serial fronthaul-edge transmission. Lower bounds on the NDT for an F-RAN are derived in Section \ref{sec:lb}, while achievable schemes are proposed in Section \ref{sec:ub}. In Section \ref{sec:opt}, we present the mentioned finite-gap and exact characterization of the minimum NDT. 
Section \ref{sec:pipe} discusses the F-RAN model with pipelined fronthaul-edge transmissions. General upper and lower bounds on the minimum NDT for this model are presented along with a finite-gap characterization of the minimum NDT. Section \ref{sec:open} highlights some of the open problems and directions for future work, while Section \ref{sec:conc} concludes the paper. \vspace{5pt}

\n \textbf{Notation:} For any two integers $a$ and $b$ with $a\leq b$, we define the notation $[a:b] \triangleq \{a,a+1,\ldots,b\}$. We also use the notation $b\in[a,c]$ to imply that $b$ lies in the interval $a\leq b \leq c$ for any $a,b,c$. Furthermore, $b\in (a,c]$ denotes $a<b\leq c$. We use the notation $x\in\{a,b,\ldots,c\}$ to denote that the variable $x$ takes the values in the set $\{a,b,\ldots,c\}$. We define the function $(x)^+ \triangleq \max\{0,x\}$. The set of all positive integers is denoted by $\mathbb{N}^+$ and the set of all complex numbers is denoted by $\mathbb{C}$.

\section{System Model and Performance Metrics}\label{sec:sysmod}
In this section, we first present a model for the cloud and cache aided F-RAN system under study. Then, we introduce the normalized delivery time (NDT) metric, along with a number of remarks to provide additional context on the adopted model and performance metric.

\subsection{System Model}\label{ssec:sysmod}
We consider an $M\times K$ F-RAN, shown in Fig. \ref{fig:sysmod}, where $M$ ENs serve a total of $K$ users through a shared wireless channel. The ENs can cache content from a library of $N$ files, $F_1,\ldots,F_N$, where each file is of size $L$ bits, for some $L\in \mathbb{N}^+$. Formally, the files $F_n$ are independent and identically distributed (i.i.d.) as
\begin{align}
F_n \sim \text{Unif}\left\{1,2,\ldots, 2^{L}\right\}, ~~ \forall n \in [1:N].
\end{align}
Each EN is equipped with a cache in which it can store $\mu NL$ bits, where the fraction $\mu$, with $\mu\in[0,1]$, is referred to as the \textit{fractional cache size} and can be interpreted as the fraction of each file which can be cached at an EN. The cloud has full access to the library of $N$ files, and each EN is connected to the cloud by a fronthaul link of capacity of $C_F$ bits per symbol, where a symbol refers to a channel use of the downlink wireless channel.

In a transmission interval, each user $k\in[1:K]$ requests one of the $N$ files from the library. The demand vector is denoted by $\mathbf{D}\triangleq \left(d_1,\ldots,d_K\right) \in [1:N]^K$. This vector is known at the beginning of a transmission interval by both cloud and ENs, which attempt to satisfy the users' demands within the lowest possible delivery latency. As illustrated in Fig. \ref{fig:sysmod}, we first assume a serial operation over the fronthaul and wireless segments, whereby the cloud first communicates to the ENs and then the ENs transmit on the shared wireless channel to the users. As a result, the total latency is the sum of fronthaul and edge latencies (see Remark \ref{rem:pipe} for additional discussion on this point).

All the nodes have access to the global CSI about the wireless channels $\mathbf{H} = \left\{\left\{h_{km}\right\}:{k=[1:K], m=[1:M]}\right\}$, where $h_{km}\in \mathbb{C}$, denotes the wireless channel between user $k\in [1:K]$ and $\text{EN}_m$, $m\in[1:M]$. The coefficients are assumed to be drawn i.i.d. from a continuous distribution and to be time-invariant within each transmission interval.

As mentioned, the design of the system entails the definition of caching and delivery policies, which are formalized next for the case of serial fronthaul-edge transmission. Various generalizations of the definition below are presented in Sections \ref{sec:lb}, \ref{sec:ub} and \ref{sec:pipe}.

\begin{Def}[\textit{Policy}]\label{def:pol}
A caching, fronthaul, edge transmission, and decoding policy $\pi = (\pi_c, \pi_f, \pi_e,\pi_d)$ is characterized by the following functions.

\noindent\textit{a)~{Caching Policy $\pi_{c}$}:} The caching policy at each edge node $\textrm{EN}_m, ~m\in[1:M]$, is defined by functions $\pi_{c,n}^{m}(\cdot)$ that maps each file $F_n$ to its cached content $S_{m,n}$ as
\begin{align}\label{eq:cache}
S_{m,n} \triangleq \pi_{c,n}^{m}\left(F_n\right), ~~ \forall n\in[1:N].
\end{align}
The mapping is such that $H(S_{m,n})\leq \mu L$ in order to satisfy the cache capacity constraints. The overall cache content at $\textrm{EN}_m$ is given by $S_m = \left(S_{m,1},S_{m,2},\ldots,S_{m,N}\right)$. Note that the caching policy $\pi_{c,n}$ allows for arbitrary coding within each file, but it does not allow for inter-file coding.
Furthermore, the caching policy is kept fixed over multiple transmission intervals and is thus agnostic to the demand vector $\mathbf{D}$ and the global CSI $\mathbf{H}$.

\noindent {\textit{b)~Fronthaul Policy} $\pi_f$}: A fronthaul policy is defined by a function $\pi_f (\cdot)$, which maps the set of files $F_{[1:N]}$, the demand vector $\mathbf{D}$ and CSI $\mathbf{H}$ to the fronthaul message
\begin{align}\label{eq:fh}
\mathbf{U}_m^{T_F} = \left\{U_m[t]\right\}_{t=1}^{T_F} = \pi_f^m\left(\{F_{[1:N]}\},S_m,\mathbf{D},\mathbf{H}\right),
\end{align}
which is transmitted to $\text{EN}_m$ via the fronthaul link of capacity $C_F$ bits per symbol. Here, $T_F$ is the duration of the fronthaul message. In keeping with the definition of fronthaul capacity $C_F$, all time intervals, including $T_F$, are normalized to the symbol transmission time on the downlink wireless channel. Thus, the fronthaul message cannot exceed $T_FC_F$ bits.

\noindent {\textit{c)~Edge Transmission Policy} $\pi_e$}: After fronthaul transmission, each edge node $\textrm{EN}_m$ follows an edge transmission policy $\pi^m_e(\cdot)$ to map the demand vector $\mathbf{D}$ and global CSI $\mathbf{H}$, along with its local cache content and the  received fronthaul message, to output a codeword
\begin{align}
\mathbf{X}_m^{T_E} \hspace{-2pt}= \left\{X_m[t]\right\}_{t=1}^{T_E} \hspace{-2pt}= \pi^m_e\Big(S_m,\mathbf{U}_m^{T_F},\mathbf{D},\mathbf{H}\Big),
\end{align}
which is transmitted to the users on the shared wireless link. Here, $T_E$ is the duration of the transmission on the wireless channel, on which an average power constraint of $P$ is imposed for each codeword $\mathbf{X}_m^{T_E}$.
Note that the fronthaul policy $\pi_f$ and the edge transmission policy $\pi_e$ can adapt to the instantaneous demands and CSI at each transmission interval, unlike the caching policy, $\pi_c$, which remains unchanged over multiple transmission intervals.

\noindent {\textit{d)~Decoding Policy} $\pi_d$}: Each user $k \in [1:K]$, receives a channel output given by
\begin{align}\label{eq:chop}
\mathbf{Y}_k^{T_E} \hspace{-5pt}=\left\{Y_k[t]\right\}_{t=1}^{T_E} \hspace{-5pt}= \sum_{m=1}^M h_{km}\mathbf{X}_m^{T_E} + \mathbf{n}_k^{T_E},
\end{align}
where the noise $\mathbf{n}_k^{T_E} = \{n_k[t]\}_{t=1}^{T_E}$ is such that $n_k[t] \sim \mathcal{CN}(0,1)$ is i.i.d. across time and users. Each user $k \in [1:K]$, implements a decoding policy $\pi_d(\cdot)$, which maps the channel outputs, the receiver demands and the channel realization to the estimate
\begin{align}
\widehat{F}_{d_k} \triangleq \pi^k_d\Big(\mathbf{Y}_k^{T_E},d_k,\mathbf{H}\Big)
\end{align}
of the requested file ${F}_{d_k}$. The caching, fronthaul, edge transmission and decoding policies together form the policy $\pi = (\pi_c^{m},\pi_f^{m},\pi_e^m,\pi_d^k)$ that defines the operation of the F-RAN system. The probability of error of a policy $\pi$ is defined as
\begin{align}
P_e = \max_{\mathbf{D}}\max_{k\in [1:K]} \mathbb{P}\left(\widehat{F}_{d_k} \neq {F}_{d_k}\right),
\end{align}
which is the worst-case probability of decoding error measured over all possible demand vectors $\mathbf{D}$ and over all users $k\in[1:K]$. A sequence of policies, indexed by the file size $L$, is said to be \emph{feasible} if, for almost all channel realizations $\mathbf{H}$, i.e., with probability $1$, we have $P_e \rightarrow 0$ when $L\rightarrow \infty$.
\end{Def}

\subsection{Performance Metric: Normalized Delivery Time}\label{ssec:ndt}
We next define the proposed performance metric of normalized delivery time (NDT) by first introducing the notion of delivery time per bit.
\begin{Def}[\textit{Delivery time per bit}]\label{def:delta}
A \textit{delivery time per bit} $\Delta(\mu,C_F,P)$ is achievable if there exists a sequence of feasible policies such that
\begin{align}\label{eq:delta}
\Delta(\mu,C_F,P) &= \limsup_{L\rightarrow\infty} \frac{T_F + T_E}{L}.
\end{align}
\end{Def}
The delivery time per bit accounts for the latency within each transmission interval. Specifically, the total latency is given by the sum of the fronthaul and edge contributions, namely $T_F$ and $T_E$, respectively. In order to obtain a vanishing probability of error, as required by Definitions \ref{def:pol} and \ref{def:delta}, the latencies $T_F$ and $T_E$ need to scale with $L$, and it is this scaling that is measured by \eqref{eq:delta}. We also observe that the definition of delivery time per bit in \eqref{eq:delta} is akin to the completion time studied in \cite{erkip_comptime,erkip_comptime2} for standard channel models such as broadcast and multiple access channels.

 While the delivery time per bit $\Delta(\mu,C_F,P)$ generally depends on the power level $P$, as well as on the fronthaul capacity $C_F$ and fractional cache size $\mu$, we next define a more tractable metric that reflects the latency performance in the high-SNR regime. To this end, we let the fronthaul capacity scale with the SNR parameter $P$ as $C_F = r \log(P)$, where the fronthaul rate $r$ measures high-SNR ratio between the fronthaul capacity and the capacity of each EN-to-user wireless link in the absence of interference.
\vspace{-3pt}
\begin{Def}[\textit{NDT}]\label{def:ndt}
For any achievable $\Delta(\mu,C_F,P)$, with $C_F = r\log(P)$, the \textit{normalized delivery time} (NDT), is defined as
\begin{align} \label{eq:ndt1}
\delta(\mu,r) = \lim_{\substack{P\rightarrow \infty}}\frac{\Delta(\mu,r\log(P),P)}{1/\log P}.
\end{align}
Moreover, for any given pair $(\mu,r)$, the minimum NDT is defined as
\begin{align}\label{eq:ndt}
\delta^*(\mu,r) = \inf \left\{\delta(\mu,r):\delta(\mu,r) ~\text{is achievable} \right\}.
\end{align}
\end{Def}

\begin{remark}[\textit{Operational significance of NDT}]\label{rem:opsig}
In \eqref{eq:ndt1}, the delivery time per bit \eqref{eq:delta} is normalized by the term $1/\log P$. The latter is the delivery time per bit in the high-SNR regime for a reference baseline system with no interference and unlimited caching, in which each user can be served by a dedicated EN which has locally stored all the files. An NDT of $\delta^*$ hence indicates that the worst-case time required to serve any possible request vector $\mathbf{D}$ is $\delta^*$ times larger than the time needed by this reference baseline system. Note that, as a result, the NDT \eqref{eq:ndt} is always greater than or equal to one. We also remark that the mentioned baseline system can also be used as a reference to study the performance of an F-RAN model that enables also caching at the receivers, as proposed in \cite{meixia_fdt}. In this case, however, the NDT can be less than one, since the presence of caching at the receiver can yield lower latencies as compared to the reference system.
\end{remark}

\begin{remark}[\textit{Cache-Only F-RAN and Cloud-Only F-RAN}]\label{rem:nofront}
Throughout this paper, we will often consider separately the two important special cases of cache-only F-RAN and cloud-only F-RAN. The former corresponds to the case in which the fronthaul capacity is zero, i.e., $r=0$, while the latter, which amounts to a C-RAN system (see Section \ref{sec:intro}), is obtained by setting $\mu=0$. We observe that, in a cache-only F-RAN, as studied in \cite{MA-CAIC,MA-CAIC-arxiv}, it is required that the collective cache size of the $M$ ENs be large enough to completely store the entire library of $N$ files in order to obtained a finite worst-case delivery latency. This requires the condition $M\times \mu NL \geq N L $, i.e., $ \mu \geq {1}/{M}$, holds. Therefore, for this case, it suffices to focus on the range $\mu\in[1/M,1]$ of fractional cache capacity.
\end{remark}

\begin{remark}[\textit{NDT vs. DoF}]\label{rem:dof}
For the specific case of a cache-only F-RAN, the NDT in \eqref{eq:ndt} is proportional to the inverse of the more conventional degrees of freedom (DoF) metric $\text{DoF}(\mu)$ defined in \cite{MA-CAIC,MA-CAIC-arxiv}. Specifically, we have the relationship $\delta^{*}(\mu,0)=K/\text{DoF}(\mu)$.
\end{remark}

We show next that the NDT is convex in the fractional size $\mu$ for any value of the fronthaul rate $r\geq 0$. The proof follows from a \textit{file-splitting and cache-sharing} argument, whereby files are split into two fractions, with the two fractions being served by different policies that share the cache resources and whose delivery times add up to yield the overall NDT.


\begin{lemma}[\textit{Convexity of Minimum NDT}]\label{lem:conv}
 The minimum NDT, $\delta^{*}(\mu,r)$, is a convex function of $\mu$ for every value of $r\geq0$.
\end{lemma}
\begin{proof}
Consider any two feasible policies $\pi_1$ and $\pi_2$, where policy $\pi_i$ requires a fractional cache capacity and fronthaul rate pair $(\mu_i,r)$ and achieves an NDT of  $\delta(\mu_i,r)$ for $i=1,2$. Given an F-RAN system with cache storage capacity $\mu = \alpha \mu_1 + (1-\alpha)\mu_2$ and fronthaul rate $r$ for some $\alpha \in [0,1]$, we consider the following policy. Each file is split into two parts of sizes $\alpha L$ and $(1-\alpha) L$, respectively, where the first is delivered by using policy $\pi_1$ and the second by using policy $\pi_2$. Note that a fractional cache capacity $\mu$ is sufficient to support the operation of this policy. The NDT achieved by this policy can be computed as $\delta(\mu,r) =  \alpha\delta(\mu_1,r) +(1-\alpha)\delta(\mu_2,r)$ since, by \eqref{eq:delta} and \eqref{eq:ndt1}, the NDT is proportional to the file size. Applying this argument to two policies that achieve minimum NDTs $\delta^*(\mu_i,r)$ for $i=1,2$ proves the inequality
\begin{align}\label{eq:lemconv}
&\delta^{*}\left(\alpha \mu_1 + (1-\alpha)\mu_2,r \right) \leq  \alpha\delta^{*}(\mu_1,r) +(1-\alpha)\delta^{*}(\mu_2,r),
\end{align}
since the right-hand side of \eqref{eq:lemconv} is achievable by file-splitting. This shows the convexity of the minimum NDT  $\delta^{*}(\mu,r)$ as a function of $\mu$ for every value of $r\geq 0$.
\end{proof}

\begin{remark}[\textit{Pipelined Fronthaul-Edge Transmission}]\label{rem:pipe}
As discussed, the system model presented in this section adopts a serial delivery model, whereby fronthaul transmission is followed by edge transmission. An alternative model is of pipelined delivery in which the ENs can simultaneously receive on fronthaul links and transmit on the wireless channel. In this case, each edge node $\text{EN}_m$ starts transmitting at the beginning of the transmission interval using an edge transmission policy $\pi^{m}_{\mathsf{P},e}(\cdot)$, such that, at any time instant $t$, the EN maps the demand vector $\mathbf{D}$, the global CSI $\mathbf{H}$, the local cache content $S_m$ and the fronthaul messages received up to time $(t-1)$, to the transmitted signal at time $t$ as
\begin{align}\label{eq:tx_pipe}
X_m[t] = \pi^{m}_{\mathsf{P},e}\Big(S_m,\left\{U_m[1], U_m[2], \ldots, U_m[t-1]\right\},\mathbf{D},\mathbf{H}\Big), ~~t\in[1:T]
\end{align}
The overall latency is given by $T$ and the NDT can be defined in a manner analogous to Definition \ref{def:ndt}, namely
\begin{align}\label{eq:ndt_pipe}
\delta_{\mathsf{P}}(\mu,r) = \lim_{P\rightarrow \infty} \limsup_{L\rightarrow \infty} \frac{T}{L/\log P}.
\end{align}
We observe that the serial fronthaul-edge transmission policies in Definition \ref{def:pol} are included as special cases in the class of pipelined fronthaul-edge transmission schemes. As a result, the minimum NDT $\delta^*_{\mathsf{P}}(\mu,r)$ under pipelined operation can be no larger than that under serial operation. Furthermore, following the same arguments as in Lemma \ref{lem:conv}, the minimum NDT $\delta^*_{\mathsf{P}}(\mu,r)$ can be seen to be a convex function of $\mu$ for any $r\geq 0$. We provide a detailed study of the pipelined delivery model in Section \ref{sec:pipe}.
\end{remark}

\section{Lower Bound on minimum NDT}\label{sec:lb}

In this section, we provide a general lower bound on the minimum NDT for the $M\times K$ F-RAN described in the previous section. The main result is stated in the following proposition.

\begin{prop}[\textit{Lower Bound on Minimum NDT}]\label{prop:glb}
For an F-RAN with $M$ ENs, each with a fractional cache size $\mu \in [0,1]$, $K$ users, a library of $N\geq K$ files and a fronthaul capacity of $C_F = r\log(P)$ bits per symbol, the minimum NDT is lower bounded as
\begin{align}\label{eq:glb}
\delta^*(\mu,r) \geq  \delta_{LB}(\mu,r),
\end{align}
where $\delta_{LB}(\mu,r)$ is the minimum value of the following linear program (LP)
\begin{align}
						\minn     &~~\delta_F + \delta_E \\
\mathrm{subject~ to}: &~~\ell \delta_E + (M-\ell) r \delta_F \geq K -(M-\ell) (K-\ell) \mu, \label{eq:c1}\\
						          &~~\delta_F \geq 0, ~\delta_E \geq 1,\label{eq:c2}
\end{align}

\n where \eqref{eq:c1} is a family of constraints with $\ell \in [0:\min\{M,K\}]$.
\end{prop}
\begin{proof}
The proof of Proposition \ref{prop:glb} is presented in Appendix \ref{ap:thglb}.
\end{proof}

 In Proposition \ref{prop:glb}, and henceforth, we refer to \emph{fronthaul-NDT} as the normalized delivery time for fronthaul transmission, that is, 
\begin{align} \label{eq:delf}
\delta_F = \lim_{\substack{P\rightarrow \infty}}\frac{T_F \log(P)}{L},
\end{align}
and \emph{edge-NDT} as the normalized delivery time for edge transmission, that is, 
\begin{align} \label{eq:dele}
\delta_E = \lim_{\substack{P\rightarrow \infty}}\frac{T_E \log(P)}{L}.
\end{align}
Note that the NDT \eqref{eq:ndt1} is the sum of fronthaul-NDT and edge-NDT i.e., $\delta = \delta_F + \delta_E$. Proposition \ref{prop:glb} hence provides a lower bound on the minimum NDT by means of lower bounds on linear combinations of fronthaul and edge NDTs.

 The proof of the main bound \eqref{eq:c1} is based on a \textit{cut-set-like argument}, which is illustrated in Fig.  \ref{fig:proofspawc}. Specifically, it can be argued that, for all sequence of feasible policies guaranteeing a vanishing probability of error, in the high-SNR regime, any $K$ requested files must be decodable with low error probability from the received signal of $\ell$ users along with the cache contents and fronthaul messages of the remaining $(M-\ell)$ ENs. This is because, any $\ell\leq \min\{M,K\}$ received signals $\mathbf{Y}^{T_E}_{[1:\ell]}$ are functions of $M$ channel inputs $\mathbf{X}^{T_E}_{[1:M]}$, which in turn are functions of the $M$ user caches and their corresponding fronthaul messages $\mathbf{U}^{T_F}_{[1:M]}$. Thus, using these $\ell$ signals and the contents of $(M-\ell)$ caches, $S_{[1:(M-\ell)]}$ and associated fronthaul messages $\mathbf{U}^{T_F}_{[1:(M-\ell)]}$, all the inputs can be decoded using the invertible linear system of the form of \eqref{eq:chop}, neglecting the noise in the high-SNR regime. The proposition is proved by carefully bounding the joint entropy of these random variables, which upper bounds the amount of information that can be reliably conveyed in the given time intervals $T_E$ and $T_F$ or edge-NDT $\delta_E$ and fronthaul-NDT $\delta_F$.

\begin{figure}[t]
\begin{centering}
\includegraphics[width=4.7in, height=2.8in]{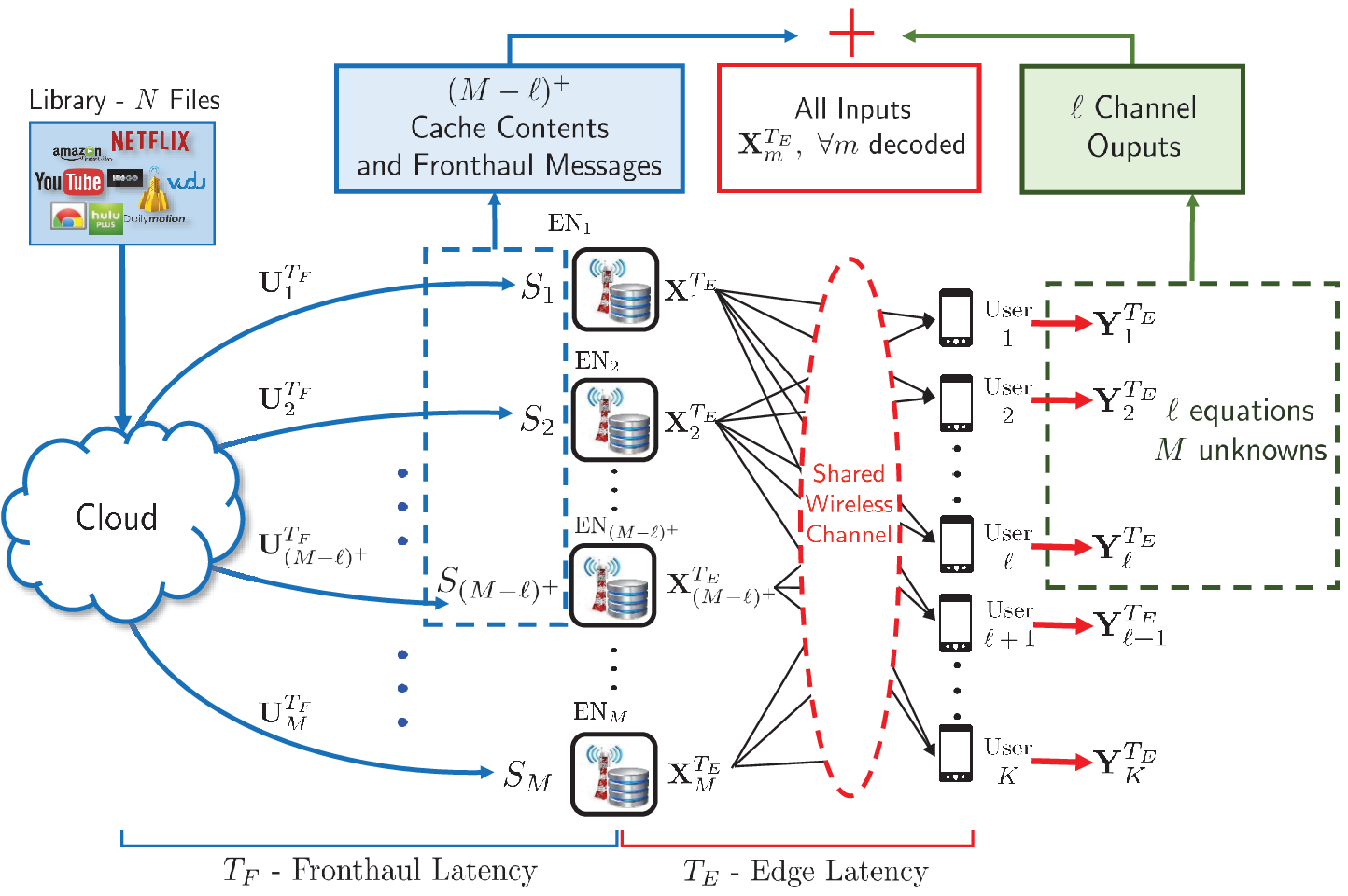}
\par\end{centering}
\protect\caption{Illustration of the proof of Proposition \ref{prop:glb}.}
\label{fig:proofspawc}\vspace{-15pt}
\end{figure}%

 We next present a sequence of corollaries that specialize the lower bound of Proposition \ref{prop:glb} to the settings of cache-only and cloud-only F-RANs (see Remark \ref{rem:nofront}).
\begin{Cor}[\textit{Lower Bound for Cache-Only F-RAN}]\label{cor:lb}
For an $M\times K$ cache-only F-RAN ($r=0$) with $\mu \in [1/M,1]$, $K$ users and a library of $N\geq K$ files, the NDT is lower bounded as
\begin{align}\label{eq:cor-ca}
\delta^*(\mu,0) \geq \max_{\ell \in [1:\min\{M,K\}]} \frac{K - (M - \ell)(K - \ell) \mu}{\ell}.
\end{align}
\end{Cor}
\begin{proof}
The proof of Corollary \ref{cor:lb} follows directly by substituting $r=0$ in constraint \eqref{eq:c1} in Proposition \ref{prop:glb} and noting that any lower bound on the optimal value of the LP in Proposition \ref{prop:glb} is also a valid lower bound on the NDT. Varying the parameter $\ell\in [1:\min\{M,K\}]$ leads to the family of lower bounds in Corollary \ref{cor:lb}.
\end{proof}

\begin{Cor}[\textit{Lower Bound for Cloud-Only F-RAN}]\label{cor:lb-cl}
For an $M\times K$ cloud-only F-RAN ($\mu = 0$), with $K$ users and a library of $N\geq K$ files, the NDT is lower bounded as
\begin{align}\label{eq:cor-cl}
\delta^*(0,r) \geq \frac{K}{\min\{M,K\}} + \frac{K}{Mr}.
\end{align}
\end{Cor}
\begin{proof}
Summing the constraints obtained from \eqref{eq:c1} by setting $\ell=M$ and $\ell=0$ yields the following lower bound on the optimal value of the LP:
\begin{align}\label{eq:cor-cl-1}
\delta^*(0,r) \geq \delta_E + \delta_F \geq \frac{K}{M} + \frac{K}{Mr}.
\end{align}
Instead, summing the constraint in \eqref{eq:c1} with $\ell=0$ and the constraint $\delta_E \geq 1$ in \eqref{eq:c2} yields the following lower bound:
\begin{align}\label{eq:cor-cl-2}
\delta^*(0,r) \geq \delta_E + \delta_F \geq 1 + \frac{K}{Mr}.
\end{align}
Combining the bounds in \eqref{eq:cor-cl-1} and \eqref{eq:cor-cl-2} yields
\begin{align}
\delta^*(0,r) \geq \max\left(\frac{K}{M} + \frac{K}{Mr}, 1 + \frac{K}{Mr}\right) = \frac{K}{\min\{M,K\}} + \frac{K}{Mr},
\end{align}
which concludes the proof.
\end{proof}

\vspace{-5pt}\section{Upper Bounds on the Minimum NDT}\label{sec:ub}
In this section, we present upper bounds on the minimum NDT by considering the performance of specific policies. We proceed by first investigating cache-aided and cloud-aided transmission strategies separately, which are then combined to obtain cloud and cache-aided policies by means of file-splitting and cache-sharing (see Lemma \ref{lem:conv}).

\vspace{-5pt}\subsection{Cache-Aided Policies}\label{ssec:cache}
We consider first cache-aided policies that do not use cloud resources and hence operate even when there is no fronthaul infrastructure, i.e., when $r=0$. We specifically focus on the two extremal scenarios in which $\mu=1$, so that all ENs can cache the entire library of files, and $\mu=1/M$, so that the library can be fully cached as long as different portions of it are stored at distinct ENs.



The following lemma provides an upper bound on the minimum NDT for the case $\mu=1$ by leveraging cache-aided EN cooperation via zero-forcing (ZF) beamforming.

\begin{lemma}[\textit{Achievable NDT with Cache-Aided EN Cooperation}]\label{lem:ach1}
For an F-RAN with fractional cache size $\mu = 1$ and any fronthaul rate $r \geq 0$, the NDT is upper bounded as $\delta^*(\mu=1,r)\leq \delta_{\mathsf{Ca-ZF}}$, where
\begin{align}\label{eq:NDT1}
&\delta_{\mathsf{Ca-ZF}} = \frac{K}{\min\{M,K\}}
\end{align}
is achieved by means of ZF-beamforming based on the cached files.
\end{lemma}
\begin{proof}
The ENs employ ZF-beamforming to serve the users' requests. Note that the worst-case demand can be easily seen to be any vector $\mathbf{D}$ of distinct files. In fact, any other vector that contains the same file for multiple users can always be delivered with the same latency by treating the files as being different.  Using ZF, a sum-rate of $\min\{M,K\}\log(P)$, neglecting $o(\log(P))$ terms, can be achieved \cite{BC_Capacity}. Thus, the delivery time per bit \eqref{eq:delta} achieved by this scheme is approximately, that is, neglecting $o(\log(P))$ terms, given by \vspace{-5pt}
\begin{align}
&\Delta\left(\mu = 1,0,P\right) =  \frac{K/\log(P)}{{\min\{M,K\}}},
\end{align}
which, by definition of the NDT (Definition \ref{def:ndt}), yields an achievable NDT $\delta_{\mathsf{Ca-ZF}} = K/\min\{M,K\}$, hence concluding the proof.
\end{proof}

The next lemma provides an upper bound on the minimum NDT that is obtained by means of cache-aided coordination strategies based on interference alignment for the case $\mu=1/M$.
\begin{lemma}[\textit{Achievable NDT with Cache-Aided EN Coordination}]\label{lem:ach1M}
For an F-RAN with fractional cache size $\mu = 1/M$ and any fronthaul rate $r \geq 0$, the NDT is upper bounded as $\delta^*(\mu,0)\leq \delta_{\mathsf{Ca-IA}}$, where
\begin{align}\label{eq:NDT1bM}
&\delta_{\mathsf{Ca-IA}} = \frac{M+K-1}{M}
\end{align}
is achievable by means of interference alignment.
\end{lemma}
\begin{proof}
Following the ideas discussed in \cite{MA-CAIC,MA-CAIC-arxiv}, each file is split into $M$ non-overlapping fragments $F_n = \left(F_{n,1},F_{n,2},\ldots,F_{n,M}\right)$, each of size $L/M$ bits. The fragment $F_{n,m}$ is stored in the cache of $\textrm{EN}_m$ for $n \in [1:N]$. Thus, the cache storage for each EN is $NL/M$ bits and $\mu = 1/M$. For any file $d_k$ is requested by a user $k$, each of the ENs has a fragment $F_{d_k,m}$ to transmit to the user. For the worst-case demand vector in which all users request different files (see proof of Lemma \ref{lem:ach1}), the $M\times K$ system becomes an X-channel, for which a reliable sum-rate of $(MK/(M+K-1))\log(P)$, neglecting $o(\log(P))$ terms, is achievable by interference alignment \cite{CJ_Xch,MA_Xch}. Thus, the achievable delivery time per bit in Definition \ref{def:delta} is approximately given by
\begin{align}
&\Delta\left(\mu = \frac{1}{M},0,P\right)  =  \frac{M+K-1}{M\log(P)}, \vspace{8pt}
\end{align}
yielding an NDT equal to $\delta_{\mathsf{Ca-IA}} = (M+K-1)/M$. 
\end{proof}

\subsection{Cloud-Aided Policies}\label{ssec:cloud}
We now move to considering cloud-aided policies that neglect the caches at the ENs and hence operate even in the case in which the ENs have no storage capabilities, that is, when $\mu=0$. We first discuss a more conventional hard-transfer fronthauling approach, whereby the fronthaul is used to send the requested files in raw form to the ENs. Then, we elaborate on the the soft-transfer scheme that is typical of C-RAN, in which quantized coded signals are transferred on the fronthaul links. 

\begin{example}[\textit{Cloud-Aided Hard-Transfer Fronthauling}]\label{exmp:htf}
Consider an F-RAN with $M=3$ ENs and $K=3$ users with a library of $K=3$ files $\{A,B,C\}$, each of size $L$ bits. We are interested in developing delivery strategies that only rely on cloud processing and fronthaul transfer, while neglecting the use of caches. We focus again on the worst-case in which each user requests a different file, i.e., $\mathbf{D}=(d_1,d_2,d_3) = (A,B,C)$. With hard-transfer fronthaul, the cloud sends files, or subfiles, over the fronthaul links to each EN, which then encodes the signal to be transmitted on the shared wireless channel. A first approach would be to send all three files, and hence $3L$ bits, to each EN, so as to enable the ENs to perform cooperative ZF-beamforming on the wireless channel. Using the fact that the fronthaul capacity is $C_F = r\log(P)$ bits per symbol, the fronthaul delivery time is $T_F = 3L/(r\log(P))$, yielding a fronthaul-NDT equal to $\delta_F= 3/r$. Since, with ZF, the edge-NDT is $\delta_E = 1$, the overall NDT achieved by this strategy is $\delta = 1 + 3/r$. Alternatively, the cloud can divide each file into three fragments and send the fragments $(A_i,B_i,C_i)$ to $\text{EN}_i$ over the corresponding fronthaul link. In this case, the fronthaul delivery time is $T_F = L/(r\log(P))$ yielding a fronthaul-NDT of $\delta_F = 1/r$. The ENs then transmit on wireless channel using interference alignment for an X-channel, achieving an edge-NDT of $\delta_E = 5/3$. Thus the achievable NDT with this approach is $\delta = 5/3 + 1/r$. Based on the available fronthaul rate $r$, the cloud can choose the policy which yields the minimum NDT. In this case, when $r\leq 3$ the interference alignment-based scheme should be utilized, while the ZF-based strategy is to be preferred otherwise.
\end{example}

Generalizing the previous example, the following proposition gives an upper bound on the minimum NDT, that can be achieved by the use of hard-transfer fronthauling.
\begin{prop}[\textit{Achievable NDT with Cloud-Aided Hard-Transfer Fronthauling}]\label{prop:htf}
For an $M\times K$ F-RAN with each EN having a fractional cache size $\mu \in [0,1]$ and a fronthaul rate of $r\geq 0$, the NDT is upper bounded as $\delta^*(\mu,r)\leq \delta_{\mathsf{Cl-Hf}}$, where
\begin{align}\label{eq:htf}
\delta_{\mathsf{Cl-Hf}} = \min\left\{\frac{K}{\min\{M,K\}} + \frac{K}{r},~ \frac{M+K-1}{M} + \frac{K}{Mr}\right\},
\end{align}
which is achieved by means of hard-transfer fronthauling.
\end{prop}
\begin{proof}
Following the discussion in Example \ref{exmp:htf}, we consider the selection between two different strategies to prove Proposition \ref{prop:htf}.

\subsubsection{Cloud-Aided EN Cooperation via ZF Beamforming} In the first strategy, the cloud transmits all the requested files to each EN over the fronthaul links. Thus, for any request vector $\mathbf{D}$, the cloud needs to transmit $KL$ bits to each EN. Since the fronthaul links have capacity $C_F = r\log(P)$ each, the fronthaul delivery time is $T_F = KL/(r\log(P))$, yielding a fronthaul-NDT of $\delta_F = K/r$. Furthermore, ZF-based EN cooperation achieves an edge-NDT of $\delta_E = K/\min\{M,K\}$ as shown in Lemma \ref{lem:ach1}. Thus, the achievable NDT with this strategy is
\begin{align}\label{eq:htf1}
\delta_F + \delta_E = \frac{K}{\min\{M,K\}} + \frac{K}{r}.
\end{align}

\subsubsection{Cloud-Aided EN Coordination via Interference Alignment} With this second strategy, for the $K$ requested files $F_{d_1},F_{d_2},\ldots,F_{d_K}$, the cloud splits each file into $M$ non-overlapping fragments $F_{d_k} = \left(F_{d_k,1},F_{d_k,2},\ldots,F_{d_k,M}\right)$, for $k\in[1:K]$, where each fragment is of size $L/M$ bits. The fragments $F_{[d_1,\ldots,d_K],m}$ are transmitted to $\textrm{EN}_m$ for $m \in [1:M]$. Thus, for a fronthaul capacity of $C_F = r\log(P)$ bits per symbol, the fronthaul delivery time is $T_F = KL/(Mr\log(P))$, yielding a fronthaul-NDT of $\delta_F = K/(Mr)$. As seen in Lemma \ref{lem:ach1M}, using an X-channel interference alignment scheme achieves an edge-NDT of $\delta_E = (M+K-1)/M$. Thus, the NDT
\begin{align}\label{eq:htf2}
\delta_F + \delta_E = \frac{M+K-1}{M} + \frac{K}{Mr},
\end{align}
is achievable via interference alignment. For a given fronthaul rate of $r$, the cloud then chooses the transmission strategy which achieves the minimum NDT between \eqref{eq:htf1} and \eqref{eq:htf2}, which yields \eqref{eq:htf}, hence concluding the proof.
\end{proof}

We now move to the consideration of the soft-transfer fronthauling approach.

\begin{example}[\textit{Cloud-Aided Soft-Transfer Fronthauling}]\label{ex:stf}
Consider an F-RAN with $M$ ENs and $K=M$ users. With soft-transfer fronthauling, as first proposed in \cite{stf}, the cloud implements ZF-beamforming and quantizes the resulting encoded signals. Using a resolution of $\log(P)$ bits per downlink baseband sample, it can be shown that the effective SNR in the downlink scales proportionally to the power $P$ (see Appendix \ref{ap:soft} and \cite{stf}). As a result, this scheme entails a fronthaul transmission time $T_F$ that equals the edge transmission time $T_E$ of the ZF-beamforming scheme, namely $T_E = L/(\log(P))$, multiplied by the time needed to carry each baseband sample on the fronthaul link, namely $\log(P)/(r\log(P))$, yielding the NDT $\delta_{\mathsf{Cl-Sf}} = 1 + 1/r$. Comparing with the NDT obtained in Example \ref{exmp:htf} by means of hard-transfer fronthauling, we see that soft-transfer fronthaul yields a lower NDT.
\end{example}

 The following proposition generalizes the previous example to give an upper bound on the minimum NDT, which is achieved by a cloud-aided policies using soft-transfer fronthauling.

\begin{figure}[!t]
\centering
\includegraphics[width=\textwidth, height=2.15in]{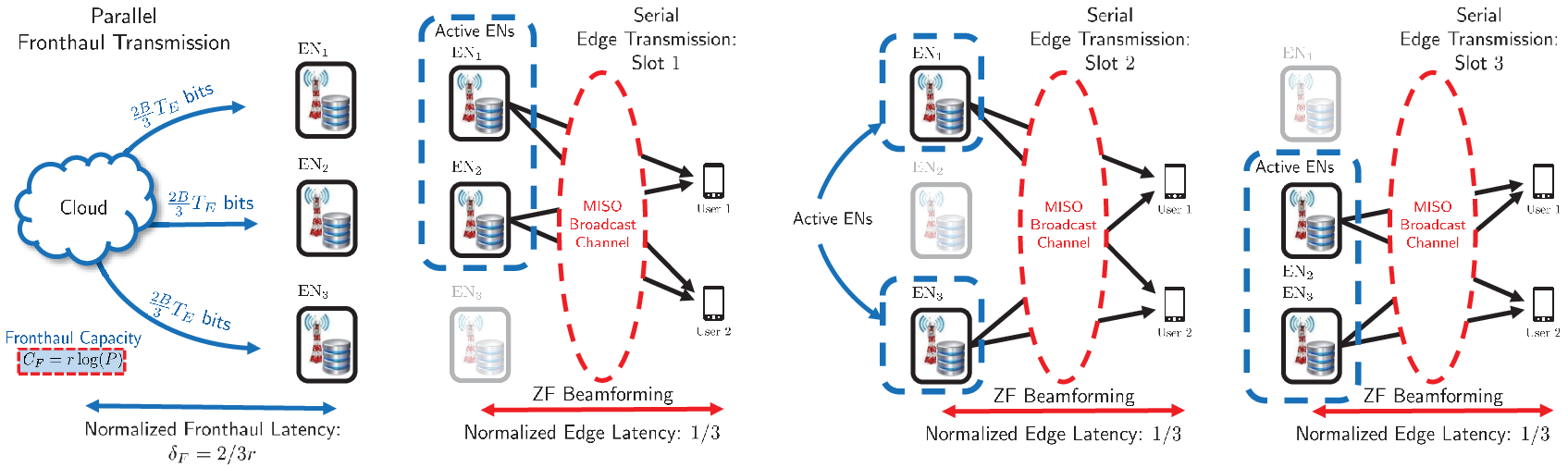}
\vspace{-12pt}
\caption{Illustration of the proposed cloud-aided soft-transfer fronthauling acheme with $M=3$ ENs and $K=2$ users.}\vspace{-10pt}
\label{fig:stf}
\end{figure}

\begin{prop}[\textit{Achievable NDT with Cloud-Aided Soft-Transfer Fronthauling}]\label{prop:stf}
For an $M\times K$ F-RAN with each EN having a fractional cache size $\mu \in [0,1]$ and a fronthaul rate $r\geq 0$, the NDT is upper bounded as $\delta^*(\mu,r)\leq \delta_{\mathsf{Cl-Sf}}$, where
\begin{align}\label{eq:stf}
\delta_{\mathsf{Cl-Sf}} = \frac{K}{\min\{M,K\}} + \frac{K}{Mr},
\end{align}
which can be achieved by means of soft-transfer fronthauling.
\end{prop}
\begin{proof}
A formal proof of Proposition \ref{prop:stf} is presented in Appendix \ref{ap:soft}. A proof sketch outlining the main ideas is provided below.
\end{proof}
For the case $M \leq K$, the main arguments follow in a manner similar to Example \ref{ex:stf}. The case $M\geq K$ instead requires a novel delivery approach that is based on the parallel transmission on the fronthaul links of quantized encoded signals that are sent using time-sharing on the wireless channel. We explain the scheme at hand with an example for an F-RAN with $M=3$ ENs serving $K=2$ users, which is illustrated in Fig.  \ref{fig:stf}. We first list all possible ${3 \choose 2} = 3$ clusters of $2$ ENs. Each of the $3$ clusters of ENs is scheduled to transmit sequentially for $1/3$ of the total edge delivery time $T_E$ on the wireless channel. The signals to be transmitted by each cluster on the wireless channels are sent in parallel on the fronthaul links by the cloud by means of a soft-transfer fronthauling strategy. Specifically, each EN participates in two clusters and hence it needs to receive only $2T_E/3$ quantized samples from the cloud on the fronthaul link. Thus, using a resolution of $B=\log(P)$ bits per sample as in Example \ref{ex:stf}, a fronthaul latency of $T_F = 2BT_E/(3C_F) = 2T_E/3r$ is achieved. This yields a fronthaul-NDT of $\delta_F = 2\delta_E/(3r)$ for a total achievable NDT of $\delta_{\mathsf{Cl-Sf}} = 1 + 2/(3r)$.

\begin{remark}[\textit{Hard vs. Soft-Transfer Fronthaul}]\label{rem:hvs}
Comparing the NDT of soft-transfer fronthauling in Proposition \ref{prop:stf} with the achievable NDT for hard-transfer fronthauling in Proposition \ref{prop:htf}, we see that the achievable NDT in Proposition \ref{prop:stf} is strictly lower, demonstrating that soft-transfer fronthauling is to be preferred when the goal is to minimize the NDT.
\end{remark}

\subsection{Cache and Cloud-Aided Policies}\label{ssec:ach}
Here, we propose a general upper bound on the minimum NDT for an F-RAN with $M$ ENs, $K$ users and $N\geq K$ files, which is attained by combining the cache-aided strategy discussed in Section \ref{ssec:cache} and the cloud-aided soft-transfer fronthaul policy of Section \ref{ssec:cloud} by means of file-splitting and cache-sharing (see Lemma \ref{lem:conv}). Note that the choice of soft-transfer fronthauling over hard-transfer fronthauling is motivated by Remark \ref{rem:hvs}.

\begin{prop}[\textit{Achievable NDT via Cloud and Cache-Aided Policies}]\label{prop:ub}
For an $M\times K$ F-RAN with a fronthaul rate of $r\geq 0$,
the minimum NDT is upper bounded as $\delta^*(\mu,r)\leq \delta_{\mathsf{Ach}}(\mu,r)$, where we have defined
\begin{align}\label{eq:r_th}
r_{\mathsf{th}} = \frac{K(M-1)}{M\left(\min\{M,K\}-1\right)},
\end{align}
and the achievable NDT is given for three distinct regimes of operation as follows:
\begin{itemize}
\item \textit{Low Cache and Low Fronthaul Regime }$\mathit{\left(\mu\leq 1/M ~\text{and}~  r\leq r_{\mathsf{th}}\right)}$:
\begin{align}\label{eq:propub1}
&\delta_{\mathsf{Ach}}(\mu,r) = (M+K-1)\mu + \left(1 - \mu M\right)\left[\frac{K}{\min\{M,K\}} + \frac{K}{Mr} \right];
\end{align}
\item \textit{High Cache and Low Fronthaul Regime }$\mathit{(\mu\geq 1/M ~\text{and}~ r\leq r_{\mathsf{th}})}$:
\begin{align}\label{eq:propub2}
& \delta_{\mathsf{Ach}}(\mu,r) = \frac{K}{\min\{M,K\}}\left(\frac{\mu M - 1}{M-1}\right) + (1-\mu)\frac{M+K-1}{M-1};
\end{align}
\item \textit{High Fronthaul Regime }$\mathit{(\mu\in[0,1]~\text{and}~ r\geq r_{\mathsf{th}})}$:
\begin{align}\label{eq:propub3}
\delta_{\mathsf{Ach}}(\mu,r) = \frac{K}{\min\{M,K\}}+\frac{(1-\mu)K}{Mr}.
\end{align}
\end{itemize}
\end{prop}

 \begin{figure}[!t]
	\centering
	\includegraphics[width=0.95\linewidth]{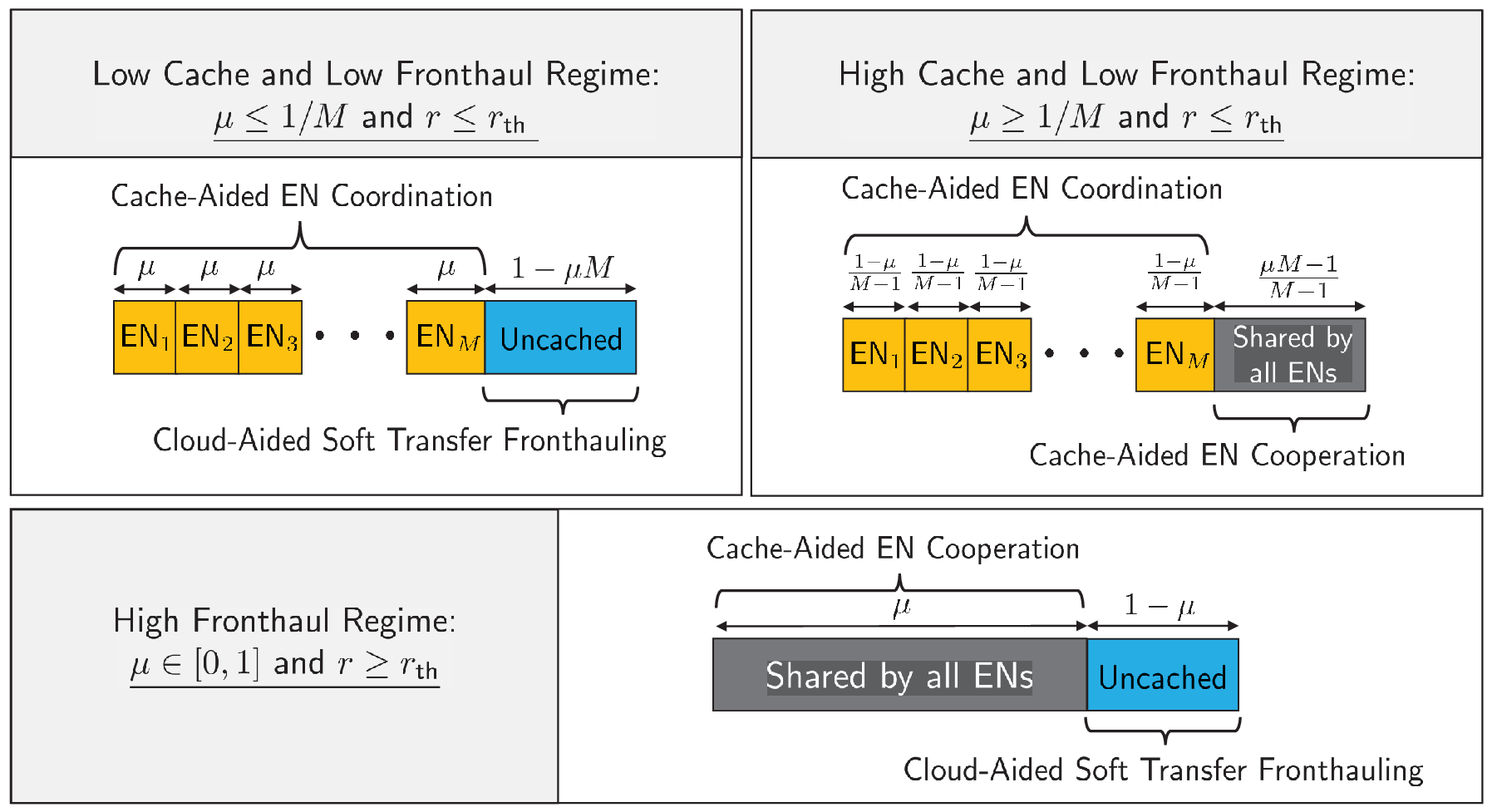}
	\caption{Caching and delivery schemes achieving the NDT in Proposition 4.}
	\label{fig:caching}
\end{figure}

\begin{proof} A formal proof follows, while a high level description of the proposed scheme can be found in Remark \ref{rem:ach_ndt}. The proposition is proved by considering the NDT of a policy that performs file-splitting and cache-sharing, as described in the proof of Lemma \ref{lem:conv}, between cache-aided and cloud-aided schemes. Specifically, for the low-cache, low fronthaul regime where $\mu\in[0,1/M]$ and $r\leq r_{\mathsf{th}}$, we use the cache-aided policy described in Lemma \ref{lem:ach1M}, yielding $\delta_{\mathsf{Ca-IA}}$, for a fraction of the files equal to $\mu M$ and the cloud-aided soft-transfer fronthauling policy described in Proposition \ref{prop:stf}, yielding $\delta_{\mathsf{Cl-Sf}}$, for the remaining $(1-\mu M)$ fraction of the files. This requires a fractional cache capacity of $\mu M \times (1/M) + (1 - \mu M)\times 0 = \mu $, since the two schemes at hand use fractional cache size $1/M$ and $0$ respectively. Moreover, the achievable NDT is
\begin{align}\label{eq:ndtub1}
\delta^{'}_{\mathsf{Ach}}(\mu,r) = (\mu M) \delta_{\mathsf{Ca-IA}} + (1 - \mu M)\delta_{\mathsf{Cl-Sf}},
\end{align}
which equals the achievable NDT in \eqref{eq:propub1}. In a similar manner, for the high-cache, low fronthaul regime where $\mu\geq 1/M$ and $r\leq r_{\mathsf{th}}$, we use the cache-aided policy described in Lemma \ref{lem:ach1M}, yielding $\delta_{\mathsf{Ca-IA}}$, for a fraction $M(1-\mu)/(M-1)$ of the files and the cache-aided policy described in Lemma \ref{lem:ach1}, yielding $\delta_{\mathsf{Ca-ZF}}$, for the remaining $(\mu M -1)/(M-1)$ fraction of files. This requires a fractional cache size of $M(1-\mu)/(M-1)\times (1/M) + (\mu M -1)/(M-1)\times 1 = \mu$ since the schemes at hand use fractional cache size of $1/M$ and $1$ respectively. The achievable NDT is 
\begin{align}\label{eq:ndtub2}
\delta^{''}_{\mathsf{Ach}}(\mu,r) = \frac{M(1-\mu)}{(M-1)}\delta_{\mathsf{Ca-IA}} +  \frac{(\mu M -1)}{(M-1)} \delta_{\mathsf{Ca-ZF}},
\end{align}
which equals the achievable NDT in \eqref{eq:propub2}. Finally, for fractional cache size $\mu\in[0,1]$ and high fronthaul rate of $r\geq r_{\mathsf{th}}$, the NDT
\begin{align}\label{eq:ndtub3}
\delta^{'''}_{\mathsf{Ach}}(\mu,r) = \mu \delta_{\mathsf{Ca-ZF}} + (1-\mu) \delta_{\mathsf{Cl-Sf}}
\end{align}
is achieved by file-splitting between the cache-aided policy described in Lemma \ref{lem:ach1}, yielding $\delta_{\mathsf{Ca-ZF}}$, for a fraction $\mu$ of the files and the cloud-aided soft transfer fronthaul policy of Proposition \ref{prop:stf}, which gives $\delta_{\mathsf{Cl-Sf}}$, for the remaining  $(1-\mu)$ fraction of the files. Note that this requires a fractional cache size of $\mu \times 1 + (1-\mu)\times 0 = \mu$ since the schemes at hand use fractional cache size of $1$ and $0$ respectively. The NDT \eqref{eq:ndtub3} equals the achievable NDT in \eqref{eq:propub3}. The fronthaul threshold, $r_{\mathsf{th}}$, is the limiting fronthaul rate for which the achievable NDT $\delta^{'''}_{\mathsf{Ach}}$ is always lower than the achievable NDTs $\delta^{'}_{\mathsf{Ach}}$ and $\delta^{''}_{\mathsf{Ach}}$. This completes the proof of Proposition \ref{prop:ub}.
\end{proof}

\begin{remark}[\textit{Achievable Scheme}]\label{rem:ach_ndt}
 The scheme that achieves the NDT \eqref{eq:propub1}-\eqref{eq:propub3} can be explained with reference to Fig. \ref{fig:caching} as follows. In the regime of high fronthaul rate, i.e., $r\geq r_{\mathsf{th}}$, the same fraction $\mu$ of each file is stored at all ENs. For any request vector, the cached fraction $\mu$ of the requested files is then transmitted using cache-aided EN cooperation, while the uncached fraction $1-\mu$ is delivered using cloud-aided soft transfer fronthauling. In contrast, in the regime of low fronthaul rate, i.e., $r<r_{\mathsf{th}}$, placement is aimed at caching fractions of each file that are as distinct as possible at different ENs. The rationale of this choice is that one would like to minimize fronthaul usage in this regime. As a result, when $r<r_{\mathsf{th}}$, we need to distinguish the two cases represented in Fig. \ref{fig:caching} in which the $\mu$ fractions of each file stored at different ENs are either fully disjoint, possibly leaving a portion of each file uncached, or partially overlapping. Delivery of the shared and uncached fractions takes place using cache-aided EN cooperation and cloud-aided soft transfer fronthauling, respectively. Instead, the fractions cached at distinct ENs are delivered using cached-aided EN coordination. 
\end{remark}

\begin{remark}[\textit{Size of File Library}]
	In this work we assume that the number of files in the library is greater than the number of users i.e., $N\geq K$. This is justified for scenarios in which the users simultaneously active in the area monitored by a cloud processor is less than the size of the content library. Given that the latter may contain at least tens of thousands of files, this is a practically relevant assumption. Furthermore, from a technical standpoint, the lower bound in \eqref{eq:glb} relies on this assumption since the confition $N\geq K$ ensures that the worst-case request vector is such that each user has a distinct demand. For $N\leq K$, the lower bound will be smaller, reflecting the possible gains that can be accrued via multicasting. Instead, the upper bound on the NDT (achievable NDT) is still valid for the case of $N\leq K$. In fact, we can always deliver common files as distinct files for each user.
\end{remark}

\section{Characterization of the Minimum NDT}\label{sec:opt}
Based on the lower and upper bounds presented in Sections \ref{sec:lb} and \ref{sec:ub}, in this section we
show that the proposed achievable schemes in Section \ref{sec:ub} are optimal 
in a number of important special cases, including cloud-only F-RANs, also known as C-RAN; cache-only F-RANs for extremal values of fractional cache size $\mu$; and general F-RAN models with both cloud processing and caching for the case when the number of users exceeds the number of ENs in the low fronthaul regime. Furthermore, we present a constant factor approximation of the minimum NDT, $\delta^*(\mu,r)$, for all values of problem parameters, which shows that the proposed achievable schemes are approximately optimal to within a factor of at most $2$. To proceed we first consider separately cache-only and cloud-only F-RAN and then study the general F-RAN  model.

\subsection{Minimum NDT for Cache-Only F-RAN}
The following proposition characterizes the minimum NDT for a cache-only F-RAN ($r=0$) for extremal values of the fractional cache size i.e., for $\mu\in \{1/M,1\}$. We recall that with $\mu\leq 1/M$, the minimum NDT is unbounded (see Remark \ref{rem:nofront}).

\begin{prop}[\textit{Minimum NDT for Cache-Only F-RAN}]\label{prop:corner}
For an $M\times K$ F-RAN with a fronthaul rate $r =0$, the minimum NDT given by
\begin{align}\label{eq:corr11}
&\delta^*(\mu,0) = \begin{cases}
									 \delta_{\mathsf{Ca-IA}} ~~~ \text{for}~~ \mu = 1/M,\\
									 \delta_{\mathsf{Ca-ZF}} ~~~ \text{for}~~ \mu = 1,
									 \end{cases}	
\end{align}
where $\delta_{\mathsf{Ca-IA}}$ can be achieved by means of EN coordination via interference alignment (see \eqref{eq:NDT1bM}) and $\delta_{\mathsf{Ca-ZF}}$ can be achieved by EN cooperation via ZF-beamforming (see \eqref{eq:NDT1}).
\end{prop}
\begin{proof}
The proof of Proposition \ref{prop:corner} is provided in Appendix \ref{ap:corner}. 
\end{proof}

The result indicates that, in a cache-only F-RAN, the proposed converse in Corollary \ref{cor:lb} is tight at extremal values of fractional cache size $\mu$, and that cache-aided EN cooperation and coordination are optimal for $\mu=1$ and $\mu=1/M$, respectively. 

\subsection{Minimum NDT for Cloud-Only F-RAN}
The following proposition gives the minimum NDT for a cloud-only F-RAN ($\mu = 0$), showing the optimality of soft-transfer fronthauling.

\begin{prop}[\textit{Minimum NDT for Cloud-Only F-RAN}]\label{prop:corner_cl}
For an $M\times K$ F-RAN with $\mu = 0$, the minimum NDT is characterized as
\begin{align}\label{eq:corr10}
\delta^*(0,r) = \delta_{\mathsf{Cl-Sf}}
\end{align}
for $r\geq 0$ which can be achieved by soft-transfer fronthauling (see \eqref{eq:stf}).
\end{prop}

\begin{proof}
The proof follows directly from the lower bound on the minimum NDT for cloud-only F-RANs, presented in  Corollary \ref{cor:lb-cl} and from the achievable NDT presented in Proposition \ref{prop:stf} that uses soft-transfer fronthauling.
\end{proof}

\subsection{Approximate Characterization of the Minimum NDT for a Cache and Cloud-Aided F-RAN}\label{ssec:ndtapprox}
We next provide another exact characterization of the NDT, in addition to the results in Propositions \ref{prop:corner} and \ref{prop:corner_cl} for $r=0$ and $\mu\in\{1/M,1\}$ and for $\mu=0$ respectively.
Specifically, in the low cache memory regime in which $\mu \in[0, 1/M]$, when the number of ENs is smaller than the number of users, i.e., $M\leq K$, and the fronthaul rate is small i.e., $r\leq 1/(M-1)$, the following proposition gives the minimum NDT. 

\begin{prop}[\textit{Minimum NDT for F-RAN with Low Fronthaul and Cache Size}]\label{prop:lowmu}
For an $M\times K$ F-RAN with $M\leq K$ and with each EN having a fractional cache size $\mu \in [0,1/M]$ and a fronthaul rate of  $r \in (0,1/(M-1)]$, the minimum NDT is given as
\begin{align}\label{eq:lowmu}
\delta^*(\mu,r) = (M+K-1)\mu + \frac{K\left(1 - \mu M\right)}{M}\left(1 + \frac{1}{r} \right).
\end{align}
\end{prop}
\begin{proof}
The proof of Proposition \ref{prop:lowmu} is provided in Appendix \ref{ap:lowmu}.
\end{proof}

We finally provide an approximate characterization of the minimum NDT for a general $M\times K$ F-RAN by showing that the lower bound in Proposition \ref{prop:glb} and the upper bound in Proposition \ref{prop:ub}, are within a constant multiplicative gap equal to $2$, independent of problem parameters for all regimes of fractional cache size $\mu$ and fronthaul rate $r$.

\begin{prop}[\textit{Minimum NDT for a General F-RAN}]\label{prop:gap}
For a general $M\times K$ F-RAN, we have
\begin{align}\label{eq:gap}
\frac{\delta_{\mathsf{Ach}}(\mu,r)}{\delta^*(\mu,r)} \leq 2,
\end{align}
for $\mu\in[1/M,1]$ when $r=0$ (cache-only F-RAN) and for $\mu\in[0,1]$ when $r>0$ (cloud and cache-aided F-RAN).
\end{prop}
\begin{proof}
The proof of Proposition \ref{prop:gap} is given in Appendix \ref{ap:gap}. 
\end{proof}

\begin{remark}\label{rem:caim}
When considering a cache-aided F-RAN ($r=0$), the system studied in this paper becomes a special case of the system considered in \cite{avestimehr-caim}, which is a cache-aided system with caching at both ENs and users. The authors in \cite{avestimehr-caim} show that, under the constraint of linear precoding strategies for transmission over the wireless channel, the optimal sum-DoF can be characterized to within a factor of $2$. Proposition \ref{prop:gap} shows that the factor $2$ approximation of the minimum NDT, and hence of the sum-DoF, as seen in Remark \ref{rem:dof}, holds over a larger class of precoding schemes, including non-linear transmission strategies, and that it extends to cloud and cache-aided F-RANs.
\end{remark}

\begin{remark}[\textit{Sub-Packetization of Files}]
The caching and delivery schemes designed for the cache-only systems studied in \cite{MA-CAIC} and \cite{avestimehr-caim} are based on techniques that require splitting each file into a number of sub-packets which increases \textit{exponentially} in the number of ENs. In contrast, as illustrated in Fig. \ref{fig:caching}, the scheme proposed in this paper requires only at most $M+1$ sub-packets, which is \textit{linear} in the number $M$ of ENs.
\end{remark}

%

\begin{remark}[\textit{Minimum NDT for the }$2\times 2$ \textit{F-RAN}]\label{rem:nk22}
As proved in \cite{RT_ISIT16}, the achievable scheme described in Remark \ref{rem:ach_ndt} is optimal for an F-RAN with $M=2$ ENs and $K=2$ users as summarized in the following corollary.

\begin{Cor}[\textit{Theorem} $\mathit{1}$, \cite{RT_ISIT16}]\label{cor:mk22}
For an F-RAN with $M = 2$ ENs, $K = 2$ users and $N \geq 2$ files, the minimum NDT is characterized as

\n $\bullet$~Cache-Only F-RAN $(r=0)$:
\begin{align}
&\delta^*(\mu,r) = 2 - \mu.
\end{align}

\n $\bullet$~Low Fronthaul $(r\in (0,1])$:
\begin{align}
&\delta^*(\mu,r) = \max\Big(1 + \mu + \frac{1-2\mu}{r}, 2-\mu \Big).
\end{align}

\n $\bullet$~High Fronthaul $(r>1)$:
\begin{align}
&\delta^*(\mu,r) = 1 + \frac{1-\mu}{r}.
\end{align}
\end{Cor}
For the converse, which is not reported in \cite{RT_ISIT16}, please see Appendix \ref{ap:mk22}. The achievability follows as per Remark \ref{rem:ach_ndt}. A partial characterization of the minimum NDT of an F-RAN with $M=3$ ENs and $K=3$ users is provided in our related work in \cite{aviksg-spawc} by leveraging the results presented in the previous sections and the achievable schemes presented in \cite{MA-CAIC,MA-CAIC-arxiv}. 
\end{remark}

\begin{figure}[!t]
\centering
\subfigure[]{
\includegraphics[width=3.5in, height=2.75in]{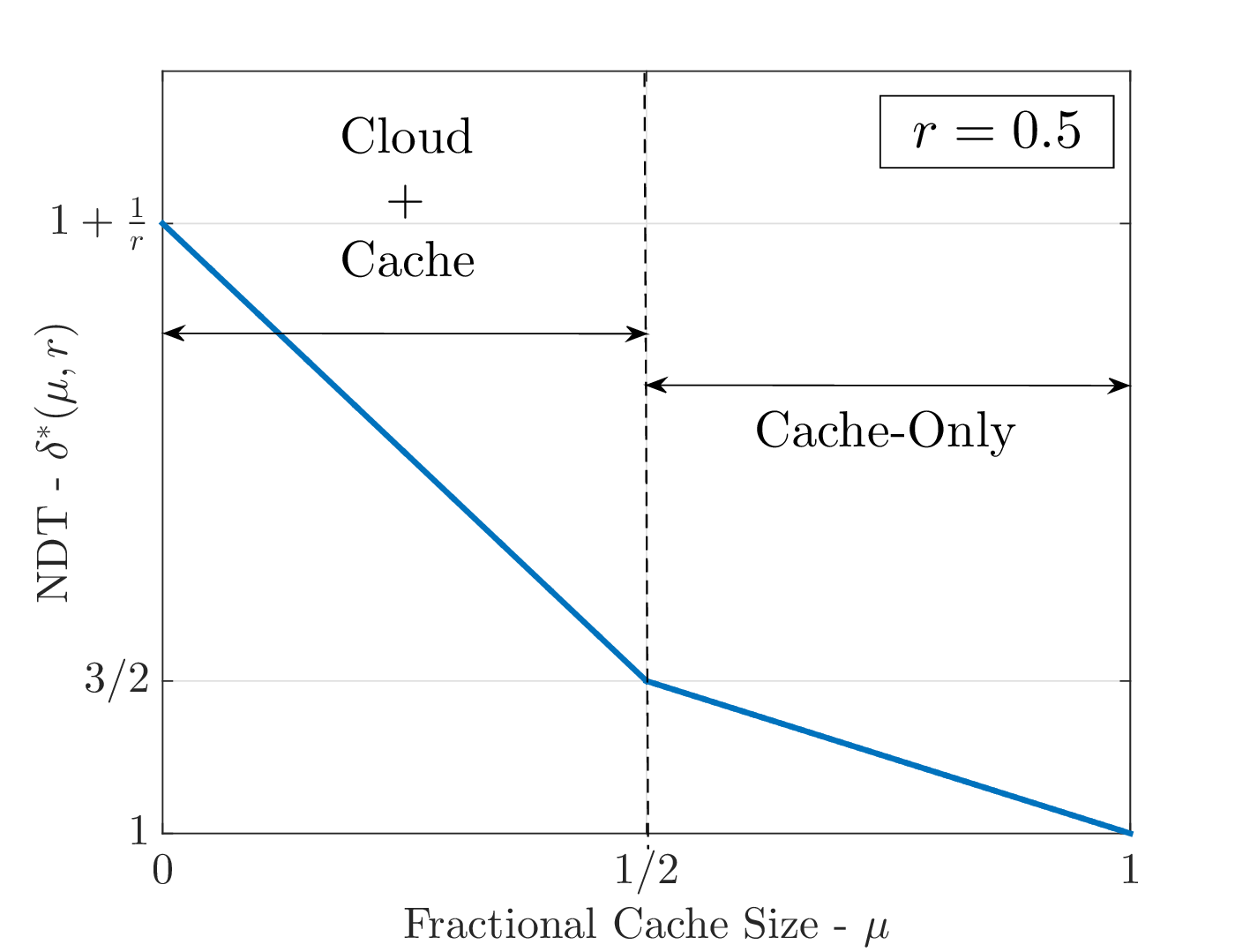}
\label{fig:MK22_r0p5}
}\hspace{-15pt}
\subfigure[]{
\includegraphics[width=3.5in, height=2.75in]{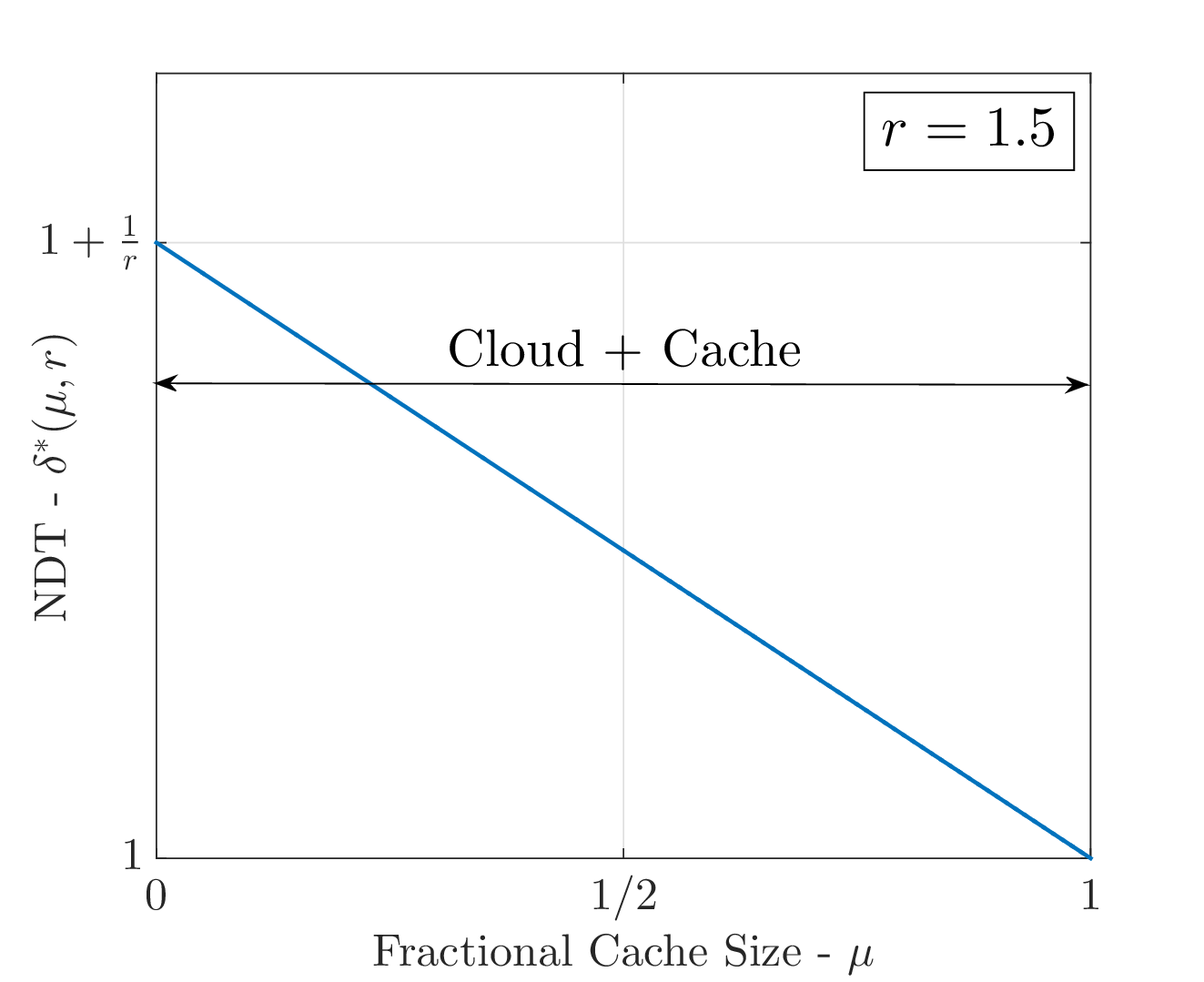}
\label{fig:MK22_r1p5}
}
\caption{Minimum NDT for an F-RAN with $M=K=2$: $(a)$ low fronthaul regime, here $r=0.25$; and $(b)$ high fronthaul regime, here $r=1.5$. The labels "Cache" and "Cloud" refer to the achievable schemes.}
\end{figure}
\section{Pipelined Fronthaul-Edge Transmission}\label{sec:pipe}
In this section, we elaborate on the F-RAN model with pipelined fronthaul-edge transmission introduced in Section \ref{ssec:ndt} (see Remark \ref{rem:pipe}). The following lemma bounds the improvement in NDT that can be achieved by the use of pipelined fronthaul-edge transmission as compared to serial fronthaul-edge transmission.

\begin{lemma}[\textit{Pipelined vs. Serial Fronthaul-Edge Transmission}]\label{lem:pipe_serial}
For an $M\times K$ cloud and cache-aided F-RAN, pipelined fronthaul-edge transmission can improve the minimum NDT as compared to serial transmission by a factor of at most $2$, i.e., 
\begin{align}
\delta_{\mathsf{P}}^*(\mu,r) \geq \frac{\delta^*(\mu,r)}{2}.
\end{align}
\end{lemma}
\begin{proof}
For the case of pipelined fronthaul-edge transmission, consider an optimal policy $\pi^*_{\mathsf{P}}$ that achieves the minimum NDT $\delta^*_{\mathsf{P}}(\mu,r)$. We use this policy $\pi^*_\mathsf{P}$ to construct a policy $\pi$ for serial fronthaul-edge transmission model as follows: the caching and fronthaul policies for $\pi$ are the same as for $\pi^*_{\mathsf{P}}$; and the edge-transmission policy for $\pi$ is the same as for $\pi^*_{\mathsf{P}}$ with the caveat that the ENs start transmitting only after the fronthaul transmission is complete. The NDT $\delta(\mu,r)$ achieved by the serial policy $\pi$ is no larger than $2\delta^*_{\mathsf{P}}(\mu,r)$ since the durations of fronthaul and edge transmission for $\pi^*_{\mathsf{P}}$ are by definition of the NDT \eqref{eq:ndt_pipe}, both limited by $\delta^*_{\mathsf{P}}(\mu,r)$ when normalized by $L/\log(P)$ in the limit of large $L$ and $P$. This concludes the proof.
\end{proof}

We next derive a lower bound on the minimum NDT $\delta^*_{\mathsf{P}}(\mu,r)$ based on Proposition \ref{prop:glb} and an upper bound that relies on the fronthaul and edge transmission strategies discussed in Section \ref{sec:ub}. Since the results concerning cache-only F-RANs ($r=0$) coincide with those presented thus far, we focus here only on the case $r>0$.

\subsection{Lower Bound on the Minimum NDT}\label{ssec:pipe_lb}
Here, we provide a general lower bound on the minimum NDT for the $M\times K$ F-RAN with pipelined fronthaul-edge transmission. The main result is stated in the following corollary which can be derived based on Proposition \ref{prop:glb}.
\begin{Cor}[\textit{Lower Bound on the Minimum NDT for Pipelined Fronthaul-Edge Transmission}]\label{cor:pipe_glb}
For an F-RAN with $M$ ENs, each with a fractional cache size $\mu \in [0,1]$, $K$ users, a library of $N\geq K$ files and a fronthaul capacity of $C_F = r\log(P)$ bits per symbol, the minimum NDT for pipelined fronthaul-edge transmission is lower bounded as
\begin{align}\label{eq:pipe_glb}
\delta_{\mathsf{P}}^*(\mu,r) \geq \max\left\{\max_{\ell \in [0:\min\{M,K\}]}\frac{K - (M-\ell)(K-\ell)\mu}{\ell + (M-\ell) r}~, ~1\right\}.
\end{align}
\end{Cor}
\begin{proof}
The corollary is proved via the same steps as in the proof of Proposition \ref{prop:glb} (see Appendix \ref{ap:thglb}) with the following caveat. For pipelined fronthaul-edge transmission, the vectors $\mathbf{U}_m^{T},\mathbf{X}_m^{T},\mathbf{Y}_k^{T}$ and $\mathbf{n}_k^{T}$ corresponding to the fronthaul messages and transmitted signal for each $\text{EN}_m$, and the received signal and channel noise for each user $k$, respectively, have $T$ entries, as per \eqref{eq:tx_pipe}, where $T$ is the overall transmission latency. This is because pipelining allows for parallel fronthaul-edge transmissions. Using these definitions, along with \eqref{eq:ndt_pipe}, and following the same steps as in \eqref{eq:p1}-\eqref{eq:proof_c1} in Appendix \ref{ap:thglb}, the first term in the lower bound can be derived. The second term follows in a similar manner from \eqref{eq:proof_c2} in Appendix \ref{ap:thglb}.
\end{proof}

To provide some intuition on the lower bound \eqref{eq:pipe_glb} in relation to Proposition \ref{prop:glb}, we note that, for an F-RAN with pipelined fronthaul-edge transmission, the fronthaul and edge transmission intervals generally overlap and hence the fronthaul-NDT $\delta_F$ and the edge-NDT $\delta_E$, which may be defined as in \eqref{eq:delf} and \eqref{eq:dele}, satisfy $\max\{\delta_F,\delta_F\}\leq \delta$, where $\delta$ is the overall NDT. Therefore, from constraint \eqref{eq:c1} of Proposition \ref{prop:glb}, by setting $\delta_E\leq\delta$ and $\delta_F \leq \delta$ and maximizing over all $\ell$ we obtain the first term inside the $\max(\cdot)$ function. The second term follows in a similar manner from \eqref{eq:c2}. We also observe that the lower bound \eqref{eq:pipe_glb} is strictly smaller than the lower bound \eqref{eq:glb} derived under serial operation in accordance with the discussion in Remark \ref{rem:pipe}. Next, we consider achievable schemes that yield upper bounds on the minimum NDT for the pipelined fronthaul-edge transmission model. 

\subsection{Upper Bounds on the Minimum NDT}\label{ssec:pipe_ub}

The proposed achievable scheme for pipelined fronthaul-edge transmission leverages \textit{block-Markov encoding} to convert serial transmission policies discussed in Section \ref{sec:ub} to pipelined policies. We further integrate block-Markov encoding with \textit{per-block file splitting} to time-share between two transmission policies within each block. 

\begin{itemize}
	\item \textit{Block-Markov Encoding}: To convert a serial policy into a pipelined policy, we split each file in the library into $B$ blocks, so that each block is of size $L/B$ bits.  Correspondingly, we also divide the total delivery time $T$ into $B+1$ slots, each of duration $T/(B+1)$. In each slot $b\in [1:B]$, the cloud operates the fronthaul according to the serial policy to deliver the $b$th blocks of the requested files, while the ENs apply the corresponding edge delivery policy to deliver the $(b-1)$th blocks of the requested files, as illustrated in Fig. \ref{fig:bmarkov}. 
	
	Let $T_F^{(B)}$ denote the per-block fronthaul time and $T_E^{(B)}$ denote the per-block edge time required by the selected policies in each block. These times are related to the total fronthaul and edge delivery times $T_F$ and $T_E$ of the serial policy as $T_F^{(B)} = T_F/B$ and $T_E^{(B)} = T_E/B$, since in each block, only a fraction $L/B$ of a file is transmitted. The total delivery time per bit is hence given by
	\begin{align}
	\Delta_{\mathsf{P}}(\mu,C_F,P) &= \limsup_{L\rightarrow\infty}~ \frac{(B+1)\max\left(T_F^{(B)}, T_E^{(B)}\right)}{L} = \limsup_{L\rightarrow\infty}~\frac{(B+1)}{B}\frac{\max\left(T_F,T_E\right)}{L}.
	\end{align} 
	
	\n The corresponding NDT \eqref{eq:ndt} is computed as 
	\begin{align}\label{eq:ach_pipeline}
	\delta_{\mathsf{P,Ach}}(\mu,r) &= \lim_{B\rightarrow \infty}\lim_{P\rightarrow \infty} \limsup_{L\rightarrow \infty} \frac{(B+1)}{B}\frac{\max\left(T_F,T_E\right)}{L/\log(P)}\nonumber\\
	& = \max\left(\delta_F,\delta_E\right),
	\end{align}
	where $\delta_F$ and $\delta_E$ are the fronthaul and edge NDTs of the serial transmission scheme. Thus, under the limit of an arbitrarily large number of blocks $B$, the achievable NDT under pipelined fronthaul-edge transmission is the \textit{maximum} of the edge and fronthaul NDTs of the serial policy. 
	\begin{figure}[!t]
		\centering
			\includegraphics[width=5.25in,height=2.25in]{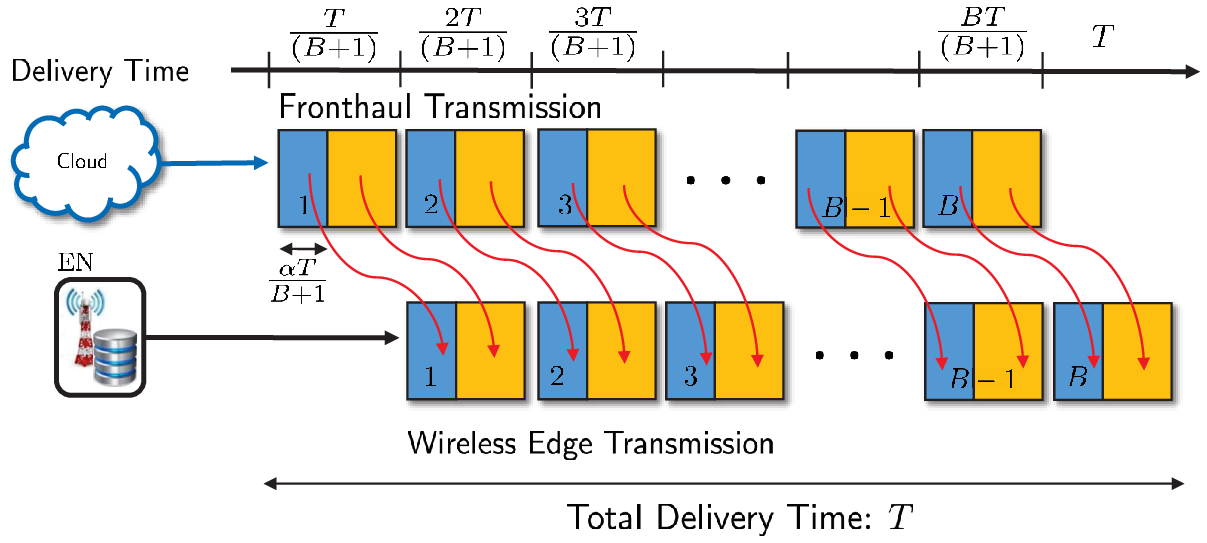}
		\protect\caption{Illustration of pipelined F-RAN operation with per-block file-splitting.}
		\label{fig:bmarkov}
	\end{figure}
	
	\item \textit{Per-Block File Splitting}: To further improve the performance of the block-Markov coding, we propose a per-block \textit{file-splitting} strategy in order to time-share between any two serial fronthaul-edge policies. To elaborate, for some $\alpha \in [0,1]$ fraction of each file block (of size $L/B$ bits), a (serial) policy requiring total fronthaul and edge NDTs $\delta_F^{(1)}$ and $\delta_E^{(1)}$ is used, and for the remaining $(1-\alpha)$ fraction of each file block, a (serial) policy requiring NDTs $\delta_F^{(2)}$ and $\delta_E^{(2)}$ is used. This is illustrated in Fig. \ref{fig:bmarkov}. Based on the discussion above, this yields an achievable NDT of
	\begin{align}\label{eq:pipe_ach_ts}
	\delta_{\mathsf{P,Ach}} = \max\left(\alpha\delta_F^{(1)} + (1-\alpha)\delta_F^{(2)}, \alpha\delta_E^{(1)} + (1-\alpha)\delta_E^{(2)}\right).
	\end{align} 
\end{itemize}

The following proposition gives an achievable NDT by considering a pipelined fronthaul-edge transmission strategy that utilizes cloud-aided soft-transfer fronthauling along with either cache-aided EN coordination via interference alignment or cache-aided EN cooperation via ZF-beamforming (see Section \ref{sec:ub}) as the constituent schemes, as for Proposition \ref{prop:ub}. We note that, unlike Proposition \ref{prop:ub}, we do not consider file-splitting between cache-only schemes (cf. \eqref{eq:ndtub2}) without the use of fronthaul transmission, since it can be shown that this would not improve the NDT in the presence of pipelined fronthaul-edge transmission.

\begin{prop}[\textit{Achievable NDT for Pipelined Fronthaul-Edge Transmission}]\label{prop:ach_pipe}
For an $M\times K$ F-RAN with a fronthaul rate of $r > 0$, the minimum NDT for pipelined fronthaul-edge transmission is upper bounded as $\delta_{\mathsf{P}}^*(\mu,r)\leq \delta_{\mathsf{P,Ach}}(\mu,r)$, where
\begin{align}\label{eq:ach_pipe}
\delta_{\mathsf{P,Ach}}(\mu,r) = \begin{cases}\displaystyle
																  \delta_{\mathsf{P-IA}} ~=~ \dfrac{(1-\mu M)K}{Mr}  \hspace{163.75pt} \text{for}~~ \mu\in[0,\mu_1],\vspace{5pt}\\
																	\delta_{\mathsf{P-FS}} ~=~ \dfrac{K}{Mr} \left[1 - \mu_2 - \left[\mu_1 M - \mu_2\right]\left(\dfrac{\mu_2-\mu}{\mu_2 - \mu_1} \right)^+ \right]  \hspace{15pt} \text{for}~~ \mu\in[\mu_1,\mu_2],\vspace{5pt}\\
																	\delta_{\mathsf{P-ZF}} ~=~ \dfrac{K}{\min\{M,K\}} \hspace{162.5pt} \text{for}~~ \mu\in[\mu_2,1],
																 \end{cases}	
\end{align}
and
\begin{align}\label{eq:mu12}
\mu_1 = \left(\frac{K - \max\{M,K\} r}{KM + Mr\left[\min\{M,K\} - 1\right]}\right)^+, ~~ \mu_2 = \left(1 - \frac{Mr}{\min\{M,K\}}\right)^+,
\end{align}
with $\mu_1\leq \mu_2 \leq 1$. The NDT $\delta_{\mathsf{P-IA}}$ is achieved by file-splitting between cloud-aided soft-transfer fronthauling and cache-aided EN coordination via X-channel based interference alignment; the NDT $\delta_{\mathsf{P-ZF}}$ is achieved by file-splitting between cloud-aided soft-transfer fronthauling and cache-aided EN cooperation via ZF-beamforming; and the NDT $\delta_{\mathsf{P-FS}}$ is achieved by file-splitting between the schemes achieving $\delta_{\mathsf{P-IA}}$ at $\mu=\mu_1$ and $\delta_{\mathsf{P-ZF}}$ at $\mu=\mu_2$ respectively.
\end{prop}

\begin{proof}
The proof is presented in Appendix \ref{ap_ssec:ach_pipe}.
\end{proof}

As indicated in Proposition \ref{prop:ach_pipe}, the NDT \eqref{eq:ach_pipe} is achieved by selecting the best among (with smallest NDT) three block-Markov strategies which use as constituent schemes cloud-aided soft-transfer on the fronthaul and either cache-aided ZF-beamforming or X-channel-based interference alignment on the edge. An illustration will be provided below for a $2\times 2$ F-RAN. 

%
%

\vspace{-5pt}\subsection{Minimum NDT for a Cloud and Cache-Aided F-RAN}\label{ssec:pipe_gen}
We next provide a partial characterization of the minimum NDT for a general cloud and cache-aided F-RAN with pipelined fronthaul-edge transmission. Specifically, the following proposition gives the minimum NDT for the low cache regime with $\mu\in [0,\mu_1]$; for the high cache regime with $\mu\in[\mu_2,1]$; and for the high fronthaul regime with $r \geq ((1-\mu)\min\{M,K\})/{M}$.

\begin{prop}[\textit{Minimum NDT for a General F-RAN with Pipelined Fronthaul-Edge Transmissions}]\label{prop:pipe_gen}
For a general $M\times K$ F-RAN, with pipelined fronthaul-edge transmission and with fronthaul rate $r>0$, we have
\begin{align}\label{eq:pipe_gen1}
\delta_{\mathsf{P}}^*(\mu,r) = \begin{cases}
									\delta_{\mathsf{P-IA}}, ~~~ \text{for}~~~ \mu \in [0,\mu_1],\vspace{7.5pt}\\ 
									\delta_{\mathsf{P-ZF}}, ~~~ \text{for}~~~ \mu \in [\mu_2,1],
									\end{cases}
\end{align}
where $\delta_{\mathsf{P-IA}}$ and $\delta_{\mathsf{P-ZF}}$ are defined in \eqref{eq:ach_pipe} and the fractional cache sizes $\mu_1,\mu_2$ are defined in \eqref{eq:mu12}. Furthermore, for any fractional cache size $\mu\in[0,1]$, we have
\begin{align}\label{eq:pipe_gen2}
\delta_{\mathsf{P}}^*(\mu,r) = \delta_{\mathsf{P-ZF}}, ~~~ \text{for}~~~ r \geq \frac{(1-\mu)\min\{M,K\}}{M}.
\end{align}
\end{prop}

\begin{proof}
The proof is presented in Appendix \ref{ap_ssec:pipe_gen}.
\end{proof}

\begin{remark}\label{rem:cloud_pipe}
Proposition \ref{prop:pipe_gen}, along with Proposition \ref{prop:ach_pipe}, demonstrate that, even with partial caching, i.e., with $\mu < 1$, it is possible to achieve the same performance as in a system with full caching or ideal fronthaul, namely $\delta = \delta_{\mathsf{P-ZF}} = K/\min\{M,K\}$. This is the case as long as either the fronthaul capacity is large enough (see \eqref{eq:pipe_gen2}) or the fronthaul capacity is positive and the cache capacity $\mu$ is sufficiently large (see \eqref{eq:pipe_gen1}). We observe that this is not true for serial fronthaul-edge transmission, in which case no policy can achieve the NDT $\delta = K/\min\{M,K\}$ for $\mu<1$ and finite fronthaul capacity. The intuition behind this result is that, with pipelined transmission, cloud resources can be leveraged to make up for partial caching by transmitting on the fronthaul while edge transmission takes place (see \cite{Leconte} for practical implications).
\end{remark}

We finally provide an approximate characterization of the minimum NDT for a general $M\times K$ F-RAN with pipelined fronthaul-edge transmission by showing that the lower bound in Corollary \ref{cor:pipe_glb} and the upper bound in Proposition \ref{prop:ach_pipe} are within a constant multiplicative gap, independent of problem parameters for any fronthaul rate $r>0$, in the intermediate cache regime with $\mu \in [\mu_1,\mu_2]$, where the minimum NDT is not characterized by Proposition \ref{prop:pipe_gen}.

\begin{prop}[\textit{Approximate Characterization of Minimum NDT in the Intermediate Cache Regime}]\label{prop:gap_pipe}
For a general $M\times K$ F-RAN with pipelined fronthaul-edge transmission and with fronthaul rate $r>0$, we have
\begin{align}\label{eq:gap_pipe}
\frac{\delta_{\mathsf{P,Ach}}(\mu,r)}{\delta_{\mathsf{P}}^*(\mu,r)} \leq 2, ~~~~ \text{for}~~\mu\in[\mu_1,\mu_2].
\end{align}
\end{prop}
\begin{proof}
The proof of Proposition \ref{prop:gap_pipe} is presented in Appendix \ref{ap_ssec:gap_pipe}. 
\end{proof}

\subsection{Case Study: $2\times 2$ F-RAN with Pipelined Fronthaul-Edge Transmission}\label{ssec:nk22_pipe}
In this section, we provide the complete characterization of the minimum NDT of an F-RAN with $M=2$ ENs and $K=2$ users with pipelined fronthaul-edge transmission and we offer insights on optimal delivery policies.
\begin{figure}[!t]
\centering
\includegraphics[width=3.5in,height=3in]{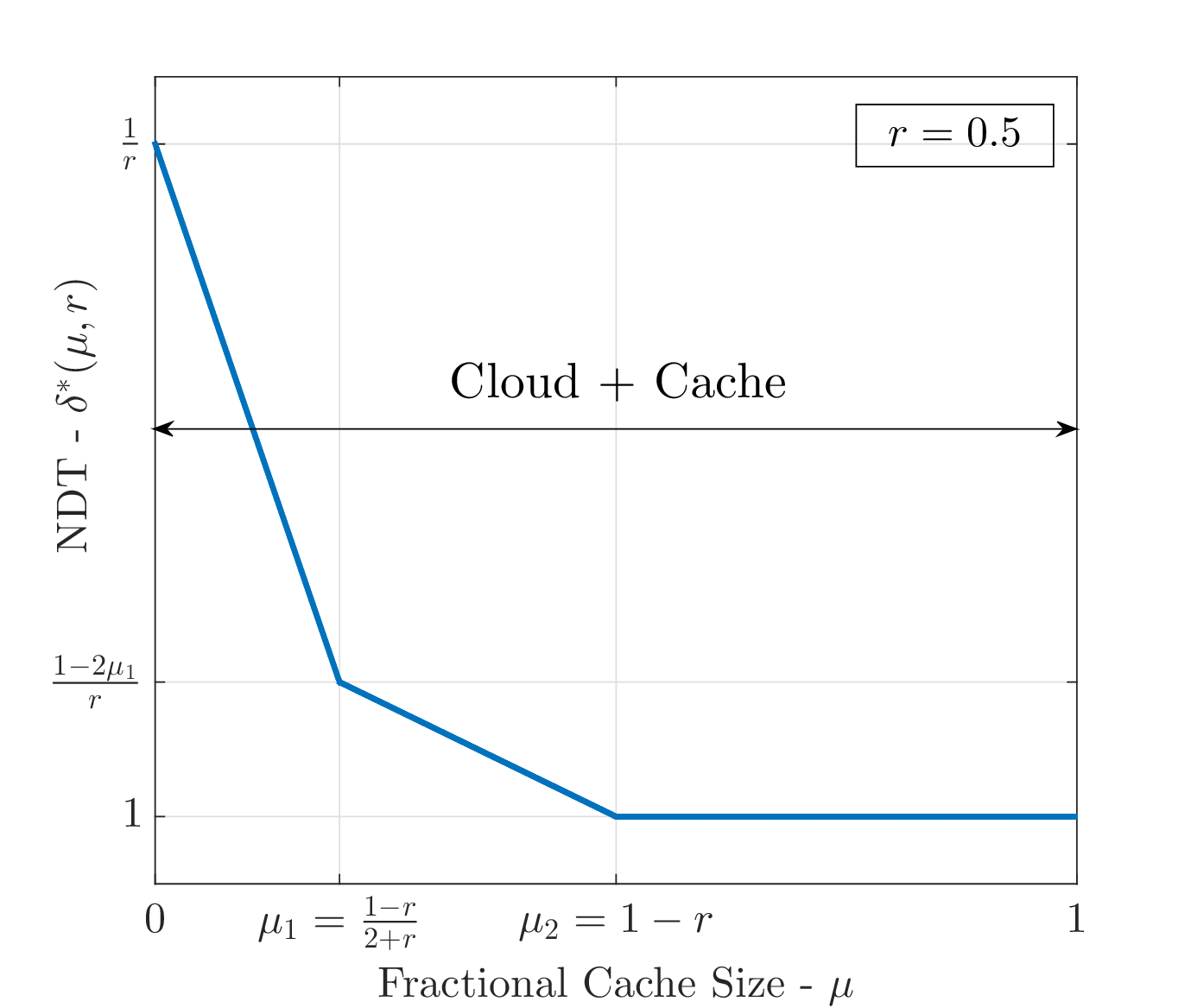}
\vspace{-5pt}
\protect\caption{Minimum NDT for an F-RAN with $M=K=2$ and pipelined fronthaul-edge transmissions in the low fronthaul regime, here with $r=0.5$.}
\label{fig:mk22_pipe}
\vspace{-15pt}
\end{figure}
\begin{Cor}\label{cor:mk22_pipe}
For an F-RAN with $M = 2$ ENs, $K = 2$ users, $N \geq 2$ files, and with pipelined fronthaul-edge transmissions, the minimum NDT is characterized as

\n $\bullet$~Low Fronthaul $(r (0,1])$:
\begin{align}
&\delta_{\mathsf{P}}^*(\mu,r) = \begin{cases}\displaystyle
										~\dfrac{1 - 2\mu}{r},~~~~~~~~~~\text{for} ~~ \mu \in \left[0, \mu_1 = (1-r)/(2+r) \right]\vspace{2.5pt}\\
										~\dfrac{2- \mu}{1 + r}, ~~~~~~~~~~~\text{for} ~~ \mu \in \left[\mu_1,\mu_2 = (1-r)\right]\vspace{2.5pt}\\
										~1, ~~~~~~~~~~~~~~~~~\text{for} ~~ \mu \in \left[\mu_2, 1\right]
									 \end{cases}
\end{align}

\n $\bullet$~High Fronthaul $(r\geq 1)$:
\begin{align}
&\delta_{\mathsf{P}}^*(\mu,r) = 1, ~~~~~\text{for} ~~ \mu \in [0,1].
\end{align}
\end{Cor}
\begin{proof}
The proof of Corollary \ref{cor:mk22_pipe} is provided in Appendix \ref{ap_ssec:mk22_pipe}.
\end{proof}

The minimum NDT for a $2\times 2$ F-RAN is shown in Fig. \ref{fig:mk22_pipe} in the regime of low fronthaul rate, here with $r= 0.5$. The optimal strategy uses block-Markov encoding with cloud-aided soft transfer fronthaul in conjunction with cache-aided EN cooperation or coordination as for Proposition \ref{prop:ach_pipe}. We observe that, in contrast to serial fronthaul-edge transmission (see Corollary \ref{cor:mk22}), the optimal strategy leverages cloud resources for any given fronthaul rate $r>0$. Furthermore, in line with the discussion in Remark \ref{rem:cloud_pipe}, by using cloud resources, it is possible here to obtain the minimum NDT $\delta^*_{\mathsf{P}}(\mu,r) = 1$ for all $\mu\geq \mu_2$ as well as for $r\geq 1$.

\section{Generalizations and Open Problems}\label{sec:open}
In this section, we discuss some of the open problems and directions for future work on the topic of cloud and cache-aided content delivery in F-RAN architectures.

\subsection{Is an Equal Cache Allocation Optimal?}\label{ssec:unequalcache}
Throughout the paper, as per Definition \ref{def:pol}, we have assumed that each file $F_n$ is cached with the same maximum number of bits, namely $\mu L$, at each EN. Here, we aim at understanding if the minimum NDT could be potentially reduced by allocating a different number of bits to each file at the ENs under the relaxed constraints 
\begin{align}
& \sum_{m=1}^M H\left(S_{m,n} \right) \leq M \mu L, ~~\forall n \in [1:N],\label{eq:relM}
\end{align}
and
\begin{align}
& \sum_{n=1}^N H\left(S_{m,n} \right) \leq N \mu L, ~~\forall m \in [1:M]\label{eq:relN},
\end{align}
where the first constraints impose the per-file condition that the overall number of bits used to cache file $F_n$ across all ENs cannot exceed $M \mu L$ bits, while the second is the per-EN cache capacity constraint. By using the same arguments as in the proof of Proposition \ref{prop:glb} in Appendix \ref{ap:thglb}, it can be shown that the lower bound in Proposition \ref{prop:glb} holds also under the relaxed constraints \eqref{eq:relM}-\eqref{eq:relN}. To this end, we first note that the bounds in \eqref{eq:c2} remain unchanged since they do not make use of the cache constraints. For the proof of \eqref{eq:c1}, we refer to Appendix \ref{ap:thglb}. This argument shows that the strategy of allocating an equal number of bits to each file at every EN as in Definition \ref{def:pol} is in fact information-theoretically optimal under the assumption of uncoded cache placement.

\subsection{Caching with Inter-File Coding}
The results presented in this paper are developed under the assumption that the caching policy at the ENs do not allow for inter-file coding (as in \eqref{eq:cache}). For schemes with caching only at receivers, it has been shown that coding across files during cache placement has the potential to improve the system performance \cite{Maddah-Ali,improve_fund,aviksg-isit15,ghasemi}. To elaborate on the potential gains of inter-file coding for F-RANs, we observe that, under such more general caching policies, the joint entropy of the cache contents of each EN is generally bounded as $H(S_{m}) \leq \mu N L$, for all $m\in [1:M]$, instead of as $H(S_{m,n}) \leq \mu L$ for all $m\in[1:M]$ and $n\in[1:N]$ as is the case without inter-file coding. As a result, following the proof of Proposition \ref{prop:glb} in Appendix \ref{ap:thglb}, we can see that, the constraint \eqref{eq:c1} in Proposition \ref{prop:glb} is modified as
\begin{align}\label{eq:coding1}
 \ell \delta_E + (M-\ell) r \delta_F \geq K -(M-\ell)N\mu,
\end{align}
which yields strictly lower bounds on the minimum NDT. Whether this lower bound is achievable by caching strategies with inter-file coding remains an open problem.

\subsection{Imperfect CSI}\label{ssec:delcsi}
Another aspect that is left open by this study is the impact of imperfect CSI availability at the ENs on the minimum NDT. To elaborate on this point, we consider a cache-only F-RAN and we assume that within a transmission interval $i$, the channel $\mathbf{H}_i^t$ varies across every channel use $t$ according to an i.i.d. process. When CSI is delayed, at any time $t$ on the $i$th transmission interval, the ENs only have access to the CSI of the previous $t-1$ channel uses, namely $\mathbf{H}_i^1,\mathbf{H}_i^2,\ldots,\mathbf{H}_i^{t-1}$. 

For illustration, we consider a cache-only F-RAN with $M=K=2$ and $N\geq 2$ with $\mu\in[1/2,1]$. For the case of perfect CSI, the minimum NDT can be characterized as in Corollary \ref{cor:mk22}. Next, we elaborate on the achievable NDT results for the case of delayed and no CSI.
\begin{figure}[!t]
\begin{centering}
\includegraphics[width=3.85in,height=3.4in]{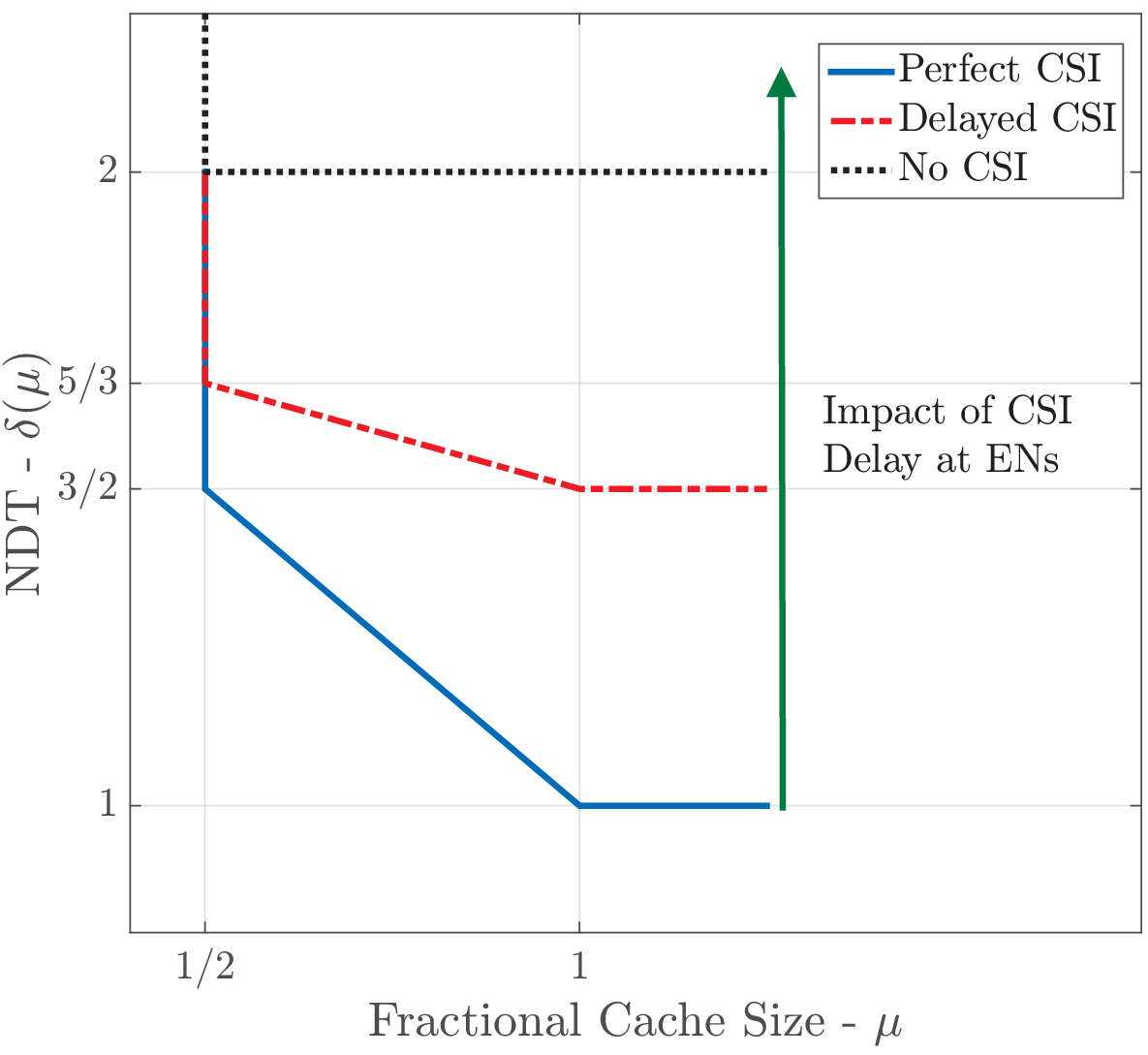}
\par\end{centering}
\protect\caption{ Effect of delayed or no CSI on the NDT for $M=K=2$.}
\label{fig:delayed}\vspace{-10pt}
\end{figure}

\subsubsection{\textit{ Delayed CSI at ENs}} For the case of delayed CSI, consider the corner point $\mu=1/2$ where the system behaves like a $2\times 2$ X-channel. The maximum known sum-DoF for the $2\times 2$ X-channel with delayed CSI is $6/5$ \cite{delayed_csit_Xch}. As a result, an NDT of $\delta_{\mathsf{Ach}}(\mu,0) = 5/3$ is achievable by Remark \ref{rem:dof}. Compared to the perfect CSI case, for which the NDT is $3/2$, this achievable NDT thus incurs a loss due to delayed CSI. Next, consider the corner point $\mu=1$, where the system reduces to a $2\times 2$ broadcast channel with delayed CSI. The maximum sum-DoF for such a system is $4/3$ \cite{delayed_csit_BC}, i.e., a NDT of $\delta_{\mathsf{Ach}}(\mu,0) = 3/2$ is achievable, which is larger than the NDT of $1$ with perfect CSI. 

\subsubsection{\textit{No CSI at ENs}} In case of no CSI, it is known that the optimal strategy on the edge channel is to transmit using time-division to each user in a separate slot \cite{Vaze_BC_Dof}. Therefore a sum-DoF of $1$ i.e., an NDT of $2$ can be achieved, which is hence optimal for all values of $\mu\in [1/2,1]$ and shows a significant loss as compared to the cases with full or delayed CSI as shown in Fig. \ref{fig:delayed}.

Quantifying the impact of delayed CSI on a general $M\times K$ F-RAN with cloud and cache-aided delivery, as considered in this work, remains an area of future work.

\section{Conclusions}\label{sec:conc}
In this paper, we presented a latency-centric study of the fundamental information-theoretic limits of cloud and cache-aided wireless networks, which we referred to as fog radio access networks (F-RANs). To this end, we introduced a new metric, namely the normalized delivery time (NDT), which measures the worst-case end-to-end latency required to deliver requested content to the end users in the high-SNR regime. We developed a converse result for the NDT of a general F-RAN with arbitrary number of ENs and users and then presented achievable schemes which leverage both cache and cloud resources. We characterized the minimum NDT for cloud-only F-RANs for all problem parameters; and for cache-only F-RANs in the regime of extremal values of the fractional cache size. Furthermore, we showed that the proposed achievable schemes are approximately optimal to within a constant factor of $2$ for all parameter values for the general F-RAN with fronthaul and edge-caching. We elaborated on two case studies, consisting of F-RANs with two or three ENs and users and (partially) characterized the NDT for these systems using the proposed upper and lower bounds. We also considered an alternative F-RAN model with pipelined fronthaul-edge transmissions. We presented a general lower bound on the NDT and proposed achievable schemes which are shown to be approximately optimal in terms of NDT to within a constant factor of $2$. Open problems were finally presented to highlight the richness of the problem introduced in this paper.

\appendices

\section{Proof of Proposition \ref{prop:glb}}\label{ap:thglb}
In this section, we present a detailed proof of Proposition \ref{prop:glb}. To obtain a lower bound on the NDT, we fix a specific request vector $\mathbf{D}$, namely one for which all requested files $(F_1,...,F_K)=F_{[1:K]}$ are different, and a given channel realization $\mathbf{H}$. Note that this is possible given the assumption $N\geq K$. We denote as $T_F$ and $T_E$ the fronthaul and edge transmission latencies, as per Definition \ref{def:pol} for any given feasible policy $\pi = (\pi_c,\pi_f,\pi_e,\pi_d)$ which guarantees a vanishing probability of error $P_e$ as $L\rightarrow\infty$ for the given request $\mathbf{D}$, channel $\mathbf{H}$ and fronthaul rate $C_F = r\log(P)$. Our goal is to obtain a lower bound on the minimum NDT $\delta^*(\mu,r)$ for any $r\geq 0$. To this end, consider the fronthaul messages $\mathbf{U}^{T_F}_m$ which are  $1\times T_F$ row vectors and the corresponding channel outputs in \eqref{eq:chop}, where $\mathbf{Y}_k^{T_E}, \mathbf{X}_m^{T_E}$ and $\mathbf{n}_k^{T_E}$ are $1\times T_E$ row vectors.

For ease of exposition, we next introduce the following notation which we use throughout the appendix. For any integer pair $(a,b)$ with $ a\leq b\leq K$, let $\mathbf{Y}^{T_E}_{[a:b]}$ be the $(b-a+1)\times T$ matrix of channel outputs of a subset $[a:b]$, of receivers. The notation is also used for the channel inputs $\mathbf{X}^{T_E}$ and noise $\mathbf{n}^{T_E}$. Furthermore, for any integers $1\leq a\leq b \leq K$ and $1\leq c\leq d\leq M$, we define the following sub-matrix of the channel matrix $\mathbf{H}$:
\begin{align*}
\mathbf{H}_{[a:b]}^{[c:d]} = \begin{bmatrix} h_{a,c} & h_{a,c+1} & \cdots & h_{a,d} \\
																						 h_{a+1,c} & h_{a+1,c+1} & \cdots & h_{a+1,d} \\
																						 \vdots   & \vdots & \ddots & \vdots \\
																						 h_{b,c} & h_{a,c+1} & \cdots & h_{b,d}
																						\end{bmatrix}.
\end{align*}

\n Using this notation, we can represent the channel outputs at all $K$ receivers as
\begin{align}\label{eq:chop2}
&\mathbf{Y}_{[1:K]}^{T_E} = \mathbf{H}_{[1:K]}^{[1:M]}~\mathbf{X}_{[1:M]}^{T_E} +  \mathbf{n}_{[1:K]}^{T_E},
\end{align}

To obtain the constraint \eqref{eq:c1}, we make the following key observation. Given any set of $\ell\leq \min\{M,K\}$ output signals $\mathbf{Y}^{T_E}_k$, say $\mathbf{Y}^{T_E}_{[1:\ell]}$, and the content of any $(M-\ell)$ caches, say $S_{[1:(M-\ell)]}$ and their corresponding fronthaul messages $\mathbf{U}^{T_F}_{[1:(M-\ell)]}$, all transmitted signals $\mathbf{X}^{T_E}_{[1:M]}$, and hence also all the files $F_{[1:K]}$, can be resolved in the high-SNR regime. This is because: $(i)$ from the cache contents $S_{[1:(M-\ell)]}$ and fronthaul messages $\mathbf{U}^{T_F}_{[1:(M-\ell)]}$, one can reconstruct the corresponding channel inputs $\mathbf{X}^{T_E}_{[1:(M-\ell)]}$; $(ii)$ neglecting the noise in the high-SNR regime, the relationship between the variables $\mathbf{Y}^{T_E}_{[1:\ell]}$ and the remaining inputs $\mathbf{X}^{T_E}_{[(M-\ell):M]}$ is given almost surely by an invertible linear system as in \eqref{eq:chop}. We use this argument in the following:\vspace{-15pt}

\begin{align}\label{eq:p1}
    &KL = H\left(F_{[1:K]}\right) \nonumber\\
	  &\myeq{(a)} H\left(F_{[1:K]}| F_{[K+1:N]}\right)\nonumber\\
    &= I\left(F_{[1:K]}; \mathbf{Y}^{T_E}_{[1:\ell]},\mathbf{U}^{T_F}_{[1:(M-\ell)]},S_{[1:(M-\ell)]}|F_{[K+1:N]}\right) + H\left(F_{[1:K]}|\mathbf{Y}^{T_E}_{[1:\ell]},\mathbf{U}^{T_F}_{[1:(M-\ell)]},S_{[1:(M-\ell)]},F_{[K+1:N]} \right)
\end{align}
where step \textsf{(a)} follows from the fact that all files are independent of each other. The first term in \eqref{eq:p1} can be upper bounded as follows:
\begin{align}\label{eq:A}
\hspace{-25pt} &I\left(F_{[1:K]}; \mathbf{Y}^{T_E}_{[1:\ell]},\mathbf{U}^{T_F}_{[1:(M-\ell)]},S_{[1:(M-\ell)]}|F_{[K+1:N]}\right) \nonumber\\
& =I\left(F_{[1:K]}; \mathbf{Y}^{T_E}_{[1:\ell]} |F_{[K+1:N]}\right) + I\left(F_{[1:K]}; \mathbf{U}^{T_F}_{[1:(M-\ell)]},S_{[1:(M-\ell)]}|\mathbf{Y}^{T_E}_{[1:\ell]},F_{[K+1:N]}\right) \nonumber\\
&\leq I\left(F_{[1:K]}; \mathbf{Y}^{T_E}_{[1:\ell]} |F_{[K+1:N]}\right) + I\left(F_{[1:K]}; \mathbf{U}^{T_F}_{[1:(M-\ell)]},S_{[1:(M-\ell)]},F_{[1:\ell]}|\mathbf{Y}^{T_E}_{[1:\ell]},F_{[K+1:N]}\right) \nonumber\\
& = I\left(F_{[1:K]}; \mathbf{Y}^{T_E}_{[1:\ell]} |F_{[K+1:N]}\right) + I\left(F_{[1:K]}; F_{[1:\ell]}|\mathbf{Y}^{T_E}_{[1:\ell]},F_{[K+1:N]}\right) \nonumber\\
&~~~~~~+ I\left(F_{[1:K]}; \mathbf{U}^{T_F}_{[1:(M-\ell)]},S_{[1:(M-\ell)]}|\mathbf{Y}^{T_E}_{[1:\ell]},F_{[1:\ell]\cup[K+1:N]}\right) \nonumber\\
&\myleq{(a)} I\left(F_{[1:K]}; \mathbf{Y}^{T_E}_{[1:\ell]} |F_{[K+1:N]}\right) + H\left( F_{[1:\ell]}|\mathbf{Y}^{T_E}_{[1:\ell]} \right)  \nonumber\\
&~~~~~~+ H\left(\mathbf{U}^{T_F}_{[1:(M-\ell)]},S_{[1:(M-\ell)]}|\mathbf{Y}^{T_E}_{[1:\ell]},F_{[1:\ell]\cup[K+1:N]}\right) - H\left(\mathbf{U}^{T_F}_{[1:(M-\ell)]},S_{[1:(M-\ell)]}|\mathbf{Y}^{T_E}_{[1:\ell]},F_{[1:N]}\right)\nonumber\\
& \myleq{(b)} h\left(\mathbf{Y}^{T_E}_{[1:\ell]}\right) - h\left(\mathbf{Y}^{T_E}_{[1:\ell]}|F_{[1:N]}  \right) + L\epsilon_L  \nonumber\\
&~~~~~~+ H\left(\mathbf{U}^{T_F}_{[1:(M-\ell)]},S_{[1:(M-\ell)]}|F_{[1:\ell]\cup[K+1:N]}\right) - H\left(\mathbf{U}^{T_F}_{[1:(M-\ell)]},S_{[1:(M-\ell)]}|\mathbf{Y}^{T_E}_{[1:\ell]},F_{[1:N]}\right) \nonumber\\
&\myleq{(c)} \ell {T_E}\log \Big(2\pi e \big(\Lambda P + 1\big)\Big)  - h\left(\mathbf{n}^{T_E}_{[1:\ell]}| F_{[1:N]}\right) + L\epsilon_L + H\left(\mathbf{U}^{T_F}_{[1:(M-\ell)]} \right) +\sum_{i=1}^{\Scale[0.6]{(M-\ell)}} H\Big(S_{i,[1:N]}|F_{[1:\ell]},F_{[K+1:N]}\Big) \nonumber\\
&\myleq{(d)} \ell {T_E}\log \big(\Lambda P + 1\big) + L\epsilon_L + (M-\ell)(K-\ell) \mu L + (M-\ell) rT_F\log(P) ,
\end{align}
\n where, the steps in \eqref{eq:A} are explained as follows:
\begin{itemize}
\item Step \textsf{(a)} follows from careful expansion of the second term in the previous step and that conditioning reduces entropy.
\item Step \textsf{(b)} follows from the fact that $\mathbf{Y}^{T_E}_{[1:\ell]}$ are continuous random variables and that dropping the conditioning in the first term increases entropy. We apply Fano's inequality to the second term where $\epsilon_L$ is a function, independent of $P$, which vanishes as $L\rightarrow \infty$.
\item Step \textsf{(c)} can be explained as follows. The first term is upper bounded by the use of Lemma \ref{lem:1} detailed in Appendix \ref{ap:lemma}. The parameter $\Lambda$ is a constant dependent only on the channel parameters. The last term is zero since the cache contents $S_{[1:(M-\ell)]}$ and fronthaul messages $\mathbf{U}^{T_F}_{[1:(M-\ell)]}$ are functions of the library of files $F_{[1:N]}$. Moreover, given all the files, the channel outputs are a function of the channel noise at each receiver.
\item Step \textsf{(d)} follows from the fact that the channel noise is i.i.d. across time and distributed as $\mathcal{N}(0,1)$.
\end{itemize}
Next, the second term in \eqref{eq:p1} can be upper bounded by use of Lemma \ref{lem:termB} as follows:
\begin{align}\label{eq:termB}
&H\left(F_{[1:K]}|\mathbf{Y}^{T_E}_{[1:\ell]},\mathbf{U}^{T_F}_{[1:(M-\ell)]},S_{[1:(M-\ell)]},F_{[K+1:N]} \right) \leq~~ L\epsilon_L + T_E\log\det \left( \mathbf{I}_{\Scale[0.7]{[K-\ell]}} + \tilde{\mathbf{H}}\tilde{\mathbf{H}}^H\right),
\end{align}
where $\epsilon_L$ is a function, independent of $P$ and vanishes as $L\rightarrow \infty$. Furthermore, the term $\log\det \left( \mathbf{I}_{\Scale[0.7]{[K-\ell]}} + \tilde{\mathbf{H}}\tilde{\mathbf{H}}^H\right)$ is independent of signal power $P$ and file size $L$ and is dependent only on the noise variance and the channel coefficients. The proof of \eqref{eq:termB} follows from Lemma \ref{lem:termB} which is detailed in Appendix \ref{ap:lemma}. Substituting \eqref{eq:A} and \eqref{eq:termB} into \eqref{eq:p1}, we have
\begin{align}\label{eq:p2}
KL &\leq \ell {T_E}\log \left( \Lambda P+1\right) + (M-\ell)(K-\ell) \mu L + (M-\ell) rT_F\log(P)  \nonumber\\
	 &~~~~~~~~~~~~~~~~+ L\epsilon_L + {T_E}\log\det \left( \mathbf{I}_{\Scale[0.7]{[K-\ell]}} + \tilde{\mathbf{H}}\tilde{\mathbf{H}}^H\right).
\end{align}

\n Rearranging \eqref{eq:p2}, we get the following
\begin{align}\label{eq:p3}
& \ell\delta_E \left[1 + \frac{ \ell\log\left(\Lambda + \frac{1}{P}\right) +  \log\det \left( \mathbf{I}_{\Scale[0.7]{[K-\ell]}} + \tilde{\mathbf{H}}\tilde{\mathbf{H}}^H\right) }{\ell\log(P)} \right]  + (M-\ell) r \delta_F \geq {K - (M-\ell)(K-\ell) \mu  - \epsilon_L}.
\end{align}

Now, using \eqref{eq:p3}, we first take the limit of ${L\rightarrow\infty}$ such that $\epsilon_L \rightarrow 0$  as $P_e \rightarrow 0$. Further, taking the limit ${P\rightarrow \infty}$, for the high-SNR regime, we arrive at \eqref{eq:c1}:
\begin{align}\label{eq:proof_c1}
\ell \delta_E + (M-\ell) r \delta_F \geq K -(M-\ell) (K-\ell) \mu,
\end{align}
where the multiplier of $\ell \delta_E$ converges to $1$ under the limit $P\rightarrow \infty$.

Next, we prove that the constraint $\delta_E\geq 1$ in \eqref{eq:c2} holds under the decodability constraint for file delivery to all users irrespective of the value of the fronthaul rate $r$. To this end, without loss of generality, we consider that the users $[1:K]$ demand the first $K$ distinct files $F_{[1:K]}$ i.e., $\mathbf{D} = (d_1,d_2,\ldots,d_K) = (1,2,\ldots,K)$. Next, consider the following set of inequalities:
\begin{align}\label{eq:proof_c2}
KL &= H\left(F_{[1:K]}\right) \nonumber\\
	 &= I\left(F_{[1:K]};\mathbf{Y}^{T_E}_{[1:K]}\right) + H\left(F_{[1:K]}|\mathbf{Y}^{T_E}_{[1:K]}\right)\nonumber\\
	 & \myleq{(a)} h\left(\mathbf{Y}^{T_E}_{[1:K]}\right) - h\left(\mathbf{Y}^{T_E}_{[1:K]}|F_{[1:K]}\right) + L\epsilon_L\nonumber\\
	 & = h\left(\mathbf{Y}^{T_E}_{[1:K]}\right) - h\left(\mathbf{n}^{T_E}_{[1:K]}\right)+ L\epsilon_L\nonumber\\
	 & \myleq{(b)} KT_E\log \left( \Lambda P+1\right) + L\epsilon_L,
\end{align}	
where step \textsf{(a)} follows from a Fano's Inequality and the fact that all requested files should be decoded by the received signals. Step \textsf{(b)} follows from the use of Lemma \ref{lem:1} (see Appendix \ref{ap:lemma}). Again taking the limits $P\rightarrow\infty$ and $L\rightarrow \infty$, and rearranging, we arrive at the constraint $\delta_E\geq 1$. Note that, by substituting $\ell=M$ in \eqref{eq:c1}, we get the following lower bound on the edge latency: \[\delta_E \geq K/M, ~~~\forall~K,M. \] This bound is tighter for $M\leq K$, while the constraint $\delta_E\geq 1$, proved here, supersedes the bound for the case when $M\geq K$.
Using constraints \eqref{eq:c1}-\eqref{eq:c2} to minimize the sum-latency, i.e., using linear combinations of the family of constraints in \eqref{eq:c1} and \eqref{eq:c2} over all possible choices of $\ell\in [0:\min\{M,K\}]$, gives the family of lower bounds for the $M\times K$ cache-aided F-RAN.

We conclude this section by addressing the scenario discussed in Section \ref{ssec:unequalcache} in which the relaxed cache placement constraints \eqref{eq:relM}-\eqref{eq:relN} are imposed. To prove \eqref{eq:c1}, under the relaxed constraints, we consider all possible \textit{sets} of $(M-\ell)$ ENs and follow steps similar to \eqref{eq:p1}-\eqref{eq:A}. Considering the step \textsf{(c)} in \eqref{eq:A} and using the different sets of $(M-\ell)$ ENs to decode the files, we will obtain ${M\choose (M-\ell)}$ different inequalities of this form. Summing and symmetrizing over all the ${M\choose (M-\ell)}$ inequalities and using the constraint in \eqref{eq:relM} to upper bound the overall number of bits required to store $(K-\ell)$ files across the $M$ ENs yields a bound which is identical to \eqref{eq:c1}. This shows that the strategy of allocating an equal number of bits to each file at every EN as in Definition \ref{def:pol} is in fact information-theoretically optimal under the assumption of uncoded cache placement.

\section{Proof of Proposition \ref{prop:stf}}\label{ap:soft}

In order to prove Proposition \ref{prop:stf}, we first discuss the NDT performance of a scheme that uses fronthaul and wireless channels in the standard fashion that is adopted, for instance, in the CPRI fronthaul interface in C-RANs \cite{cpri,SimeoneMPSY15}. In this scheme, the cloud quantizes the encoded baseband samples, and all the ENs \textit{simultaneously} transmit the quantized baseband signals. We then generalize this policy by allowing for a more general transmission schedule in which different \textit{clusters} of ENs can transmit on the wireless channel at distinct time intervals as introduced in Section \ref{ssec:cloud} (cf. Fig. \ref{fig:stf}). The proof is divided into two parts as follows.
\subsection{Standard Soft-Transfer Fronthauling}
Here, we prove that an NDT equal to
\begin{align}\label{eq:softNDT}
\delta(\mu,r) = \frac{K}{\min\{M,K\}}\left(1 + \frac{1}{r}\right),
\end{align}
is achievable by means of standard serial soft-transfer fronthauling for any fractional cache size $\mu\geq 0$ and for any fronthaul rate $r\geq 0$. To interpret \eqref{eq:softNDT}, we note that the NDT $\delta(\mu,r) = K/\min\{M,K\}$ can be achieved by means of ZF-beamforming in an ideal system in which there is either full caching, i.e., $\mu = 1$, or no fronthaul capacity limitations, i.e., $r \rightarrow \infty$. In fact, in such systems, full cooperation is possible at the ENs for any users' demand vector, including the worst case in which users request distinct files, and hence transmission at the maximum per-user multiplexing gain $\min\{M,K\}/K$ can be attained. The achievable NDT \eqref{eq:softNDT} hence shows a multiplicative penalty term equal to $1 + 1/r$ due to fronthaul capacity limitations. 

The proof of \eqref{eq:softNDT} relies on the use of the fronthaul and transmission policies introduced in Example \ref{ex:stf}. Note that caching is not used, in accordance with the assumption that $\mu$ may be zero. The cloud encodes the signals using ZF
beamforming under a power constraints smaller than $P$ that will be specified below. The resulting baseband signals are quantized and sent to the ENs on the fronthaul links. The ENs transmit simultaneously the respective received quantized samples on the wireless channel. Reception at the users is affected by the fronthaul quantization noise, as well as by the channel noise. If the quantization rate is properly chosen, it can be proved that the achievable NDT is \eqref{eq:softNDT}, where the term $K/\min\{M,K\}$ is the edge-NDT in \eqref{eq:dele}, which is the same as for the ideal ZF scheme, and the term $K/(r \min\{M,K\})$ is the fronthaul-NDT \eqref{eq:delf}. A more detailed discussion is provided next.

In the cloud-based scheme under study, the cloud performs ZF precoding, producing signal $\bar{X}_i$ for each $\text{EN}_i$ with power constraint $\bar{P} = E[|\bar{X}_i|^2]$. The signal $\bar{X}_i$ is quantized to obtain the signal $X_i$ that is to be transmitted by $\text{EN}_i$ as
\begin{align}\label{eq:baseband}
X_i = \bar{X}_i + Z_i,
\end{align}
where $Z_i \sim \mathcal{CN}(0,\sigma^2)$ represents the quantization noise with zero mean and variance $\sigma^2$. In order to satisfy the power constraint $P$, we enforce the condition
\begin{align}\label{eq:power}
P = \bar{P} + \sigma^2.
\end{align}
Furthermore, let $B$ denote the number of bits used for each baseband signal sample on the fronthaul link. From rate-distortion arguments \cite{cover} and using \eqref{eq:baseband}, we obtain the condition
\begin{align}\label{eq:ratedist}
I\left(X_i;\bar{X}_i\right) &~~=~~ \log_2\left(1 + \frac{\bar{P}}{\sigma^2}\right) = B\nonumber\\
\text{i.e.,}~~~~~~~~~~~~ \sigma^2 &~~=~~ \frac{\bar{P}}{2^B - 1}.
\end{align}
Therefore, from \eqref{eq:power} and \eqref{eq:ratedist}, we obtain the power constraint on the precoded signal as
\begin{align}\label{eq:powerconst}
\bar{P} = P\left(1 - 2^{-B}\right),
\end{align}
and the quantization noise power as
\begin{align}\label{eq:quantnoise}
\sigma^2 = 2^{-B} P.
\end{align}
The quantization noise terms $Z_i$ for all ENs $i \in [1:M]$, contribute to raising the noise level at each user. In particular, for any user $k\in[1:K]$, the power of the effective noise on the received signals in \eqref{eq:chop} is given by
\begin{align}
1 + \sigma^2 \sum_{m=1}^{M} |h_{km}|^2 = 1 + \sigma^2 G,
\end{align}
where $G = \sum_{m=1}^{M} |h_{km}|^2$. Normalizing the received signal so that the variance of the effective noise is $1$, using \eqref{eq:powerconst} and \eqref{eq:quantnoise}, we obtain an equivalent signal model in which the effective power constraint is
\begin{align}\label{eq:effpow}
\frac{\bar{P}}{1 + \sigma^2 G} = \frac{P(1 - 2^{-B})}{1 + 2^{-B}PG}.
\end{align}
Now, setting $B = \log(P)$, the effective power becomes $(P-1)/(1+G)$, which scales linearly with $P$.
Using the proposed soft-transfer fronthaul scheme, it follows that the fronthaul latency is given by 
\begin{align}\label{eq:front_lat1}
T_F = T_E\frac{B}{C_F},
\end{align} 
since $BT_E$ bits need to be sent on each fronthaul link at a rate of $C_F = r\log(P)$ to represent the quantized signals. It follows that the total latency of this scheme is
\begin{align}\label{eq:totlat}
T_E + T_F  = T_E \left(1 + \frac{B}{C_F}\right) \myeq{(a)} T_E\left(1 + \frac{1}{r}\right),
\end{align}
where $(\mathsf{a})$ follows from the choice of $B = \log(P)$. Furthermore, in the high-SNR regime we have the following limit:
\begin{align}\label{eq:te}
\lim_{P\rightarrow \infty}\lim_{L\rightarrow \infty} \frac{T_E \log\left((P-1)(1+G)\right)}{L} = \frac{K}{\min\{M,K\}},
\end{align}
due to achievability of the NDT $K/\min\{M,K\}$ in the ideal ZF system mentioned above and due to the effective noise power $(P-1)/(1+G)$ for the scheme at hand. We can thus conclude our proof by computing the NDT
\begin{align}\label{eq:ndtsoft}
\lim_{P\rightarrow \infty}\lim_{L\rightarrow \infty} \frac{(T_E + T_F)\log(P)}{L} &= \left(1 + \frac{1}{r}\right) \lim_{P\rightarrow \infty}\lim_{L\rightarrow \infty}  \frac{T_E \log(P)}{L}\nonumber\\
&= \left(1 + \frac{1}{r}\right) \frac{K}{\min\{M,K\}},
\end{align}
where the second equality follows due to \eqref{eq:te}.

\subsection{Soft-Transfer Fronthauling with Clustering}
Here, we prove that the following NDT is achievable by means of a generalized soft-transfer fronthaul scheme based on sequential scheduling of distinct clusters of ENs on the wireless channel, when the number $M$ of ENs is larger than the number $K$ of users. In particular we show that for any $M\geq K$, an NDT of
\begin{align}\label{eq:softNDT2}
\delta(\mu,r) = 1 + \frac{K}{Mr}
\end{align}
is achievable by means of soft-transfer fronthauling in conjunction with EN clustering with sequential scheduling for any fractional cache size $\mu \geq 0$ and any fronthaul rate $r\geq 0$.

We start by observing that, if $M\geq K$, the NDT for the ideal system with full caching or unlimited fronthaul is given by $\delta = 1$, which is achieved by ZF-beamforming. Comparing the NDT \eqref{eq:softNDT2} with \eqref{eq:softNDT}, and recalling the discussion in the previous subsection, we can conclude that clustering and sequential scheduling of ENs allows one to reduce the normalized latency associated with the fronthaul transmission from $1/r$ to $K/(Mr)$. We also emphasize that, unlike the NDT in \eqref{eq:softNDT}, which is based on standard C-RAN fronthauling, the improved NDT \eqref{eq:softNDT2} tends to the ideal NDT $\delta = 1$ when the number of transmit
antennas grows large. As detailed next, this is due to a novel use of the fronthaul in soft-transfer mode, whereby quantized baseband signals received at the same time on the fronthaul by different ENs can be scheduled at different times on the wireless channel.

We first present the proposed scheme for the case in which $M$ is a multiple of $K$, so that $M/K$ is an integer number, and then we generalize the strategy for any $M$. As explained in Section \ref{ssec:cloud}, the main idea is to partition the ENs into $M/K$ disjoint clusters of $K$ ENs and to schedule each cluster for a time equal to $T_E K/M$, that is, on one of $M/K$ equal time intervals dividing $T_E$. Note that the fact that $K$ ENs are active at any given time enables the use of ZF-beamforming for all time intervals on the wireless channel. In particular, we can use the same scheme based on fronthaul quantization presented in the previous subsection with a key caveat: each EN needs to receive only $T_EK/M$ baseband samples, and hence the fronthaul latency is 
\begin{align}\label{eq:front_lat2}
T_F = T_E\frac{BK }{M C_F} ,
\end{align}
i.e., the fronthaul latency is $M/K$ times smaller than the latency in \eqref{eq:front_lat1} for the scheme discussed in the previous section. Following the same reasoning as in \eqref{eq:te}-\eqref{eq:ndtsoft} concludes the proof of \eqref{eq:softNDT2} for the case of $M/K$ being an integer number.

We consider now, the more general case in which $M/K \geq 1$ may not be an integer. Here, we proceed by clustering the ENs into all possible ${M \choose K}$ subsets of $K$ ENs, and then scheduling each cluster into distinct time intervals of duration $T_E/{M \choose K}$. Note that, unlike the case with integer $M/K$, here the clusters of ENs overlap. The number of samples that each EN needs to receive on its fronthaul is equal to
\begin{align}
T_E {M-1 \choose K-1}/{M\choose K} = T_E \frac{K}{M},
\end{align}
since each EN participates in ${M-1 \choose K-1 }$ clusters and the fronthaul latency is again given by \eqref{eq:front_lat2}. Following the same arguments above leads to the NDT in \eqref{eq:softNDT2}. 
Finally, combining the fronthaul latency expressions in \eqref{eq:front_lat1} and \eqref{eq:front_lat2}, we have
\begin{align}\label{eq:front_lat}
T_F = T_E\frac{B\min\{M,K\}}{MC_F}.
\end{align}
 Using this and following the same arguments as in the previous cases leads to the NDT in \eqref{eq:stf} which completes the proof of Proposition \ref{prop:stf}.

\section{Proof of Proposition \ref{prop:corner}}\label{ap:corner}

To prove Proposition \ref{prop:corner}, we expound on the minimum NDT for the two extremal values of fractional cache size $\mu\in\{1/M,1\}$. For $\mu = 1/M$, we substitute $\ell = 1$ in \eqref{eq:cor-ca} to get
\begin{align}\label{eq:cp1lb}
\delta^*(1/M,0) \geq K - \frac{(M-1)(K-1)}{M} = \frac{M+K-1}{M}.
\end{align}
To obtain an upper bound on NDT, consider the cache-aided EN coordination scheme achieving the NDT $\delta_{\mathsf{Ca-IA}}$ given in \eqref{eq:NDT1bM} as discussed in Lemma \ref{lem:ach1M}. Thus, we have the upper bound
\begin{align}\label{eq:cp1ub}
& \delta^*(1/M,0) \leq \delta_{\mathsf{Ca-IA}} = \frac{M+K-1}{M}.
\end{align}
Combining \eqref{eq:cp1lb} and \eqref{eq:cp1ub} shows that the lower bound in Corollary \ref{cor:lb} is tight at $\mu=1/M$. Next, considering the NDT at $\mu=1$, substituting $\ell = \min\{M,K\}$ into \eqref{eq:cor-ca}, we get \vspace{-5pt}
\begin{align}\label{eq:cp2lb}
\delta^*(1,0) \geq \frac{K}{\min\{M,K\}}, ~~ \text{for} ~~ r = 0.
\end{align}
Again, when $\mu=1$, consider the cache-aided EN cooperation scheme leveraging ZF-beamforming achieving the NDT $\delta_{\mathsf{Ca-ZF}}$ given in \eqref{eq:NDT1} as discussed in Lemma \ref{lem:ach1}. Using this, we have the upper bound
\begin{align}\label{eq:cp2ub}
&~\delta^*(1,0) \leq \delta_{\mathsf{Ca-ZF}} = \frac{K}{\min\{M,K\}}.
\end{align}
Combining \eqref{eq:cp2lb} and \eqref{eq:cp2ub}, shows that the lower bound in Corollary \ref{cor:lb} is tight at $\mu=1$. This concludes the proof of Proposition \ref{prop:corner}.


\section{Proof of Proposition \ref{prop:gap}}\label{ap:gap}

In this section, we present the proof of the approximate optimality of the achievable schemes presented in Section \ref{sec:ub}. To this end, we consider two regimes for the fractional cache size $\mu$ namely low-cache regime with $\mu \in [0, 1/M]$ and high-cache regime with $\mu \in [1/M, 1]$. Next, we consider each of the two regimes separately.

\n \textbf{\textit{Low-Cache Regime} $\bm{\left(\mathit{\mu \in [0, 1/M]}\right)}$:} For the low cache size regime, we consider two different cases where $(i)$ the number of users exceeds the number of ENs, i.e., $M\leq K$; and $(ii)$ the number of ENs exceeds the number of users, i.e., $M \geq K$. Next, we treat each of the two cases separately.

\n $\bullet$~ \textbf{\textit{Case} $\bm{\mathit{1}\left(\mathit{M\leq K}\right)}$:}~
For the case when the number of users exceed the number of ENs, we consider two different subcases: $(i)$ a high fronthaul regime with $r\geq 1$; and $(ii)$ a low fronthaul regime with $r\in(0,1]$. We consider each of these regimes separately.

\underline{\textit{High Fronthaul Regime $\mathit{r\geq 1}$}:}
~In this regime, consider the achievable NDT in \eqref{eq:ndtub3}. We have
\begin{align}
\delta_{\mathsf{Ach}}(\mu,r)\leq ~\frac{K}{M} + (1-\mu)\frac{K}{Mr} \leq~ \frac{K}{M}\left(1 + \frac{1}{r}\right).
\end{align}
Consider the LP in Theorem \ref{prop:glb} and the fact that any lower bound on the solution of this LP is also a valid lower bound on the minimum NDT. Thus, substituting $\ell = M$ in constraint \eqref{eq:c1} and using the fact that $\delta_F\geq 0$, we have
\begin{align}
\delta^*(\mu,r) \geq \delta_E + \delta_F \geq \frac{K}{M}.
\end{align}
Thus, we have
\begin{align}\label{eq:gap_h_front}
\frac{\delta_{\mathsf{Ach}}(\mu,r)}{\delta^*(\mu,r)} \leq \left(1 + \frac{1}{r}\right) \leq 2,
\end{align}
for any fronthaul rate $r\geq 1$. Thus the proposed schemes are approximately optimal to within a factor of $2$ for any parameter values of $M,K$ in the high fronthaul regime.

\underline{\textit{Low Fronthaul Regime $\mathit{r\in (0, 1]}$}:}~
For the low fronthaul regime of $r\in(0,1]$, we first divide the fronthaul rate into multiple non-overlapping regimes based on the number of ENs $M$ in the F-RAN as shown in Fig. \ref{fig:front_reg}, which indicates that for $r\in(0,1/(M-1)]$, the result in Proposition \ref{prop:lowmu} characterizes the minimum NDT for any $M \leq K$. Therefore, we focus our attention to the remaining regimes of interest for which $ r \in (1/(M-1), 1]$. As illustrated in Fig. \ref{fig:front_reg}, we further sub-divide these intervals into the sub-intervals
\begin{align}\label{eq:r_reg}
r \in \left[\frac{\ell - 1}{\left(M - \ell + 1 \right)}~ ,~\frac{\ell}{\left(M - \ell\right)}\right],
\end{align}
indexed by $\ell \in [1:\min\{M,K\}]$. 


\begin{figure}[!t]
\centering
\includegraphics[width = 5.5in, height=1.65in]{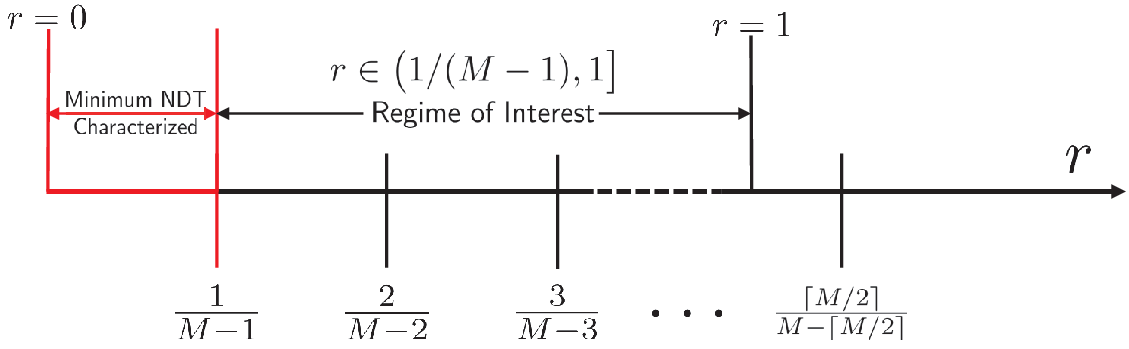}
\caption{Division of fronthaul rate $r\in(0,1]$ into parametrized regimes. }
\label{fig:front_reg}\vspace{-10pt}
\end{figure}
For each sub-interval indexed by $\ell-1$ and $\ell$, we obtain a lower bound on the minimum NDT by considering the constraints \eqref{eq:c1}, which are rewritten here as
\begin{align}
& \mathsf{Ineq~1}: ~~(\ell - 1)\delta_E + (M-\ell+1) r \delta_F  ~\geq~ K - (M - \ell + 1)(K - \ell + 1)\mu, \nonumber\\
& \mathsf{Ineq~2}: ~~~~~~~~~~~~~~\ell \delta_E + (M-\ell) r \delta_F ~\geq~ K - (M - \ell)(K - \ell)\mu. \nonumber
\end{align}
Specifically, we take a  a linear combination of the two inequalities \[\alpha \times \mathsf{Ineq~1} ~+~ \beta \times \mathsf{Ineq~2},\] which with $\alpha,\beta \geq 0$ to yield the following lower bound on the minimum NDT
\begin{align}\label{eq:lpsol_s1}
&\delta^*(\mu,r) \geq \alpha \big[K - (M - \ell + 1)(K - \ell + 1)\mu \big] + \beta\big[ K - (M - \ell)(K - \ell)\mu\big].
\end{align}
Choosing the weights $\alpha$ and $\beta$ as
\begin{align}
& \alpha ~=~ \frac{\ell}{M}\left(1 + \frac{1}{r}\right) - 1,~~~ \beta  ~=~ 1 - \frac{\ell-1}{M}\left(1 + \frac{1}{r}\right),
\end{align}
we have the following set of inequalities:
\begin{align}\label{eq:lpsol_lb}
\delta^*(\mu,r) &\geq \big[K - (M - \ell + 1)(K - \ell + 1)\mu \big] \left[\frac{\ell}{M}\left(1 + \frac{1}{r}\right) - 1\right]  \nonumber\\
								&~~~~~~~~~~~~~~~~~~~~~~~~~~~~~~~~~+ \big[ K - (M - \ell)(K - \ell)\mu\big]\left[1 - \frac{\ell-1}{M}\left(1 + \frac{1}{r}\right)\right]\nonumber\\
								&= \big[ K - (M - \ell)(K - \ell)\mu \big] \left[ \frac{1}{M}\left(1 + \frac{1}{r}\right)\right] - \mu\big(M+K+1 - 2\ell\big)\left[\frac{\ell}{M}\left(1 + \frac{1}{r}\right) - 1\right]\nonumber\\
								&= \left[ \frac{K - (M-\ell)(K - \ell)\mu}{M} - \frac{(M+K+1 - 2\ell)\ell\mu}{M} \right]\left(1 + \frac{1}{r}\right)  + \mu\big(M+K+1-2\ell \big) \nonumber\\
								&= \left[\frac{K}{M} - K\mu + \frac{\left(\ell^2 - \ell\right)\mu}{M} \right]\left(1 + \frac{1}{r}\right) + (M+K-1)\mu + 2\mu(1-\ell)\nonumber\\
								& = (M+K-1)\mu + \frac{K(1- \mu M)}{M}\left(1 + \frac{1}{r}\right) + \frac{\ell^2-\ell}{M}\left(1 + \frac{1}{r}\right)\mu + 2\mu(1-\ell).
\end{align}

An achievable NDT is obtained by considering \eqref{eq:ndtub1}, that is, the first term inside the $\min(\cdot)$ function in \eqref{eq:propub1} and substituting $\min\{M,K\}=M$, yielding
\begin{align}\label{eq:achub}
\delta_{\mathsf{Ach}}(\mu,r) &\leq(M+K-1)\mu + \frac{K(1- \mu M)}{M}\left(1 + \frac{1}{r}\right).
\end{align}
Now, from \eqref{eq:lpsol_lb} and \eqref{eq:achub} we have
\begin{align}\label{eq:gap_ub}
\frac{\delta_{\mathsf{Ach}}(\mu,r)}{\delta^*(\mu,r)} & \leq \frac{(M+K-1)\mu + \frac{K(1- \mu M)}{M}\left(1 + \frac{1}{r}\right) }{(M+K-1)\mu + \frac{K(1- \mu M)}{M}\left(1 + \frac{1}{r}\right) + \frac{\ell^2-\ell}{M}\left(1 + \frac{1}{r}\right)\mu + 2\mu(1-\ell)}\nonumber\\																						
																										 &\myeq{(a)} 1 + \frac{ -  \frac{\ell^2-\ell}{M}\left(1 + \frac{1}{r}\right)\mu - 2\mu(1-\ell)}{(M+K-1)\mu + \frac{K(1- \mu M)}{M}\left(1 + \frac{1}{r}\right) + \frac{\ell^2-\ell}{M}\left(1 + \frac{1}{r}\right)\mu + 2\mu(1-\ell)}\nonumber\\
																										 & \myleq{(b)} 1 + \frac{ 2\mu(\ell-1)}{(M+K-1)\mu + \frac{K(1- \mu M)}{M}\left(1 + \frac{1}{r}\right) + \frac{\ell^2-\ell}{M}\left(1 + \frac{1}{r}\right)\mu + 2\mu(1-\ell)}\nonumber\\																										
																										 &\myleq{(c)} 1 + \frac{ 2\mu(\ell-1)}{(M+K-1)\mu + \frac{\ell^2-\ell}{M}\left(1 + \frac{1}{r}\right)\mu + 2\mu(1-\ell)}\nonumber\\
																										 &\myleq{(d)} 1 + \frac{ 2(\ell-1)}{(M+K-1) + 2\left(\frac{\ell^2-\ell}{M}\right) + 2(1-\ell)}\nonumber\\
																										 & =  1 + \frac{ 2}{\frac{M+K-1}{\ell-1} + 2\left(\frac{\ell}{M} - 1\right)}\nonumber\\
																										 & \myleq{(e)} 1 + \frac{ 2}{\frac{2M-1}{\ell-1} + 2\left(\frac{\ell}{M} - 1\right)}\nonumber\\
																										 & \myleq{(f)} 1 + \frac{2}{\frac{4M - 2}{M-2} -1} = 1 + \frac{2(M-2)}{3M} \leq 1 + \frac{2}{3}< 2,
\end{align}
where step \textsf{(a)} follows by addition and subtraction of the term $\frac{\ell^2-\ell}{M}\left(1 + \frac{1}{r}\right)\mu + 2\mu(1-\ell)$ from the numerator; step \textsf{(b)} follows from the fact that the term $\frac{\ell^2-\ell}{M}\left(1 + \frac{1}{r}\right)\mu$ in the numerator of the second term is positive in the regime of interest; step \textsf{(c)} follows from the fact that the term $\frac{K(1- \mu M)}{M}\left(1 + \frac{1}{r}\right)$ in the denominator of the second term is positive and omitting it leads to an upper bound; step \textsf{(d)} follows by setting $r=1$ which, in turn, follows from the fact that we are interested in the low fronthaul regime with $r\leq 1$; step \textsf{(e)} follows from the fact that $M\leq K$ in the regime of interest and step \textsf{(f)} follows from the fact that for $r\in(0,1]$, we have $\ell\leq M/2$ and putting $\ell=M/2$ minimizes the denominator of the second term which is a decreasing function of $\ell$. 
We conclude that the maximum multiplicative gap between the achievable NDT and the minimum NDT is at most $2$ for all the sub-intervals in Fig. \ref{fig:front_reg}.


\n $\bullet$~ \textbf{\textit{Case} $\bm{\mathit{2} \left(\mathit{M\geq K}\right)}$:}~
For this case, we obtain a lower bound by considering the sum of the constraints in \eqref{eq:c1} with $\ell=0$ and \eqref{eq:c2}, yielding
\begin{align}\label{eq:lpsol2_s2}
\delta^*(\mu,r) \geq \delta_E + \delta_F = 1 + \frac{K\left(1 - \mu M \right)}{Mr} .
\end{align}

For an achievable NDT, we consider again, the first term inside the $\min(\cdot)$ function in \eqref{eq:propub1}, which for $K = \min\{M,K\}$ gives the following upper bound
\begin{align}\label{eq:achub2}
\delta_{\mathsf{Ach}}(\mu,r) \leq (M+K-1)\mu + (1- \mu M)\left(1 + \frac{K}{Mr}\right) = 1 + \frac{K\left(1 - \mu M\right)}{Mr} + (K-1)\mu.
\end{align}
Thus, from \eqref{eq:lpsol2_s2} and \eqref{eq:achub2}, we have
\begin{align}
\frac{\delta_{\mathsf{Ach}}(\mu,r)}{\delta^*(\mu,r)} & \leq  1 + \frac{(K-1)\mu}{1 + \frac{K\left(1 - \mu M \right)}{Mr}}~ \myleq{(a)} 1 + \frac{(K-1)}{M} ~~\leq 1 + \frac{K}{M} ~~\myleq{(b)} 2,
\end{align}
where step \textsf{(a)} follows from setting the fractional cache size $\mu = 1/M$ which is the maximum value it can assume in the low-cache memory regime at hand; and step \textsf{(b)} follows from the fact that $M \geq K$.  Next, we consider the regime of high cache i.e., $\mu \in [1/M,1]$. \vspace{5pt}

\n \textbf{\textit{High-Cache Regime} $\bm{\left( \mathit{\mu \in [1/M,1]}\right)}$:}
In this regime, we consider 
the achievable NDT in \eqref{eq:ndtub2}, i.e., the first term in \eqref{eq:propub2}, which, using $\mu = 1/M$ yields the upper bound
\begin{align}
\delta_{\mathsf{Ach}}(\mu,r) \leq \frac{M+K-1}{M}.
\end{align}
For the lower bounds, we first consider the case of $M\leq K$. From constraint \eqref{eq:c1} in Proposition \ref{prop:glb}, using $\ell = M$, we have $\delta_E\geq K/M$ which yields the lower bound on the minimum NDT
\begin{align}
\delta^*(\mu,r)\geq \frac{K}{M},
\end{align}
where we have used the fact that $\delta_F\geq 0$. Thus we have the desired gap
\begin{align}\label{eq:gapr2_1}
\frac{\delta_{\mathsf{Ach}}(\mu,r)}{\delta^*(\mu,r)} \leq \frac{(M+K-1)}{M}\frac{M}{K} \leq  1 + \frac{M}{K} \leq 2.
\end{align}
Next consider the case of $M\geq K$, From constraint \eqref{eq:c2} in Proposition \ref{prop:glb}, we have $\delta_E \geq 1$. Again, using the fact that that $\delta_F\geq 0$, we have
\begin{align}
\delta^*(\mu,r) \geq 1.
\end{align}
Thus we have the desired gap
\begin{align}\label{eq:gapr2_2}
\frac{\delta_{\mathsf{Ach}}(\mu,r)}{\delta^*(\mu,r)} \leq \frac{(M+K-1)}{M}\leq  1 + \frac{K}{M} \leq 2.
\end{align}
 This concludes the proof of Proposition \ref{prop:gap}.

\section{Proof of Proposition \ref{prop:lowmu}}\label{ap:lowmu}
The minimum NDT is first proved to be upper bounded by the right-hand side of \eqref{eq:lowmu} by substituting $M = \min\{M,K\}$ into the achievable rate in Proposition \ref{prop:ub} for the regime $\mu\in[0,1/M]$. 
For the matching lower bound, in the LP of Proposition \ref{prop:glb}, we substitute $\ell=1$ and $\ell = 0$ in \eqref{eq:c1}, yielding respectively:
\begin{align}
&\mathsf{Ineq~ 1}: ~ \delta_E + (M-1)r\delta_F \geq (M+K-1)\mu + K(1 - \mu M),\\
&\mathsf{Ineq~ 2}: ~ \delta_F \geq K(1 - \mu M)/Mr.
\end{align}
Since $r\in (0, 1/(M-1)]$, we obtain a lower bound by considering the linear combination $\mathsf{Ineq~ 1} + (1 - (M-1)r)\times \mathsf{Ineq~ 2}$, leading to the expression on the right-hand side of \eqref{eq:lowmu}.
This concludes the proof.

\section{Converse for Corollary \ref{cor:mk22}}\label{ap:mk22} 

We characterize the lower bounds for the $2\times 2$ F-RAN in order to show the optimality of the achievable schemes discussed in Section \ref{sec:ub}. We again consider each of the fronthaul regimes separately.

 {\textit{Cache-Only F-RAN} ($r=0$):} For the cache-only F-RAN, considering the lower bound from Corollary \ref{cor:lb} and using $\ell=1$, we get
\begin{align}
\delta^*(\mu,r) \geq 2 - \mu,
\end{align}
which is identical to the achievable NDT in \cite{RT_ISIT16}. Next, we consider the more general case when fronthaul is available i.e., $r>0$. To this end, we consider the LP in Proposition \ref{prop:glb}. The constraints of the LP can be rewritten as:
\begin{align}
\mathsf{Ineq~ 1}: &~~~~ (\delta_E + r\delta_F) \geq (2 - \mu)\\
\mathsf{Ineq~ 2}: &~~~~ \delta_F \geq (1 - 2\mu)/r\\
\mathsf{Ineq~ 3}: &~~~~ \delta_E \geq 1.
\end{align}
$\mathsf{Ineq~ 1}$ and $\mathsf{Ineq~ 2}$ are obtained from \eqref{eq:c1} by substituting $\ell=1$ and $\ell=0$ respectively, while $\mathsf{Ineq~ 3}$ follow directly from \eqref{eq:c2}. We next utilize these inequalities to prove the converse for different regimes of $r$.

{\textit{Low Fronthaul} $(r\in (0,1])$:} In this regime, using $\mathsf{Ineq~ 1} + (1-r)\times \mathsf{Ineq~ 2}$ gives the lower bound:
\begin{align}
\delta^*(\mu,r) \geq 1 + \mu + \frac{1-2\mu}{r}.
\end{align}
Substituting $r=1$ in $\mathsf{Ineq~ 1}$ gives the lower bound
\begin{align}
\delta^*(\mu,r) \geq 2-\mu.
\end{align}
Combining this with the upper bounds presented in \cite{RT_ISIT16} gives the minimum NDT for the low fronthaul regime as shown in Fig. \ref{fig:MK22_r0p5}.

{\textit{High Fronthaul} ($r\geq 1$):} For this regime, using $\mathsf{Ineq~ 1} + (r-1)\times \mathsf{Ineq~ 3}$, we have
\begin{align}
\delta^*(\mu,r)\geq 1 + \frac{1-\mu}{r}.
\end{align}
Combining with the upper bound presented in \cite{RT_ISIT16} gives the minimum NDT for the high fronthaul regime as shown in Fig. \ref{fig:MK22_r1p5}. Thus, the lower bound on the NDT in Proposition \ref{prop:glb} and the achievable scheme presented in Proposition \ref{prop:ub} completely characterizes the minimum NDT for the $2\times 2$ F-RAN.

\section{Lemmas used in Appendix \ref{ap:thglb}}\label{ap:lemma}
In this section, we state and prove the lemmas used in the proof of Proposition \ref{prop:glb}. First, we state and prove Lemma \ref{lem:1} which was used in \eqref{eq:A} in Appendix \ref{ap:thglb}.
\begin{lemma}\label{lem:1}
For the cloud and cache-aided wireless network under consideration, the differential entropy of any $\ell$ channel outputs $\mathbf{Y}^{T_E}_{[1:\ell]}$ can be upper bounded as
\begin{align}\label{eq:lem1}
h\left(\mathbf{Y}^{T_E}_{[1:\ell]}\right) \leq  \ell{T_E}\log \Big(2\pi e \left( \Lambda P + 1\right)\Big),
\end{align}
where the parameter $\Lambda$ is a function of the channel coefficients in $\mathbf{H}$ and is defined as
\begin{align*}
\Lambda = \left(\max\limits_{k\in [1:\ell]} \left[\sum_{m=1}^M h_{km}^2 + \sum_{m\neq \tilde{m}}h_{km}h_{k\tilde{m}}\right] \right).\vspace{-3pt}
\end{align*}
\end{lemma}

\begin{proof}
The entropy of the received signals $\mathbf{Y}^{T_E}_{[1:\ell]}$ can be upper bounded as follows:
\begin{align}\label{eq:mi_ub1}
 h\left(\mathbf{Y}_{[1:\ell]}^{T_E}\right) \leq \sum_{k=1}^\ell \sum_{t=1}^{T_E} h\Big({Y}_k[t]\Big).
\end{align}
Now, we upper bound the inner sum as follows:
\begin{align}\label{eq:sub1}
\sum_{t=1}^{T_E} h\Big({Y}_k[t]\Big)
&= \sum_{t=1}^{T_E} h\left( \sum_{m=1}^M h_{km}X_m[t] + n_k[t] \right) \nonumber\\
&\leq \sum_{t=1}^{T_E} \log \left(2\pi e ~\text{Var}\left[\sum_{m=1}^M h_{km}X_m[t] + n_k[t] \right]\right)\nonumber\\
& \myeq{(a)} \sum_{t=1}^{T_E} \log \left(2\pi e \left(\text{Var}\left[\sum_{m=1}^M h_{km}X_m[t]\right] + \text{Var}\left[n_k[t]\right] \right)\right)\nonumber\\
& \myeq{(b)} \sum_{t=1}^{T_E} \log \Bigg(2\pi e \Bigg(\sum_{m=1}^M h_{km}^2\text{Var}\left[X_m[t]\right] +  \sum_{m\neq \tilde{m}}h_{km}h_{k\tilde{m}}\text{Cov}(X_m[t],X_{\tilde{m}}[t]) + 1 \Bigg) \Bigg)\nonumber\\
& \myleq{(c)} \sum_{t=1}^{T_E} \log \Bigg(2\pi e \Bigg(\sum_{m=1}^M h_{km}^2\text{Var}\left[X_m[t]\right]  +  \sum_{m\neq  \tilde{m}}h_{km}h_{k\tilde{m}}\sqrt{\text{Var}[X_m[t]]\text{Var}[X_{\tilde{m}}[t]]} + 1 \Bigg)\Bigg)\nonumber\\
& \myleq{(d)} \sum_{t=1}^{T_E} \log \left(2\pi e \left(\sum_{m=1}^M h_{km}^2 P  +  \sum_{m\neq  \tilde{m}}h_{km}h_{k\tilde{m}} P + 1 \right)\right)\nonumber \\
& = \sum_{t=1}^{T_E} \log \Big(2\pi e \big( {\Lambda} P + 1 \big)\Big)  = T_E\log \Big(2\pi e \big( {\Lambda} P + 1 \big)\Big)
\end{align}

\n where ${\Lambda} = \max_{k\in [1:\ell]}\left[\sum_{m=1}^M h_{km}^2 + \sum_{m\neq  \tilde{m}}h_{km}h_{k\tilde{m}}\right]$. The steps in \eqref{eq:sub1} as explained as follows:
\begin{itemize}
\item Step \textsf{(a)} follows from the fact that noise is i.i.d. and uncorrelated with the input symbols.
\item Step \textsf{(b)} follows from the fact that $\text{Var}\left[n_k[t]\right]=1$.
\item Step \textsf{(c)} follows from the Cauchy-Schwartz Inequality.
\item Step \textsf{(d)} follows from the average power constraint $P$ on the input symbols.
\end{itemize}
Substituting (\ref{eq:sub1}) into (\ref{eq:mi_ub1}), we have
\begin{align}
h\left(\mathbf{Y}^{T_E}_{[1:\ell]}\right)
&~\leq~  \sum_{k=1}^\ell T_E \log  \Big(2\pi e \big( {\Lambda} P + 1 \big)\Big)  ~=~\ell{T_E}\log  \Big(2\pi e \big( {\Lambda} P + 1 \big)\Big) , \label{eq:mi_ub}
\end{align}
which completes the proof of the Lemma \ref{lem:1}.
\end{proof} \vspace{5pt}

\n Next, we state and prove Lemma \ref{lem:termB} which was used to bound the second term in \eqref{eq:p1} in Appendix \ref{ap:thglb}.
\begin{lemma}\label{lem:termB}
For the cloud and cache-aided wireless network under consideration, for any feasible policy $\pi = (\pi_f,\pi_c,\pi_e,\pi_d)$, the entropy of the $K$ requested files $F_{[1:K]}$, conditioned on the channel outputs $\mathbf{Y}^{T_E}_{[1:\ell]}$, on any $(M-\ell)$ fronthaul transmissions $\mathbf{U}^{T_F}_{[1:(M-\ell)]}$ with corresponding cache contents $S_{[1:(M-\ell)]}$ and on the remaining files $F_{[K+1:M]}$, can be upper bounded as
\begin{align}
&H\left(F_{[1:K]}|\mathbf{Y}^{T_E}_{[1:\ell]},\mathbf{U}^{T_F}_{[1:(M-\ell)]},S_{[1:(M-\ell)]},F_{[K+1:N]} \right) \leq~~ L\epsilon_L + T_E\log\det \left( \mathbf{I}_{\Scale[0.7]{[K-\ell]}} + \tilde{\mathbf{H}}\tilde{\mathbf{H}}^H\right),
\end{align}
where $\epsilon_L$ is a function of the probability of error $P_e$ that vanishes as $L\rightarrow \infty$, the matrix $\tilde{\mathbf{H}}$ is a function solely of the channel matrix $\mathbf{H}$ and $\mathbf{I}_{\Scale[0.7]{[K-\ell]}}$ is a $(K-\ell)\times(K-\ell)$ identity matrix.
\end{lemma}
\begin{proof}
In order to prove this lemma, we first consider the following set of inequalities:
\begin{align}\label{eq:B}
						 &H\left(F_{[1:K]}|\mathbf{Y}^{T_E}_{[1:\ell]},\mathbf{U}^{T_F}_{[1:(M-\ell)]},S_{[1:(M-\ell)]},F_{[K+1:N]} \right) \nonumber\\
						 &\myeq{(a)}H\left(F_{[1:K]}|\mathbf{Y}^{T_E}_{[1:\ell]},\mathbf{U}^{T_F}_{[1:(M-\ell)]},S_{[1:(M-\ell)]},\mathbf{X}^{T_E}_{[1:(M-\ell)]},F_{[K+1:N]} \right) \nonumber\\
						 &\myleq{(b)} H\left(F_{[1:K]}|\mathbf{Y}^{T_E}_{[1:\ell]},\mathbf{X}^{T_E}_{[1:(M-\ell)]},F_{[K+1:N]} \right)\nonumber\\
						 &\myleq{(c)} H\left(F_{[1:\ell]}|\mathbf{Y}^{T_E}_{[1:\ell]}\right)	+ H\left(F_{[\ell+1:K]}|\mathbf{Y}^{T_E}_{[1:\ell]},\mathbf{X}^{T_E}_{[1:(M-\ell)]},F_{[1:\ell]},F_{[K+1:N]} \right)\nonumber\\
						 &\myleq{(d)} L\epsilon_L + H\left(F_{[\ell+1:K]}|\mathbf{Y}^{T_E}_{[1:\ell]},\mathbf{X}^{T_E}_{[1:(M-\ell)]},F_{[1:\ell]\cup[K+1:N]} \right),
\end{align}
\n where the steps in \eqref{eq:B} are explained as follows:
\begin{itemize}
\item Step \textsf{(a)} follows from the fact that the channel inputs $\mathbf{X}^{T_E}_{[1:(M-\ell)]}$ are functions of the fronthaul transmissions $\mathbf{U}^{T_F}_{[1:(M-\ell)]}$ and the corresponding cache contents $S_{[1:(M-\ell)]}$.
\item Step \textsf{(b)} follows from the fact that conditioning reduces entropy.
\item Step \textsf{(c)} follows from the chain rule of entropy and from the fact that conditioning reduces entropy.
\item In step $\textsf{(d)}$, we use Fano's inequality on the first term where $\epsilon_L$ is a function, independent of $P$, that vanishes as $L\rightarrow \infty$.
\end{itemize}
Next, we consider the second term in \eqref{eq:B}. We have
\begin{align}\label{eq:c_prel}
& H\left(F_{[\ell+1:K]}|\mathbf{Y}^{T_E}_{[1:\ell]},\mathbf{X}^{T_E}_{[1:(M-\ell)]},F_{[1:\ell]\cup[K+1:N]} \right)\nonumber\\
&\myeq{(a)} H\left(F_{[\ell+1:K]}|\mathbf{Y}^{T_E}_{[1:\ell]},\mathbf{X}^{T_E}_{[1:(M-\ell)]},\mathbf{n}^{T_E}_{[\ell+1:K]},F_{[1:\ell]\cup[K+1:N]} \right)\nonumber\\
&\myleq{(b)} H\Big(F_{[\ell+1:K]}| \mathbf{Y}^{T_E}_{[\ell+1:K]} + \tilde{\mathbf{n}}^{T_E}_{[\ell+1:K]},\mathbf{Y}^{T_E}_{[1:\ell]},F_{[1:\ell]\cup[K+1:N]} \Big)\nonumber\\
& \myleq{(c)} H\Big(F_{[\ell+1:K]}| \mathbf{Y}^{T_E}_{[\ell+1:K]} + \tilde{\mathbf{n}}^{T_E}_{[\ell+1:K]},F_{[1:\ell]\cup[K+1:N]} \Big) - H\Big(F_{[\ell+1:K]}|\mathbf{Y}^{T_E}_{[\ell+1:K]},F_{[1:\ell]\cup[K+1:N]} \Big) \nonumber\\
&~~~~~~~~~~~~~~~~~~~~+  H\Big(F_{[\ell+1:K]}|\mathbf{Y}^{T_E}_{[\ell+1:K]},F_{[1:\ell]\cup[K+1:N]} \Big)\nonumber\\
						&\myleq{(d)}  H\Big(F_{[\ell+1:K]}| \mathbf{Y}^{T_E}_{[\ell+1:K]} + \tilde{\mathbf{n}}^{T_E}_{[\ell+1:K]},F_{[1:\ell]\cup[K+1:N]} \Big) - H\Big(F_{[\ell+1:K]}|\mathbf{Y}^{T_E}_{[\ell+1:K]},F_{[1:\ell]\cup[K+1:N]} \Big) + L\epsilon_L
\end{align}
where the steps in \eqref{eq:c_prel} are explained as follows:
\begin{itemize}
\item Step \textsf{(a)} follows from the fact that the noise term $\mathbf{n}^{T_E}_{[\ell+1:K]}$ is independent of all the other random variables in the entropy term and can be introduced into the conditioning.
\item In Step \textsf{(b)}, we use Lemma \ref{lem:2} stated in Appendix \ref{ap:lemma} and the fact that conditioning reduces entropy. We observe that $\mathbf{n}^{T_E}_{[\ell+1:K]} \rightarrow (\mathbf{Y}^{T_E}_{[1:\ell]},\mathbf{X}^{T_E}_{[1:(M-\ell)]},F_{[1:\ell]\cup[K+1:N]}) \rightarrow F_{[\ell+1:K]}$ forms a Markov chain and as a result, the data-processing inequality \cite{cover} applies. The additive noise term $\tilde{\mathbf{n}}^{T_E}_{[\ell+1:K]}$ is defined as
\begin{align*}
\tilde{\mathbf{n}}^{T_E}_{[\ell+1:K]} = \left(\mathbf{H}_2 \cdot {\mathbf{H}_1}^{\dagger}\right)\mathbf{n}^{T_E}_{[1:\ell]},
\end{align*}
which is a $[K-\ell]\times {T_E}$ matrix, where each column is an independent Gaussian random vector distributed as $\mathcal{N}\left(0,\tilde{\mathbf{H}}\tilde{\mathbf{H}}^H\right)$ with $\tilde{\mathbf{H}} = \left(\mathbf{H}_2 \cdot {\mathbf{H}_1}^{\dagger}\right)$, where the matrices $\mathbf{H}_1$ and $\mathbf{H}_2$ are sub-matrices of the channel matrix $\mathbf{H}$ and are defined in Lemma \ref{lem:2} (see \eqref{eq:h1h2}), and $\mathbf{H_1}^{\dagger}$ is the Moore-Penrose pseudo-inverse. We note here that the noise term $\tilde{\mathbf{n}}^{T_E}_{[\ell+1:K]}$ is independent of channel inputs $\mathbf{X}^{T_E}_{[1:M]}$ and noise terms $\mathbf{n}^{T_E}_{[\ell+1:K]}$.
\item Step \textsf{(c)} follows from the fact that conditioning reduces entropy.
\item Step \textsf{(d)} follows from applying Fano's inequality to the last entropy term in the previous step, where $\epsilon_L$ is again, a function independent of $P$ that vanishes as $L\rightarrow\infty$.
\end{itemize}
Now, from \eqref{eq:c_prel}, considering the first and second entropy terms together we have:
\begin{align}\label{eq:c12}
& H\Big(F_{[\ell+1:K]}| \mathbf{Y}^{T_E}_{[\ell+1:K]} + \tilde{\mathbf{n}}^{T_E}_{[\ell+1:K]},F_{[1:\ell]\cup[K+1:N]} \Big) - H\Big(F_{[\ell+1:K]}|\mathbf{Y}^{T_E}_{[\ell+1:K]},F_{[1:\ell]\cup[K+1:N]} \Big) \nonumber\\
& = I\left(F_{[\ell+1:K]};\mathbf{Y}^{T_E}_{[\ell+1:K]} |F_{[1:\ell]\cup[K+1:N]} \right) - I\left(F_{[\ell+1:K]}; \mathbf{Y}^{T_E}_{[\ell+1:K]} + \tilde{\mathbf{n}}^{T_E}_{[\ell+1:K]}|F_{[1:\ell]\cup[K+1:N]} \right)\nonumber\\
& = h\left( \mathbf{Y}^{T_E}_{[\ell+1:K]} + \tilde{\mathbf{n}}^{T_E}_{[\ell+1:K]}|F_{[1:N]} \right) - h\left( \mathbf{Y}^{T_E}_{[\ell+1:K]} + \tilde{\mathbf{n}}^{T_E}_{[\ell+1:K]}|F_{[1:\ell]\cup[K+1:N]} \right)\nonumber\\
&~~~~~~~ + h\left(\mathbf{Y}^{T_E}_{[\ell+1:K]} |F_{[1:\ell]\cup[K+1:N]} \right) - h\left(\mathbf{Y}^{T_E}_{[\ell+1:K]}|F_{[1:N]} \right)\nonumber\\
& \myleq{(a)} h\left( \mathbf{Y}^{T_E}_{[\ell+1:K]} + \tilde{\mathbf{n}}^{T_E}_{[\ell+1:K]}|F_{[1:N]} \right) - h\left( \mathbf{Y}^{T_E}_{[\ell+1:K]} + \tilde{\mathbf{n}}^{T_E}_{[\ell+1:K]}|\tilde{\mathbf{n}}^{T_E}_{[\ell+1:K]},F_{[1:\ell]\cup[K+1:N]} \right)\nonumber\\
&~~~~~~~ + h\left(\mathbf{Y}^{T_E}_{[\ell+1:K]} |F_{[1:\ell]\cup[K+1:N]} \right) - h\left(\mathbf{Y}^{T_E}_{[\ell+1:K]}|F_{[1:N]} \right)\nonumber\\
& = h\left( \mathbf{Y}^{T_E}_{[\ell+1:K]} + \tilde{\mathbf{n}}^{T_E}_{[\ell+1:K]}|F_{[1:N]} \right) - h\left(\mathbf{Y}^{T_E}_{[\ell+1:K]}|F_{[1:N]} \right)\nonumber\\
&\myeq{(b)}  h\Big(\mathbf{n}^{T_E}_{[\ell+1:K]} + \tilde{\mathbf{n}}^{T_E}_{[\ell+1:K]} \Big) - h\Big(\mathbf{n}^{T_E}_{[\ell+1:K]}\Big)\nonumber\\
&\myeq{(c)} {T_E}\log \Big((2\pi e)^{K-\ell} \left| \mathbf{I}_{\Scale[0.7]{[K-\ell]}} + \tilde{\mathbf{H}}\tilde{\mathbf{H}}^H\right|\Big) - {T_E}\log \Big((2\pi e)^{K-\ell}\Big) \nonumber\\
&= {T_E}\log\det \left( \mathbf{I}_{\Scale[0.7]{[K-\ell]}} + \tilde{\mathbf{H}}\tilde{\mathbf{H}}^H\right).
\end{align}
\n The steps in \eqref{eq:c12} are explained as follows:
\begin{itemize}
\item Step \textsf{(a)} follows from the fact that conditioning reduces entropy.
\item Step \textsf{(b)} follows from the fact that, given all the files $F_{[1:N]}$, the channel outputs are functions of the channel noise.
\item Step \textsf{(c)} follows from the fact that the noise terms are jointly Gaussian and are i.i.d. across time ${T_E}$. The function $|\cdot|$ is the determinant.
\end{itemize}
Thus, using \eqref{eq:c_prel} and \eqref{eq:c12} in \eqref{eq:B}, we have
\begin{align}\label{eq:C}
&H\left(F_{[1:K]}|\mathbf{Y}^{T_E}_{[1:\ell]},\mathbf{U}^{T_F}_{[1:(M-\ell)]},S_{[1:(M-\ell)]},F_{[K+1:N]} \right) ~\leq~ L\epsilon_L + {T_E}\log\det \left( \mathbf{I}_{\Scale[0.7]{[K-\ell]}} + \tilde{\mathbf{H}}\tilde{\mathbf{H}}^H\right),
\end{align}
which completes the proof of the Lemma \ref{lem:termB}.
\end{proof}
Finally, we state and prove Lemma \ref{lem:2} which was used in \eqref{eq:c_prel} for the proof of Lemma \ref{lem:termB}.
\begin{lemma}\label{lem:2}
Given any $\ell \in [1:\min\{M,K\}]$, there exists a (deterministic) function of the channel outputs $\mathbf{Y}^{T_E}_{[1:\ell]}$, input symbols $\mathbf{X}^{T_E}_{[1:(M-\ell)]}$ and channel noise $\mathbf{n}^{T_E}_{[\ell+1:K]}$, that yields
\begin{align}\label{eq:lem2}
\mathbf{Y}^{T_E}_{[\ell+1:K]} + \tilde{\mathbf{n}}^{T_E}_{[\ell+1:K]},
\end{align}
where we have defined $\tilde{\mathbf{n}}^{T_E}_{[\ell+1:K]} = \left(\mathbf{H}_2 \cdot {\mathbf{H}_1}^{\dagger}\right)\mathbf{n}^{T_E}_{[1:\ell]}$ and $\mathbf{H_1}^{\dagger}$ is the Moore-Penrose pseudo-inverse. The matrices $\mathbf{H}_1$ and $\mathbf{H}_2$ are sub-matrices of the channel matrix $\mathbf{H}$ and are defined as
\begin{align}\label{eq:h1h2}
&\mathbf{H}_1 = \mathbf{H}_{[1:\ell]}^{[(M-\ell) +1:M]};~~~ \mathbf{H}_2 = \mathbf{H}_{[\ell+1:K]}^{[(M-\ell) +1:M]}.
\end{align}
\end{lemma}
\begin{proof}
Given any $\ell \in [1:\min\{M,K\}]$, from \eqref{eq:chop2}, the channel outputs $\mathbf{Y}^{T_E}_{[1:\ell]}$ are a function of the $M$ input symbols $\mathbf{X}^{T_E}_{[1:M]}$ and of the noise $\mathbf{n}^{T_E}_{[1:\ell]}$. Given the input symbols $\mathbf{X}^{T_E}_{[1:(M-\ell)]}$, we can cancel the contribution of these input symbols from the channel outputs $\mathbf{Y}^{T_E}_{[1:\ell]}$ to obtain
\begin{align}\label{eq:ty}
\tilde{\mathbf{Y}}^{T_E}_{[1:\ell]} &= \mathbf{H}^{[1:M]}_{[1:\ell]}
 \mathbf{X}_{[1:M]}^{T_E}  +  \mathbf{n}_{[1:\ell]}^{T_E}  - \mathbf{H}^{[1:M]}_{[1:\ell]}
\begin{bmatrix} \mathbf{X}_{[1:(M-\ell)]}^{T_E} \\ \mathbf{0}^{T_E}_{[(M-\ell) +1:M]} \end{bmatrix}\nonumber\\[5pt]
&  = \mathbf{H}_1 \begin{bmatrix} \mathbf{X}_{[(M-\ell) + 1:M]}^{T_E} \end{bmatrix} + \begin{bmatrix} \mathbf{n}_{[1:\ell]}^{T_E}\end{bmatrix},
\end{align}
where $\mathbf{0}^{T_E}_{[(M-\ell) +1:M]}$ is an $\ell \times {T_E}$ matrix of zeros. As a result, multiplying both sides of \eqref{eq:ty} by $\mathbf{H_1}^{\dagger}$, we get
\begin{align}\label{eq:X}
\mathbf{H_1}^{\dagger}\tilde{\mathbf{Y}}^{T_E}_{[1:\ell]}  =  \mathbf{X}_{[(M-\ell) + 1:M]}^{T_E} + \mathbf{H_1}^{\dagger}\mathbf{n}_{[1:\ell]}^{T_E}.
\end{align}
Now let
\begin{align}
\mathbf{H}_3 = \mathbf{H}^{[1:M]}_{[\ell+1:K]}.
\end{align}
Using this definition, we have
\begin{align}\label{eq:last}
&\mathbf{Y}_{[\ell + 1:K]}^{T_E}
= \mathbf{H}_3 \mathbf{X}_{[1:M]}^{T_E} + \mathbf{n}_{[\ell+1:K]}^{T_E}\nonumber\\[5pt]
& = \mathbf{H}_3
\begin{bmatrix}\mathbf{X}_{[1:(M-\ell)]}^{T_E}\\ \mathbf{H_1}^{\dagger}\tilde{\mathbf{Y}}^{T_E}_{[1:\ell]} - \mathbf{H_1}^{\dagger}\mathbf{n}^{T_E}_{[1:\ell]} \end{bmatrix}
+ \mathbf{n}_{[\ell+1:K]}^{T_E} \nonumber\\[5pt]
& \myeq{(a)} \mathbf{H}_3 \begin{bmatrix}\mathbf{X}_{[1:(M-\ell)]}^{T_E}\\ \mathbf{H_1}^{\dagger}\tilde{\mathbf{Y}}^{T_E}_{[1:\ell]} \end{bmatrix}
-  \mathbf{H}_3 \begin{bmatrix} \mathbf{0}^{T_E}_{[1:(M-\ell)]}\\ \mathbf{H_1}^{\dagger}\mathbf{n}^{T_E}_{[1:\ell]}\end{bmatrix}
+ \mathbf{n}_{[\ell+1:K]}^{T_E} \nonumber\\[5pt]
& = \mathbf{H}_3 \begin{bmatrix}\mathbf{X}_{[1:(M-\ell)]}^{T_E}\\ \mathbf{H_1}^{\dagger}\tilde{\mathbf{Y}}^{T_E}_{[1:\ell]} \end{bmatrix}
-  \mathbf{H}_2 \begin{bmatrix} \mathbf{H_1}^{\dagger}\mathbf{n}^{T_E}_{[1:\ell]}\end{bmatrix}
+ \mathbf{n}_{[\ell+1:K]}^{T_E},
\end{align}
where, in \textsf{(a)}, $\mathbf{0}^{T_E}_{[1:(M-\ell)]}$ is a $[(M-\ell)]\times {T_E}$ matrix of zeros. Rearranging \eqref{eq:last}, we obtain
\begin{align}
\mathbf{Y}^{T_E}_{[\ell+1:K]} + \tilde{\mathbf{n}}^{T_E}_{[\ell+1:K]}  &=  \mathbf{H}_3 \begin{bmatrix}\mathbf{X}_{[1:(M-\ell)]}^{T_E}\\ \mathbf{H_1}^{\dagger}\tilde{\mathbf{Y}}^{T_E}_{[1:\ell]} \end{bmatrix} + \mathbf{n}_{[\ell+1:K]}^{T_E},
\end{align}
where the RHS is a function of the $\ell$ channel outputs $\mathbf{Y}^{T_E}_{[1:\ell]}$, input symbols $\mathbf{X}^{T_E}_{[1:(M-\ell)]}$ and channel noise $\mathbf{n}^{T_E}_{[\ell+1:K]}$.
This completes the proof Lemma \ref{lem:2}. Note that we assumed in \eqref{eq:X} that the sub-matrix $\mathbf{H}_1$ is invertible, which is true for almost all channel realizations, i.e., it is true with probability $1$.
\end{proof}


\section{Pipelined Fronthaul-Edge Transmission}\label{ap:pipe}

\subsection{Proof of Proposition \ref{prop:ach_pipe}}\label{ap_ssec:ach_pipe}
We adopt block-Markov coding as explained in Section \ref{ssec:pipe_ub}, whereby, in each block, the constituent policies for fronthaul and edge transmission are obtained by file-splitting between two policies. To elaborate, for some $\alpha \in [0,1]$ fraction of each file, a (serial) policy requiring fronthaul and edge NDTs $\delta_F^{(1)}$ and $\delta_E^{(1)}$ is used, and for the remaining $(1-\alpha)$ fraction of each file, a (serial) policy requiring NDTs $\delta_F^{(2)}$ and $\delta_E^{(2)}$ is used. From \eqref{eq:ach_pipeline}, the achievable NDT with the resulting block-Markov policy is given by
\begin{align}\label{eq:pipe_ts}
\delta_{\mathsf{P,Ach}} = \max\left(\alpha\delta_F^{(1)} + (1-\alpha)\delta_F^{(2)}, \alpha\delta_E^{(1)} + (1-\alpha)\delta_E^{(2)}\right).
\end{align}
We next identify the constituent policies used to prove \eqref{eq:ach_pipe} for the three regimes namely $(i)$ \textit{low cache regime} with $\mu\in[0,\mu_1]$; $(ii)$ \textit{intermediate cache regime} with $\mu\in[\mu_1,\mu_2]$; and \textit{high cache regime} with $\mu\in[\mu_2,1]$. From the definition of $\mu_1, \mu_2$ in \eqref{eq:mu12}, note that we have $\mu_1\leq \mu_2$ when $\mu_1=0$, we have $r = K/\max\{M,K\}$ and hence $\mu_2=0$. Next we study each regime separately.

\subsubsection{Low Cache Regime $\mathit{(\mu\in[0,\mu_1])}$} 
For the regime $\mu\in[0,\mu_1]$, we set $\alpha = \mu M$, and perform file-splitting between cache-aided EN coordination, as described in Lemma \ref{lem:ach1M}, which is characterized by the NDTs
\begin{align}\label{eq:pipe_lc1}
\delta_F^{(1)} = 0;~~~~~\delta_E^{(1)} =\delta_{\mathsf{Ca-IA}} = \frac{M+K-1}{M};
\end{align}
and soft-transfer fronthauling in conjunction with EN cooperation via ZF-beamforming, as described in Proposition \ref{prop:stf}, which yields the NDTs
\begin{align}\label{eq:pipe_lc2}
\delta_F^{(2)} = \frac{K}{Mr};~~~~~\delta_E^{(2)} =\frac{K}{\min\{M,K\}}.
\end{align}
Note that we have $\mu_1\leq 1/M$ and hence $\alpha \leq 1$ for $\mu\leq \mu_1$. Substituting \eqref{eq:pipe_lc1}-\eqref{eq:pipe_lc2} in \eqref{eq:pipe_ts}, we obtain 
\begin{align}\label{eq:pipe_ia}
\delta_{\mathsf{P-IA}} &= \max\left\{ \frac{(1-\mu M)K}{Mr}~,~ \frac{(1-\mu M)K}{\min\{M,K\}} +  \mu(M+K-1)\right\},\nonumber\\
											 &= \frac{(1-\mu M)K}{Mr}, ~~~~~~~~~~~\text{for}~~~\mu \leq \mu_1.
\end{align}

\subsubsection{High Cache Regime $\mathit{(\mu\in [\mu_2,1])}$} 
For the regime $\mu\in[\mu_2,1]$, we set $\alpha = \mu$, and perform file-splitting between cache-aided EN cooperation in the form of ZF-beamforming, as described in Lemma \ref{lem:ach1}, which yields the NDTs
\begin{align}\label{eq:pipe_hc1}
\delta_F^{(1)} = 0;~~~~~\delta_E^{(1)} =\delta_{\mathsf{Ca-ZF}} = \frac{K}{\min\{M,K\}};
\end{align}
and soft-transfer fronthauling with ZF-beamforming on the edge, as described in Proposition \ref{prop:stf} leading to 
\begin{align}\label{eq:pipe_hc2}
\delta_F^{(2)} = \frac{K}{Mr};~~~~~\delta_E^{(2)} = \frac{K}{\min\{M,K\}}.
\end{align}
 Substituting \eqref{eq:pipe_hc1}-\eqref{eq:pipe_hc2} in \eqref{eq:pipe_ts}, we obtain 
\begin{align}\label{eq:pipe_zf}
\delta_{\mathsf{P-ZF}} &= \max\left\{\frac{(1-\mu)K}{Mr}, \frac{K}{\min\{M,K\}} \right\},\nonumber\\
											 &= \frac{K}{\min\{M,K\}}, \hspace{60pt} \text{for}~~~\mu \geq \mu_2.
\end{align}

\subsubsection{Intermediate Cache Regime $\mathit{(\mu\in [\mu_1,\mu_2])}$} 
For the intermediate cache regime of $\mu\in [\mu_1,\mu_2]$, we consider a strategy which performs file-splitting between the schemes achieving $\delta_{\mathsf{P-IA}}$ at $\mu=\mu_1$ and $\delta_{\mathsf{P-ZF}}$ at $\mu=\mu_2$ discussed above. Specifically, using first scheme which yields an NDT
\begin{align}
\delta_{\mathsf{P-IA}} = \frac{(1-\mu_1 M)K}{Mr}
\end{align}
for a fraction $\left(\frac{\mu_2 - \mu}{\mu_2 - \mu_1}\right)^+$ of the files, and the second scheme for the remaining fraction, which yields an NDT
\begin{align}
\delta_{\mathsf{P-ZF}} = \frac{(1-\mu_2)K}{Mr},
\end{align}
we obtain the achievable NDT
\begin{align}\label{eq:ach_pipe_ts}
\delta_{\mathsf{P-FS}} &= \left(\frac{\mu_2 - \mu}{\mu_2 - \mu_1}\right)^+  \frac{(1-\mu_1 M)K}{Mr} + \left(1 - \left(\frac{\mu_2 - \mu}{\mu_2 - \mu_1}\right)^+\right) \frac{(1-\mu_2)K}{Mr}\nonumber\\
											 &= \frac{K}{Mr} \left[1 - \mu_2 - \left[\mu_1 M - \mu_2\right]\left(\frac{\mu_2-\mu}{\mu_2 - \mu_1} \right)^+ \right].
\end{align}
This concludes the proof of Proposition \ref{prop:ach_pipe}.

\subsection{Proof of Proposition \ref{prop:pipe_gen}}\label{ap_ssec:pipe_gen}

In the regime of low cache size $\mu\in[0,\mu_1]$, the upper bound \eqref{eq:ach_pipe}, rewritten here as
\begin{align}
\delta_{\mathsf{P,Ach}}(\mu,r) \leq \delta_{\mathsf{P-IA}} = \frac{(1 - \mu M)K}{Mr},
\end{align}
matches the lower bound in Corollary \ref{cor:pipe_glb} by setting $\ell=0$, thereby characterizing the minimum NDT for the low cache regime with $\mu\in[0,\mu_1]$.

For the regime of high cache size $\mu\in[\mu_2,1]$, from \eqref{eq:ach_pipe} we have the upper bound 
\begin{align}\label{eq:gap_ub2}
\delta_{\mathsf{P,Ach}}(\mu,r) \leq \delta_{\mathsf{P-ZF}} = \frac{K}{\min\{M,K\}}.
\end{align}
For a matching lower bound, when $M\geq K$, from \eqref{eq:pipe_glb}, we have $\delta_{\mathsf{P}}^*(\mu,r)\geq 1$, while, for $M \leq K$, using $\ell=M$ in the first term inside the $\max(\cdot)$ function in \eqref{eq:pipe_glb} yields $\delta_{\mathsf{P}}^*(\mu,r)\geq K/M$. Combining the two bounds yields the following lower bound on the minimum NDT:
\begin{align}\label{eq:gap_lb2}
\delta_{\mathsf{P}}^*(\mu,r)\geq \frac{K}{\min\{M,K\}},
\end{align}
which matches the upper bound \eqref{eq:gap_ub2}, thereby characterizing the minimum NDT for the high cache regime with $\mu\in[\mu_2,1]$.

Finally, we consider the high fronthaul regime i.e., $r \geq {((1-\mu)\min\{M,K\})}/{M}$. In this regime, considering the NDT in \eqref{eq:pipe_zf}, which is achieved by file-splitting between cloud-aided soft-transfer fronthauling and cache-aided ZF beamforming, it can be seen that the second term inside the $\max(\cdot)$ dominates and we have
\begin{align}\label{eq:ach_pipe_zf}
\delta_{\mathsf{P,Ach}}(\mu,r)\leq \delta_{\mathsf{P-ZF}} = \frac{K}{\min\{M,K\}},
\end{align}
which matches the lower bound in \eqref{eq:gap_lb2}. This completes the proof of Proposition \ref{prop:pipe_gen}.

\subsection{Proof of Proposition \ref{prop:gap_pipe}}\label{ap_ssec:gap_pipe}
In this section, we present the proof of the approximate optimality of the achievable schemes presented in Section \ref{ssec:pipe_ub} in the regime of intermediate fractional cache sizes with $\mu \in [\mu_1,\mu_2]$. To this end, we consider two sub-regimes for the fractional cache size $\mu$ namely  $(i)$ the \textit{intermediate cache regime $\mathit{1}$} with $\mu \in [\mu_1, 1/M]$; and $(ii)$ the \textit{intermediate cache regime $\mathit{2}$} with $\mu \in [1/M, \mu_2]$. We consider each of the two regimes separately.

\n \textbf{\textit{Intermediate Cache Regime} $\bm{\mathit{1}\left(\mathit{\mu \in [\mu_1, 1/M]}\right)}$:} For this regime, considering the achievable NDT presented in \eqref{eq:pipe_ia}, we have the upper bound
\begin{align}\label{eq:gap_ach1}
\delta_{\mathsf{P,Ach}}(\mu,r) & \leq \max\left\{ \frac{(1-\mu M)K}{Mr}~,~ \frac{(1-\mu M)K}{\min\{M,K\}} +  \mu(M+K-1)\right\}\nonumber\\
																												 &= \frac{(1-\mu M)K}{\min\{M,K\}} +  \mu(M+K-1),
\end{align}
since the edge latency i.e., the second term inside the $\max(\cdot)$ function, dominates when $\mu \geq \mu_1$ (see Appendix \ref{ap_ssec:ach_pipe} for details). A lower bound is given by \eqref{eq:gap_lb2}. Using the mentioned upper and lower bounds on the minimum NDT, we have
\begin{align}\label{eq:pgap2}
\frac{\delta_{\mathsf{P,Ach}}(\mu,r)}{\delta_{\mathsf{P}}^*(\mu,r)}  &\leq \left[\frac{(1-\mu M)K}{\min\{M,K\}} +  \mu(M+K-1)\right] \times \frac{\min\{M,K\}}{K}\nonumber\\
																							 & = (1 - \mu M) + \mu M \left[\frac{\min\{M,K\}(M+K-1)}{KM}\right]\nonumber\\		
																							 & \leq (1 - \mu M) + \mu M \left[\frac{M+K}{\max\{M,K\}}\right] ~= (1 - \mu M) + \mu M \left[1+\frac{\min\{M,K\}}{\max\{M,K\}}\right] \nonumber\\
																							 &\leq (1 - \mu M) + 2\mu M \leq 1 + \mu M \nonumber\\
																							 &\myleq{(a)} 2,	
\end{align}
where step \textsf{(a)} follows by using $\mu\leq 1/M$.

\n \textbf{\textit{Intermediate Cache Regime} $\bm{\mathit{2}\left(\mathit{\mu \in [1/M,\mu_2]}\right)}$:}
For this regime, considering the achievable NDT presented in Proposition \ref{prop:ach_pipe}, we have the upper bound
\begin{align}
\delta_{\mathsf{P,Ach}}(\mu,r) \myleq{(a)} \delta_{\mathsf{P,Ach}}(\mu_1,r) = \frac{(1 - \mu_1 M)K}{Mr}\myleq{(b)} \frac{M+K-1}{M},
\end{align}
for any $M,K\geq 1$ and $r>0$, and where steps \textsf{(a)} and \textsf{(b)} follow from the fact that the NDT is a non-increasing function of the cache size $\mu$ and that an NDT of $\delta_{\mathsf{Ca-IA}} = (M+K-1)/M$ is achievable using cache-aided EN coordination via interference alignment at $\mu=1/M$. Again, for this regime, considering the lower bound in \eqref{eq:gap_lb2}, we have
\begin{align}\label{eq:pgap3}
\frac{\delta_{\mathsf{P,Ach}}(\mu,r)}{\delta_{\mathsf{P}}^*(\mu,r)} &\leq \frac{M+K-1}{M} \times \frac{\min\{M,K\}}{K} \nonumber\\
																							 & \leq \frac{M+K}{\max\{M,K\}} ~= 1 + \frac{\min\{M,K\}}{\max\{M,K\}} \leq 2.
\end{align}
Finally combining \eqref{eq:pgap2} and \eqref{eq:pgap3} concludes the proof of Proposition \ref{prop:gap_pipe}.

\subsection{Proof of Corollary \ref{cor:mk22_pipe}}\label{ap_ssec:mk22_pipe}
Using $M=K=2$ in \eqref{eq:pipe_gen1}, we obtain the minimum NDT 
\begin{align}
\delta_{\mathsf{P}}^*(\mu,r)= \begin{cases}
								 ~\dfrac{1-2\mu}{r},~~~~~~ \text{for}~~~~ \mu\in [0,\mu_1=(1-r)/(2+r)] \\ 
								 ~1,~~~~~~~~~~~~~ \text{for}~~~~ \mu\in[\mu_2=(1-r),1]. 
								 \end{cases}
\end{align}

For the remaining intermediate cache regime with $\mu\in[\mu_1,\mu_2]$, we adopt the achievable NDT $\delta_{\mathsf{P-FS}}$ given in \eqref{eq:ach_pipe}, which yields the upper bound
\begin{align}\label{eq:mk22_ach}
\delta_{\mathsf{P,Ach}}(\mu,r) \leq \frac{2-\mu}{1+r}, ~~~~~~ \text{for}~~~~ \mu\in[\mu_1,\mu_2],
\end{align}
 and the lower bound in Corollary \ref{cor:pipe_glb} with $\ell=1$ which can be seen to match \eqref{eq:mk22_ach}. In the high fronthaul regime, i.e., $r\geq 1$, using $M=K=2$ in \eqref{eq:pipe_gen2} yields the minimum NDT. This concludes the proof.

\section*{Acknowledgements}
The authors would like to thank Seyyed Mohammadreza Azimi for his helpful inputs on the proof in Appendix \ref{ap:gap} and Roy Karasik for helping improve the quality of presentation of our work. 
\bibliographystyle{IEEEtran}
\bibliography{references}
\end{document}